\pdfoutput=1
\documentclass[acmsmall,screen,nonacm]{acmart}
\usepackage[utf8]{inputenc}
\usepackage{tikz}

\usepackage[bibliography=common]{apxproof}
\usepackage{array,booktabs}
\usepackage{relsize}
\usetikzlibrary{decorations.markings}
\usetikzlibrary{decorations.pathreplacing,calc}
\usepackage[shortlabels]{enumitem}
\usepackage{stackengine}
\usepackage{nicefrac}
\stackMath

\AtBeginDocument{%
  \providecommand\BibTeX{{%
    \normalfont B\kern-0.5em{\scshape i\kern-0.25em b}\kern-0.8em\TeX}}}

\acmJournal{TOCL}
\acmVolume{37}
\acmNumber{4}
\acmArticle{111}
\acmMonth{8}

\newtheoremrep{theorem}{Theorem}[section]
\newtheoremrep{proposition}[theorem]{Proposition}
\newtheoremrep{lemma}[theorem]{Lemma}
\newtheoremrep{example}[theorem]{Example}
\newtheoremrep{claim}[theorem]{Claim}
\newtheoremrep{observation}[theorem]{Observation}

\definecolor{mygreen}{rgb}{0.09, 0.45, 0.27}

\renewcommand{\phi}{\varphi}
\renewcommand{\epsilon}{\varepsilon}
\renewcommand{\leq}{\leqslant}
\renewcommand{\geq}{\geqslant}
\newcommand\defeq{\stackrel{\mathrm{def}}{=}}

\newcommand{\dom}{\mathsf{dom}}

\newcommand{\consts}{\mathsf{Consts}}
\newcommand{\nulls}{\mathsf{Nulls}}
\newcommand{\N}{\mathbb{N}}
\newcommand{\arity}{\mathsf{arity}}
\renewcommand{\Pr}{\mathsf{Pr}}
\newcommand{\surj}{\mathsf{surj}}
\newcommand{\sig}{\mathsf{sig}}

\newcommand{\countcompls}{\text{$\#$}\text{\sf Comp}}
\newcommand{\countvals}{\text{$\#$}\text{\sf Val}}
\newcommand{\ucountcompls}{\text{$\#$}\text{{\sf Comp}$^{\textit{u}}$}}
\newcommand{\ucountvals}{\text{$\#$}\text{{\sf Val}$^{\textit{u}}$}}
\newcommand{\ccountcompls}{\text{$\#$}\text{{\sf Comp}$_{\textit{Cd}}$}}
\newcommand{\ccountvals}{\text{$\#$}\text{{\sf Val}$_{\textit{Cd}}$}}
\newcommand{\cucountcompls}{\text{$\#$}\text{{\sf Comp}$^{\textit{u}}_{\textit{Cd}}$}}
\newcommand{\cucountvals}{\text{$\#$}\text{{\sf Val}$^{\textit{u}}_{\textit{Cd}}$}}

\newcommand{\hamsubgraphs}{\mathrm{\#HamSubgraphs}}
\newcommand{\shpf}{\mathrm{\#PF}}
\newcommand{\ksat}{\mathrm{\#k3SAT}}

\newcommand{\mc}{\mathrm{MC}}

\newcommand{\IS}{\text{\rm IS}}
\newcommand{\VC}{\mathrm{VC}}
\newcommand{\sm}{\mathrm{SM}}
\newcommand{\checkp}{\mathrm{check}}
\newcommand{\I}{\mathcal{I}}
\newcommand{\T}{\mathrm{T}}
\newcommand{\s}{\mathrm{s}}
\newcommand\mA{\mathbf{A}}
\newcommand\mZ{\mathbf{Z}}
\newcommand\mT{\mathbf{T}}
\newcommand\mD{\mathbf{D}}
\newcommand\mV{\mathbf{V}}
\newcommand\mC{\mathbf{C}}

\newcommand{\ff}{\mathrm{f}}
\newcommand{\params}{\mathsf{params}}

\newcommand{\sjfbcq}{\text{\rm sjfBCQ}}
\newcommand{\sjfbcqs}{\text{\rm sjfBCQs}}
\newcommand{\pr}{\leq_{\mathrm{par}}^{\mathrm{p}}}

\newcommand{\ptime}{\text{\rm P}}
\newcommand{\fp}{\text{\rm FP}}
\newcommand{\rk}{\mathrm{rk}}
\newcommand{\shp}{\text{$\#$}\text{\rm P}}
\newcommand{\spanl}{\text{\rm SpanL}}
\newcommand{\spanp}{\text{\rm SpanP}}
\newcommand{\nl}{\text{\rm NL}}
\newcommand{\np}{\text{\rm NP}}

\newcommand{\up}{\text{\rm UP}}
\newcommand{\rp}{\text{\rm RP}}
\newcommand{\bpp}{\text{\rm BPP}}
\newcommand{\spp}{\text{\rm SPP}}
\newcommand{\CC}{\mathcal{C}}
\newcommand{\sIS}{\text{$\#$}\text{\rm IS}}
\newcommand{\sVC}{\text{$\#$}\text{\rm VC}}
\newcommand{\tr}{\leq_{\mathrm{T}}^{\mathrm{p}}}
\newcommand{\acc}{\text{\rm accept}}
\newcommand{\rej}{\text{\rm reject}}
\newcommand{\gap}{\text{\rm gap}}
\newcommand{\gapp}{\text{\rm GapP}}
\newcommand{\gapspanp}{\text{\rm GapSpanP}}

\newcommand{\boundellipse}[3]
{(#1) ellipse (#2 and #3)
}

\begin{document}

\title[The complexity of counting problems over incomplete databases]{The Complexity of  
Counting Problems \\over Incomplete Databases}


\author{Marcelo Arenas}
\affiliation{\institution{Universidad Cat\'olica \& IMFD Chile}}
\additionalaffiliation{%
 Department of Computer Science \& Institute for Mathematical and Computational Engineering, School of Engineering, 
  Pontificia Universidad Cat\'olica de Chile
  }
\email{marenas@ing.puc.cl}

\author{Pablo Barceló}
\affiliation{\institution{Universidad Cat\'olica \& IMFD Chile}}
\additionalaffiliation{%
Institute for Mathematical and Computational Engineering, School of Engineering and Faculty of Mathematics, Pontificia 
  Universidad Cat\'olica de Chile
  }
\email{pbarcelo@uc.cl}

\author{Mikaël Monet}
\affiliation{%
  \institution{Univ. Lille, Inria, CNRS, Centrale Lille, UMR 9189 CRIStAL, F-59000 Lille, France}
  }
\email{mikael.monet@inria.fr}


\begin{abstract}
We study the complexity of various fundamental counting problems that arise in
the context of incomplete databases, i.e., relational databases that can
contain unknown values in the form of labeled nulls. Specifically, we assume
that the domains of these unknown values are finite and, for a Boolean
query~$q$, we consider the following two problems: given as input an incomplete
database~$D$, (a) return the number of completions of~$D$ that satisfy~$q$; or
(b) return the number of valuations of the nulls of~$D$ yielding a
completion that satisfies~$q$.  We obtain dichotomies between~\#P-hardness and
polynomial-time computability for these problems when~$q$ is a self-join--free
conjunctive query, and study the impact on the complexity of the following two
restrictions: (1) every null occurs at most once in~$D$ (what is called
\emph{Codd tables}); and (2) the domain of each null is the same.  Roughly
speaking, we show that counting completions is much harder than counting
valuations: for instance, while the latter is always in~\#P, we prove that the
former is not in~\#P under some widely believed theoretical complexity
assumption. Moreover, we find that both (1) and (2) can reduce the complexity
of our problems. We also study the approximability of these problems and show
that, while counting valuations always has a fully polynomial-time randomized
approximation scheme (FPRAS), in most cases counting completions does not.
Finally, we consider more expressive query languages and situate our problems
with respect to known complexity classes.

\end{abstract}

\setcopyright{acmlicensed}
\acmJournal{TOCL}
\acmYear{2021} \acmVolume{1} \acmNumber{1} \acmArticle{1} \acmMonth{1} \acmPrice{15.00}\acmDOI{10.1145/3461642}



\maketitle

\begin{toappendix}
\begin{center}
{\Huge Appendix}
\end{center}
\vspace{1cm}
\end{toappendix}

\section{Introduction}
\label{sec:introduction}
\noindent 
\paragraph{{\bf Context}}
In the database literature, {\em incomplete databases} are often used to
represent missing information in the data; see,
e.g.,~\cite{abiteboul1995foundations,van1998logical,libkin2014incomplete}.
These are traditional relational databases whose active domain can contain both
constants and \emph{nulls}, the latter representing unknown
values~\cite{imielinski1984incomplete}.  There are many ways in which one can
define the semantics of such a database, each being equally meaningful
depending on the intended application.  Under the so called \emph{closed-world
assumption}~\cite{abiteboul1995foundations,reiter1978closed}, a standard,
complete database~$\nu(D)$ is obtained from an incomplete database~$D$ by
applying a {\em valuation}~$\nu$ which replaces each null~$\bot$ in~$D$ with a
constant~$\nu(\bot)$.  The goal is then to reason about the space formed by all
valuations~$\nu$ and completions~$\nu(D)$ of~$D$.

Decision problems related to querying incomplete databases have been well
studied already.  Consider for instance the problem {\sf
Certainty}$(q(\bar{x}))$ which, for a fixed query~$q(\bar{x})$, takes as input
an incomplete database~$D$ and a tuple~$\bar{a}$ and asks whether~$\bar{a}$ is
an answer to~$q$ for every possible completion of~$D$.  By now, we have a deep
understanding of the complexity of these kind of decision problems for
different choices of query languages, including conjunctive queries (CQs) and
FO queries~\cite{imielinski1984incomplete,AbiteboulKG91}.  However, having the
answer to this question is sometimes of little help: what if it is not the case
that~$q$ is certain on~$D$? Can we still infer some useful information? This
calls for new notions that could be used to measure the certainty with
which~$q$ holds, notions which should be finer than those previously
considered. This is for instance what the recent work in
\cite{libkin2018certain} does by introducing a notion of~\emph{best answer},
which are those tuples~$\bar a$ for which the set of completions of~$D$ over
which~$q(\bar a)$ holds is maximal with respect to set inclusion.

A fundamental complementary approach to address this issue can be obtained by
considering some {\em counting problems} related to incomplete databases; more
specifically,  determining the number of completions/valuations of an
incomplete database that satisfy a query~$q$.  These problems are relevant as
they tell us, intuitively, how close is~$q$ from being certain over~$D$, i.e.,
what is the level of {\em support} that~$q$ has over the set of
completions/valuations of~$D$.  Surprisingly, such counting problems do not
seem to have been studied for incomplete databases. A reason for this omission
in the literature might be that, in general, it is assumed that the domain over
which nulls can be interpreted is infinite, and thus incomplete databases might
have an infinite number of completions/valuations.  However, in many scenarios
it is natural to assume that the domain over which nulls are interpreted is
finite, in particular when dealing with uncertainty in
practice~\cite{andritsos2006clean,benjelloun2008databases,FGJK08,AntovaKO09,SBHNW09,ASUW10}.  By
assuming this we can ensure that the number of completions and valuations are
always finite, and thus that they can be counted.  This is the setting that we
study.

\paragraph{{\bf Problems studied}}
We focus on the problems \countcompls$(q)$ and \countvals$(q)$ for a Boolean
query~$q$, which take as input an incomplete database~$D$ together with a
finite set~$\dom(\bot)$ of constants for every null~$\bot$ occurring in~$D$,
and ask the following: How many completions, resp., valuations, of~$D$
satisfy~$q$? More formally, a \emph{valuation} of~$D$ is a mapping~$\nu$ that
associates to every null~$\bot$ of~$D$ a constant~$\nu(\bot)$ in~$\dom(\bot)$.
Then, given a valuation~$\nu$ of~$D$, we denote by~$\nu(D)$ the database that
is obtained from~$D$ after replacing each null~$\bot$ with~$\nu(\bot)$, and we
call such a database a \emph{completion}.  Besides, in this article we consider
set semantics, that is, we remove repeated tuples from~$\nu(D)$.
For~\countcompls$(q)$ we count all databases of the form~$\nu(D)$ such that $q$
holds in~$\nu(D)$. Instead, for \countvals$(q)$ we count the number of
valuations~$\nu$ such that~$q$ holds in $\nu(D)$. It is easy to see that these
two values can differ, as a completion might be obtained from two different
valuations, i.e., there might exist two distinct valuations~$\nu,\nu'$ such
that~$\nu(D) = \nu'(D)$. We think that both problems are meaningful: while
\countcompls$(q)$ determines the support for~$q$ over the databases represented
by~$D$, we have that \countvals$(q)$ further refines this by incorporating the
support for a particular completion that satisfies~$q$ over the set of
valuations for~$D$.  

\begin{table*}[t]
  \centering
  \begin{tabular}{|l|m{0.17\textwidth}|m{0.25\textwidth}|m{0.17\textwidth}|m{0.17\textwidth}|}
				\hline
				&  \multicolumn{2}{ c| }{{\bf Counting valuations}} & \multicolumn{2}{ c| }{{\bf Counting completions}} \\
				\hline
				\hline
				& {\bf Non-uniform} & {\bf Uniform} & {\bf Non-uniform} & {\bf Uniform} \\
				\hline
				{\bf Naïve} &
				\begin{tabular}{l} $R(x,x)$ \\ $R(x) \land  S(x)$ \end{tabular}
& \begin{tabular}{l} \vspace{-4pt} \\$R(x,x)$ \\ $R(x) \land S(x,y) \land T(y)$ \\ $R(x,y) \land S(x,y)$ \\ \vspace{-4pt} \end{tabular} 
& \begin{tabular}{l} $R(x)$ \end{tabular}
& \begin{tabular}{l} $R(x,x)$ \\ $R(x,y)$ \end{tabular}\\
\hline
				{\bf Codd} &
				\begin{tabular}{l} $R(x) \land S(x)$ \end{tabular} 
				&
				\begin{tabular}{l}  $R(x) \land S(x,y) \land T(y)$ \\ $R(x,y) \land S(x,y)$ \\ \end{tabular} 
				& \begin{tabular}{l} $R(x)$ \end{tabular} 
				&
				\begin{tabular}{l}  $R(x,x)$ \\ $R(x,y)$ \end{tabular}\\
				\hline
			\end{tabular}
		\caption{Our dichotomies for counting valuations and completions of \sjfbcqs. For each of the eight cases, if an \sjfbcq~$q$ contains a pattern mentioned in that case, then the problem is \shp-hard (and~\shp-complete for counting valuations, as well as for counting completions over Codd tables). In turn, for each case 
		if an \sjfbcq\ $q$ does not have any of the patterns mentioned in that case, then the problem is in \fp.}
	\label{tab:dichos-count-intro}
\end{table*}

\begin{table*}[t]
  \centering
  \begin{tabular}{|l|m{0.17\textwidth}|m{0.25\textwidth}|m{0.17\textwidth}|m{0.19\textwidth}|}
				\hline
				&  \multicolumn{2}{ c| }{{\bf FPRAS for counting valuations}} & \multicolumn{2}{ c| }{{\bf FPRAS for counting completions}} \\
				\hline
				\hline
				& {\bf Non-uniform} & {\bf Uniform} & {\bf Non-uniform} & {\bf Uniform} \\
				\hline
				{\bf Naïve} &
				\begin{tabular}{l} \text{Always} \end{tabular}
& \begin{tabular}{l}  \text{Always}  \end{tabular} 
& \begin{tabular}{l} \text{Never} \end{tabular}
& \begin{tabular}{l} \text{Only when $q$ has} \\ 
\text{only unary atoms} \end{tabular}\\
\hline
				{\bf Codd} &
				\begin{tabular}{l} \text{Always} \end{tabular} 
				&
				\begin{tabular}{l} \text{Always} \end{tabular} 
				& \begin{tabular}{l} \text{Never} \end{tabular} 
				&
				\begin{tabular}{c} ? \end{tabular}\\
				\hline
			\end{tabular}
		\caption{Our results on the existence of FPRAS for solving the problems studied in the article (assuming $\np \neq \rp$).}
	\label{tab:fpras-count}
\end{table*}
We deal with the \emph{data complexity} of the problems \countcompls$(q)$ and
\countvals$(q)$, focusing on obtaining dichotomy results for them  in terms of
counting complexity classes, as well as studying the existence of randomized
algorithms that approximate their results under probabilistic guarantees.  For
the dichotomies, we concentrate on self-join-free Boolean conjunctive queries
(\sjfbcqs). This assumption simplifies the mathematical analysis, while at the
same time defines a setting which is rich enough for many of the theoretical
concepts behind these problems to appear in full force.  Notice that a similar
assumption is used in several works that study counting problems over
probabilistic and inconsistent databases; see,
e.g.,~\cite{dalvi2011queries,maslowski2013dichotomy}. To simplify further the presentation, in the bulk of the article
we mainly consider self-join--free Boolean conjunctive queries that do not contain constants; however, we explain later (in Section~\ref{sec:extensions})
how our results can be extended to queries that can contain constants and free variables.

To refine our analysis, we study two restrictions of the problems
\countcompls$(q)$ and \countvals$(q)$ based on natural modifications of the
semantics, and analyze to what extent these restrictions simplify our problems.
For the first restriction we consider incomplete databases in which each null
occurs exactly once, which corresponds to the well-studied setting of {\em Codd
tables} -- as opposed to {\em naive tables} where nulls are allowed to have
multiple occurrences. We denote the corresponding problems by~$\ccountvals(q)$
and~$\ccountcompls(q)$.  For the second restriction, we consider {\em uniform}
incomplete databases in which all the nulls share the same domain -- as opposed
to the basic {\em non-uniform} setting in which all nulls come equipped with
their own domain.  We denote the corresponding problems by~$\ucountvals(q)$
and~$\ucountcompls(q)$.  When both restrictions are in place, we denote the
problems by~$\cucountvals(q)$ and~$\cucountcompls(q)$.

\paragraph{{\bf Our dichotomies for exact counting}} We provide 
complete characterizations of the complexity of counting valuations and
completions satisfying a given~\sjfbcq\ $q$, when the input is a Codd table or
a naive table, and is a non-uniform or a uniform incomplete database (hence we
have~eight cases in total). 
Our eight dichotomies express that these problems are  either tractable
or~\shp-hard, and that the tractable cases can be fully characterized by the
absence of certain forbidden {\em patterns} in~$q$. In essence, a pattern is
simply an \sjfbcq\ which can be obtained from~$q$  
by deleting atoms 
and occurrences of variables (the exact definition of this notion is given in
Section \ref{sec:countvals-sjfcqs}).  Our characterizations are presented in
Table~\ref{tab:dichos-count-intro}.  By analyzing this table we can draw some
important conclusions as explained next.

\medskip
\noindent \underline{\countcompls$(q)$ and \countvals$(q)$ are computationally
difficult:} For very few \sjfbcqs~$q$  the aforementioned problems can be
solved in polynomial time. Take as an example the uniform setting over naive
tables.  Then \ucountvals$(q)$ is~\shp-hard as long as~$q$ contains the
pattern~$R(x,x)$, or~$R(x)\land S(x,y)\land T(y)$, or $R(x,y) \wedge S(x,y)$.
That is, as long as there is an atom in~$q$ that contains a repeated
variable~$x$,  or a pair~$(x,y)$ of variables that appear in an atom and
both~$x$ and~$y$ appear in some other atoms in~$q$.  By contrast, for this same
setting, \ucountcompls$(q)$ is~\shp-hard as long as~$q$ contains the
pattern~$R(x,y)$ or~$R(x,x)$, that is, as long as there is an atom in~$q$ that
is not of arity one. 

\smallskip
\noindent \underline{\countvals$(q)$ is always easier than \countcompls$(q)$:}
In all of the possible versions of our problem, the tractable cases for
\countvals$(q)$ are a strict superset of the ones for \countcompls$(q)$. For
instance, $\cucountcompls(\exists x \exists y\, R(x,y))$ is hard, while
$\cucountvals(\exists x \exists y\, R(x,y))$ is tractable.

\smallskip
\noindent \underline{Even counting completions is hard:} While counting the
total number of valuations for an incomplete database can always be done in
polynomial time, observe from Table \ref{tab:dichos-count-intro} that the
problem $\cucountcompls(\exists x \exists y\, R(x,y))$ is~\shp-hard, and thus
that simply counting the completions of a uniform Codd table with a single
binary relation~$R$ is~\shp-hard.  Moreover, we show that in the non-uniform
case a single unary relation suffices to obtain~\shp-hardness. 

\smallskip
\noindent \underline{Codd tables help but not much:} We show that counting
valuations is easier for Codd tables than for naive tables. In particular,
there is always an \sjfbcq~$q$ such that counting  the valuations that
satisfy~$q$ is~\shp-hard, yet it becomes tractable when restricted to the case
of Codd tables.  However, for counting completions, both in the uniform and
non-uniform setting, the sole restriction to Codd tables presents no benefits:
for every \sjfbcq~$q$, we have that $\countcompls(q)$ (resp.,
$\ucountcompls(q)$) is \shp-hard if and only if  $\ccountcompls(q)$ (resp.,
$\cucountcompls(q)$) is \shp-hard.

\smallskip
\noindent \underline{Non-uniformity complicates things:} All versions of our
problems become harder in the non-uniform setting.  This means that in all
cases there is an \sjfbcq~$q$ for which counting valuations is tractable on
uniform incomplete databases, but becomes~\shp-hard assuming non-uniformity,
and the same holds for counting completions. 

\paragraph{{\bf Our dichotomies for approximate counting}} 
Although \countvals$(q)$ can be~\shp-hard, we prove that good
randomized approximation algorithms can be designed for this problem.  More
precisely, we give a general condition under which~$\countvals(q)$ admits a
{\em fully polynomial-time randomized approximation
scheme}~\cite{jerrum1986random} (FPRAS).  This condition applies in particular
to all \emph{unions of Boolean conjunctive queries}.  Remarkably, we show that
this no longer holds for \countcompls$(q)$; more precisely, there exists an
\sjfbcq~$q$ such that \countcompls$(q)$ does not admit an FPRAS under a widely
believed complexity theoretical assumption. More surprisingly, even counting
the completions of a uniform incomplete database containing a single binary
relation does not admit an FPRAS under such an assumption (and in the
non-uniform case, a single unary relation suffices).  Generally, for \sjfbcqs,
we obtain seven dichotomies for our problems between polynomial-time
computability of exact counting and non admissibility of an FPRAS.  The only
case that we did not completely solve is that of~$\cucountcompls(q)$.  Our
dichotomies for approximate counting are illustrated in
Table~\ref{tab:fpras-count}.

\paragraph{\bf Beyond \shp} 
It is easy to see that the problem of counting valuations is always in \shp,
provided that the model checking problem for~$q$ is in~\ptime.  This is no
longer the case for counting completions, and in fact we show that, under a
complexity theoretical assumption, there is an \sjfbcq\ $q$ for which
\ucountcompls$(q)$ is not in~\shp.  This does not hold if restricted to Codd
tables, however, as we prove that \ccountcompls$(q)$ is always in \shp\ when
the model checking problem for~$q$ is in~\ptime. 

For reasons that we explain in the article, a suitable complexity class for the
problem \countcompls$(q)$ is \spanp, which is defined as the class of counting
problems that can be expressed as the number of different accepting outputs of
a nondeterministic Turing machine running in polynomial time.  While we have
not managed to prove that there is an \sjfbcq\ $q$ for which \countcompls$(q)$
is \spanp-complete, we show that this is the case for the problem
of counting completions for the negation of an \sjfbcq, even in the uniform
setting; that is, we show that \ucountcompls$(\neg q)$ is \spanp-complete for
some \sjfbcq\ $q$. Finally, we also show that~\spanp\ is the right complexity
class for counting valuations of queries for which model checking is in~\np.

\paragraph{\bf Extension to queries with constants and free variables.}

As we said already, for pedagogical reasons we mostly present our results by
considering queries that are Boolean and that do not have constants.  In
Section~\ref{sec:extensions} however, we explain how to extend these results to
the case of queries that have free variables and that can contain constants.
For the case of a query~$q(\bar x)$ with free variables~$\bar x$, our counting
problems are defined in the expected way; for instance the
problem~$\countvals(q(\bar x))$ takes as input an incomplete database~$D$, a
tuple of constants~$\bar a$ of same arity as~$\bar x$, and it outputs the
number of valuations~$\nu$ of~$D$ such that~$\bar a$ in an answer to~$q(\bar
x)$ on~$\nu(D)$. We then extend our dichotomies and approximation results in
this setting.\\

The current article extends the conference article~\cite{arenas2020counting} in the following ways:
\begin{itemize}
  \item In~\cite{arenas2020counting} we left open the dichotomy for~$\cucountvals(q)$, i.e., for counting valuations of~\sjfbcqs\ for Codd tables under the uniform setting. 
 We close this case here, by finding one more hard pattern (namely, the pattern~$\exists x,y \, R(x,y) \land S(x,y)$) and showing that the problem can be solved in polynomial time for all other queries. 
  \item We added Section~\ref{sec:extensions} which explains how our framework can be extended to handle queries with constants and free variables;
  \item Proposition~\ref{prp:check-completion-np-c}, which establishes the \np-completeness of checking if a set of facts is a possible completion of an incomplete database, is new;
  \item Finally, full proofs of most results are included in the body of the article.
\end{itemize}

\paragraph{{\bf Organization of the article}}

		We start with the main terminology used in the article in
Section~\ref{sec:preliminaries}, and then present in
Section~\ref{sec:countvals-sjfcqs} our four dichotomies on~\countvals$(q)$
when~$q$ is an \sjfbcq, and the input incomplete database can be Codd or not,
and the domain can be uniform or not. We then establish the four dichotomies
on~\countcompls$(q)$ in Section~\ref{sec:countcompls-sjfcqs}. In
Section~\ref{sec:approx}, we study the approximability complexity of our
problems. We then give in Section~\ref{sec:misc} some general considerations
about the exact complexity of the problem~\countcompls$(q)$ going beyond \shp.
We explain in Section~\ref{sec:extensions} how to extend our results to queries
with constants and free variables. In Section~\ref{sec:related}, we discuss
related work and explain the differences with the problems considered in this
article. Last, we provide some conclusions and mention possible directions for
future work in Section~\ref{sec:conclusion}.

\section{Preliminaries}
\label{sec:preliminaries}
\paragraph{{\bf Relational databases and conjunctive queries}}
A \emph{relational schema}~$\sigma$ is a finite non-empty set of relation
symbols written $R$, $S$, $T$, \dots, each with its associated \emph{arity},
which is denoted by $\arity(R)$.  Let~$\consts$ be a countably infinite set of
constants. A {\em database}~$D$ over~$\sigma$ is a set of {\em facts} of the
form $R(a_1, \ldots, a_{\arity(R)})$ with~$R \in \sigma$, and where each
element~$a_i \in \consts$. For~$R\in \sigma$, we denote by~$D(R)$ the subset
of~$D$ consisting of facts over~$R$. Such a set is usually called a
\emph{relation} of~$D$.

A Boolean query~$q$ is a query that a database~$D$ can \emph{satisfy}
(written~$D \models q$) or not (written~$D \not\models q$).  
If~$q$ is a Boolean query, then~$\lnot q$ is the Boolean query such that~$D\models \lnot q$
if and only if~$D\not\models q$.
A {\em Boolean
conjunctive query} (BCQ) over~$\sigma$ is an FO formula of the form
\begin{equation} \label{eq:cq}  \exists \bar x \, \big(R_1(\bar x_1) \land   \,
\dots \, \land R_m(\bar x_m)\big), \end{equation} where all variables are
existentially quantified, and where for each $i \in [1,m]$, we have that $R_i$
is a relation symbol in~$\sigma$ and~$\bar x_i$ is a tuple of variables with
$|\bar x_i| = \arity(R_i)$.  To avoid trivialities, we will always assume
that~$m \geq 1$, i.e., the query has at least one atom, and also
that~$\arity(R_i) \geq 1$ for all atoms.  
Observe that we do not allow constants to appear in the query (but we will come back to this issue in Section~\ref{sec:extensions}).
For simplicity, we typically write a
BCQ~$q$ of the form \eqref{eq:cq} as $$R_1(\bar x_1)  \land   \, \dots \, \land
R_m(\bar x_m),$$ and it will be implicitly understood that all variables in~$q$
are existentially quantified.  As usual, we define the semantics of a BCQ in
terms of {\em homomorphisms}.  A homomorphism from~$q$ to a database~$D$ is a
mapping from the variables in~$q$ to the constants used in~$D$ such that
$\{R_1(h(\bar x_1)),\dots,R_m(h(\bar x_m))\} \subseteq D$.  Then, we have~$D
\models q$ if there exists a homomorphism from~$q$ to~$D$.   A
\emph{self-join--free BCQ} (\sjfbcq) is a BCQ such that no two atoms use the
same relation symbol.

\begin{figure*}[t]
	\begin{tabular}{c|cccccc}
		\toprule
		$(\nu(\bot_1),\nu(\bot_2))$ & $(a,a)$ & $(a,b)$ & $(b,a)$ & $(b,b)$ & $(c,a)$ & $(c,b)$ \\
		\hline
		$\nu(D)$

		&
		\begin{tabular}[t]{cc}
			\toprule
			\multicolumn{2}{c}{$S$} \\
			\midrule
			$a$ & $b$ \\
			$a$ & $a$ \\
			\bottomrule
		\end{tabular}

		& 
		\begin{tabular}[t]{cc}
			\toprule
			\multicolumn{2}{c}{$S$} \\
			\midrule
			$a$ & $b$ \\
			$a$ & $a$ \\
			\bottomrule
		\end{tabular}

		& 
		\begin{tabular}[t]{cc}
			\toprule
			\multicolumn{2}{c}{$S$} \\
			\midrule
			$a$ & $b$ \\
			$b$ & $a$ \\
			$a$ & $a$ \\
			\bottomrule
		\end{tabular}

		& 
		\begin{tabular}[t]{cc}
			\toprule
			\multicolumn{2}{c}{$S$} \\
			\midrule
			$a$ & $b$ \\
			$b$ & $a$ \\
			\bottomrule
		\end{tabular}

		& 
		\begin{tabular}[t]{cc}
			\toprule
			\multicolumn{2}{c}{$S$} \\
			\midrule
			$a$ & $b$ \\
			$c$ & $a$ \\
			$a$ & $a$ \\
			\bottomrule
		\end{tabular}

		&
		\begin{tabular}[t]{cc}
			\toprule
			\multicolumn{2}{c}{$S$} \\
			\midrule
			$a$ & $b$ \\
			$c$ & $a$ \\
			\bottomrule
		\end{tabular} \\[2.2cm]
		\hline
		$\nu(D) \models q$? & Yes & Yes & Yes & No & Yes & No \\
		\bottomrule

	\end{tabular}
\caption{The~six valuations of the (non-uniform) incomplete database~$D = (T,\dom)$ with $T=\{S(a,b),S(\bot_1,a),S(a,\bot_2)\}$ from Example~\ref{expl:problems}, and their corresponding completions.
	The Boolean conjunctive query~$q$ is $\exists x \, S(x,x)$.}
\label{fig:example-completions}
\end{figure*}

\paragraph{{\bf Incomplete databases}}
Let~$\nulls$ be a countably infinite set of nulls (also called \emph{labeled}
or \emph{marked} nulls in the literature), which is disjoint with~$\consts$. An
{\em incomplete database} over schema~$\sigma$ is a pair~$D = (T,\dom)$,
where~$T$ is a database over~$\sigma$ whose facts contain elements in~$\consts
\cup \nulls$, and where $\dom$ is a function that associates to every
null~$\bot$ occurring in~$D$ a subset~$\dom(\bot)$ of~$\consts$.
Intuitively,~$T$ is a database that can mention both constants and nulls,
while~$\dom$ tells us where nulls are to be interpreted.  Following the
literature, we call~$T$ a {\em naive table} \cite{imielinski1984incomplete}.

An incomplete database~$D = (T,\dom)$ can represent potentially many complete
databases, via what are called \emph{valuations}.  A valuation of~$D$ is simply
a function~$\nu$ that maps each null~$\bot$ occurring in~$T$ to a
constant~$\nu(\bot) \in \dom(\bot)$.  Such a valuation naturally defines a
\emph{completion of~$D$}, denoted by~$\nu(T)$, which is the complete database
obtained from~$T$ by substituting each null~$\bot$ appearing in~$T$
by~$\nu(\bot)$.  It is understood, since a database is a \emph{set} of facts,
that~$\nu(T)$ does not contain duplicate facts.  By paying attention to
completions of incomplete databases that are generated exclusively by applying
valuations to them, we are sticking to the so called {\em closed-world}
semantics of incompleteness~\cite{abiteboul1995foundations,reiter1978closed}.
This means that the databases represented by an incomplete database~$D =
(T,\dom)$ are not open to adding facts that are not ``justified'' by the facts
in~$T$. 

\begin{example}
\label{expl:completion}
	{\em Let~$D = (T,\dom)$ be the incomplete database consisting of the
naive table $T = \{S(\bot_1,\bot_1),S(a,\bot_2)\}$, and where $\dom(\bot_1) =
\{a,b\}$ and~$\dom(\bot_2) = \{a,c\}$.  Let~$\nu_1$ be the valuation
mapping~$\bot_1$ to~$b$ and~$\bot_2$ to~$c$.  Then~$\nu_1(T)$ is
$\{S(b,b),S(a,c)\}$.  Let~$\nu_2$ be the valuation mapping both~$\bot_1$
and~$\bot_2$ to~$a$.  Then~$\nu_2(T)$ is~$\{S(a,a)\}$. On the other hand, the
function~$\nu$ mapping~$\bot_1$ and~$\bot_2$ to~$b$ is not a valuation of~$D$,
because~$b\notin \dom(\bot_2)$.} \qed 
\end{example}

When every null occurs at most once in~$T$, then~$D$ is what is called a
\emph{Codd table}~\cite{codd1975understanding}; for instance, the incomplete
database in Example~\ref{expl:completion} is not a Codd table because~$\bot_1$
occurs twice.  We also consider {\em uniform} incomplete databases in which the
domain of every null is the same.  Formally, a uniform incomplete database is a
pair~$D = (T,\dom)$, where~$T$ is a database over~$\sigma$ and~$\dom$ is a
subset of~$\consts$. The difference now is  that a valuation~$\nu$ of~$D$ must
simply satisfy~$\nu(\bot) \in \dom$ for every null of~$D$.

We will often abuse notation and use~$D$ instead of~$T$; for instance, we
write~$\nu(D)$ instead of~$\nu(T)$, or~$R(a,a) \in D$ instead of~$R(a,a) \in
T$, or again~$D(R)$ instead of~$T(R)$.

\paragraph*{\bf Counting problems on incomplete databases.}
We will study two kinds of counting problems for incomplete databases: problems
of the form $\countvals(q)$, that count the number of \emph{valuations}~$\nu$
that yield a completion~$\nu(D)$ satisfying a given BCQ~$q$, and problems of
the form~$\countcompls(q)$, that count the number of \emph{completions} that
satisfy~$q$.  The query~$q$ is assumed to be fixed, so that each query gives
rise to different counting problems, and we are considering the \emph{data
complexity}~\cite{vardi1982complexity} of these problems.

Before formally introducing our problems, let us observe that they are well
defined if we assume that the set of constants to which a null can be mapped to
is finite.  Hence, for the (default) case of an incomplete database~$D = (T,
\dom)$, we assume that~$\dom(\bot)$ is always a finite subset of~$\consts$.
Similarly, for the case of a uniform incomplete database~$D = (T, \dom)$, we
assume that~$\dom$ is a finite subset of~$\consts$.  Finally, given a Boolean
query~$q$, we use notation $\sig(q)$ for the set of relation symbols occurring
in~$q$.  With these ingredients, we can define our problems for the (default)
case of incomplete naive tables and a Boolean query~$q$. 

\medskip
\begin{center}
\fbox{\begin{tabular}{lp{6.1cm}}
\small{PROBLEM} : & $\countvals(q)$ 
\\
{\small INPUT} : & An incomplete database $D$ over $\sig(q)$ 
\\
{\small OUTPUT} : & Number of valuations $\nu$ of $D$ with $\nu(D) \models q$ 
\end{tabular}}

\medskip
\fbox{\begin{tabular}{lp{6.1cm}}
\small{PROBLEM} : & $\countcompls(q)$
\\
{\small INPUT} : & An incomplete database $D$ over $\sig(q)$ 
\\
{\small OUTPUT} : & Number of completions $\nu(D)$ of $D$ with $\nu(D) \models q$
\end{tabular}}
\end{center}
\medskip

We also consider the uniform variants of these problems, in which the input~$D$
is a uniform incomplete database over $\sig(q)$, and the restriction of these
problems where the input is a Codd table instead of a naive table.  We then use
the terms $\ucountvals(q)$, $\ucountcompls(q)$ when restricted to the uniform
case, $\ccountvals(q)$, $\ccountcompls(q)$ when restricted to Codd tables, and
$\cucountvals(q)$, $\cucountcompls(q)$ when both restrictions are applied.

As we will see, even though the problems \countvals$(q)$ and \countcompls$(q)$
look similar, they are of a different computational nature; this is because two
distinct valuations can produce the same completion of an incomplete database.
We illustrate this phenomenon in the following example.

\begin{example}
\label{expl:problems}
	{\em Let~$q$ be the Boolean conjunctive query~$\exists x \, S(x,x)$,
and~$D$ be the (non-uniform) incomplete database~$D = (T,\dom)$, with $T \, =
\, \{S(a,b)$, $S(\bot_1,a),S(a,\bot_2)\}$,  $\dom(\bot_1)=\{a,b,c\}$ and
$\dom(\bot_2)=\{a,b\}$.  We have depicted in
Figure~\ref{fig:example-completions} the~six valuations of~$D$ together with
the completions that they define.  Out of these six valuations~$\nu$, only four
are such that~$\nu(D) \models q$, so we have	$\countvals(q)(D) = 4$.
Moreover, there are only~$3$ distinct completions of~$D$ that satisfy~$q$ --
because the first two are the same -- so $\countcompls(q)(D) = 3$.} \qed
\end{example}

\paragraph*{\bf Counting complexity classes.}

Given two problems~$A,B$, we write $A \leq_{\mathrm{T}}^{\mathrm{p}} B$
when~$A$ reduces to~$B$ under polynomial-time Turing reductions.  When both~$A$
and~$B$ are counting problems, we write~$A \pr B$ when~$A$ can be reduced
to~$B$ under polynomial-time \emph{parsimonious} reductions, i.e., when there
exists a polynomial-time computable function~$f$ that transforms an input~$x$
of~$A$ to an input~$f(x)$ of~$B$ such that~$A(x)=B(f(x))$.  We say that a
counting problem is in~\fp\ when it can be solved in polynomial time.  We will
consider the counting complexity class \shp~\cite{valiant1979complexity} of
problems that can be expressed as the number of accepting paths of a
nondeterministic Turing machine running in polynomial time.
Following~\cite{valiant1979complexity,V79}, we define \#P-hardness using Turing
reductions.  It is clear that $\fp \subseteq \shp$. Moreover, this inclusion is
widely believed to be strict. Therefore, proving that a counting problem is
\shp-hard implies that it cannot be solved in polynomial time under such an
assumption.

\section{Dichotomies for counting valuations}
\label{sec:countvals-sjfcqs}
\begin{toappendix}
\label{apx:countvals-sjfcqs}
\end{toappendix}
In this section, for a fixed \sjfbcq\ q, we study the complexity of the
problem of computing, given as input an incomplete database $D$, the number of
valuations $\nu$ of $D$ such that $\nu(D)$ satisfies $q$.  Recall that we have
four cases to consider for this problem depending on whether we focus on naive
or on Codd tables, where nulls are restricted to appear at most once, and
whether we focus on non-uniform or uniform incomplete databases, where nulls
are restricted to have the same domain.  Our specific goal then is to
understand whether the problem is tractable (in \fp) or \shp-hard in these
scenarios, depending on the shape of $q$.

To this end, the shape of an \sjfbcq\ $q$ will be characterized by the presence
or absence of certain specific patterns.  In the following definition, we
introduce the necessary terminology to formally talk about the presence of a
pattern in a query. 

\begin{definition}
\label{def:pattern}
	Let~$q,q'$ be \sjfbcqs.  We say that~$q'$ is a \emph{pattern} of~$q$ if
$q'$ can be obtained from~$q$ by using an arbitrary number of times and in any
order the following operations: deleting an atom, deleting an occurrence of a
variable, renaming a relation to a fresh one, renaming a variable to a fresh
one, and reordering the variables in an atom.\footnote{\label{foot:trivial}We
remind the reader that we assume all sjfBCQs to contain at least one atom and
that all atoms must contain at least one variable.} \qed
\end{definition}

\begin{example}
\label{expl:pattern}
 {\em Recall that we always omit existential quantifiers in Boolean queries.
Then we have that $q' = R'(u,u,y) \land  S'(z)$ is a pattern of~$ q = R(u,x,u)
\land  S'(y,y) \land  T(x,s,z,s)$. Indeed, $q'$ can be obtained from $q$ by
deleting atom $T(x,s,z,s)$, renaming $R(u,x,u)$ as $R'(u,x,u)$ to
obtain~$R'(u,x,u) \land S'(y,y)$, reordering the variables in~$R'(u,x,u)$ to
obtain~$R'(u,u,x) \land S'(y,y)$, renaming variable~$y$ into~$z$ to
obtain~$R'(u,u,x) \land S'(z,z)$, deleting the second variable occurrence
in~$S'(z,z)$ to obtain~$R'(u,u,x) \land S'(z)$, and finally renaming
variable~$x$ into~$y$ to obtain~$q'$.} \qed
\end{example}

We point out that in Definition~\ref{def:pattern}, the important parts are
those about deleting atoms and variable occurrences. The parts about reordering
variable occurences inside an atom and about renaming relations and variables
to fresh ones have obviously no effect on the complexity of the
problem\footnote{This is in particular because the conjunctive queries we
consider have no self-joins (otherwise, reordering variables inside an atom
could change the complexity).}; these are only here to allow us to formally
say, for instance, that “$R(x)$ is a pattern of~$R(y)$”, or that~“$S(x,u,x)$ is
a pattern of~$T(w,z,z)$” (as these are, in essence, the same queries).

In the following general lemma, we show that if $q'$ is a pattern of $q$, then
each of the problems considered in this section is as hard for $q$ as it is
for $q'$. Recall in this result that unless stated otherwise, our problems are
defined for naive tables under the non-uniform setting.

\begin{lemma}
\label{lem:pattern-parsimonious}
	Let~$q,q'$ be \sjfbcqs\ such that~$q'$ is a pattern of~$q$.  Then we
have~$\countvals(q') \pr \countvals(q)$.  Moreover, the same results hold if we
restrict to Codd tables, and/or to the uniform setting.
\end{lemma}

\begin{proof}
	We first present the proof for $\countvals(q') \pr \countvals(q)$, that
is, for naive tables in the non-uniform setting.  First of all, observe that we
can assume without loss of generality that we did not reorder the variables in
the atoms nor renamed relation names or variables by fresh ones, because, as
mentioned above, this does not change the complexity of the
problem.\footnote{Formally, one can first check that we can assume without loss
of generality that~$q'$ was obtained from~$q$ by first deleting some atoms and
variable occurences to obtain a query~$q''$, and then performing some renamings
and variable reorderings to obtain~$q'$ (that is, we can always push the
renaming and reordering parts at the end of the transformation). But then,
since~$\countvals(q'')$ and~$\countvals(q')$ are obviously the same problem, we
can assume~$q'=q''$.} We can then write~$q$ as $R_1(\overline{x_1}) \land
\ldots \land  R_m(\overline{x_m})$ and $q'$ as $R_{j_1}(\overline{x'_{j_1}})
\land  \ldots \land  R_{j_p}(\overline{x'_{j_p}})$, where $1 \leq j_1 < \ldots
< j_p \leq m$ and $\overline{x'_{j_k}}$ is obtained from $\overline{x_{j_k}}$
by deleting some variable occurrences but not all\footnote{See
Footnote~\ref{foot:trivial}.}, and the other atoms have been deleted.  Let~$D'$
be an incomplete database input of~$\countvals(q')$.  Let~$A$ be the set of
constants that are appearing in~$D'$ or are in a domain of some null occurring
in~$D'$.  For~$1 \leq k \leq p$, we construct the relation~$D(R_{j_k})$ from
the relation~$D'(R_{j_k})$.  Let us assume that~$\overline{x_{j_k}}$ is the
tuple $(x_1, \ldots, x_r)$ (with some variables possibly being equal).  We
initialize~$D(R_{j_k})$ to be empty, and then for every tuple~$\overline{t'}$
in $D'(R_{j_k})$ we add to~$D(R_{j_k})$ all the tuples~$\overline{t}$ that can
be obtained from~$\overline{t'}$ in the following way for $1 \leq i \leq r$:
\begin{itemize}
	\item[\bf a)] If $x_i$ is a variable occurrence that has not been deleted from~$\overline{x_{j_k}}$,
		then copy the element (constant or null) of $\overline{t'}$ corresponding to that variable occurrence to the $i$-th position of~$\overline{t}$;
	\item[\bf b)] Otherwise, if $x_i$ is a variable occurrence that has been deleted from~$\overline{x_{j_k}}$,
			then fill the~$i$-th position of~$\overline{t}$ 
			with every possible constant from~$A$.
\end{itemize}
		Then we construct the relations~$D(R_i)$ where~$R_i$ does not
appear in~$q'$ (this can happen if we have deleted the atom
$R_i(\overline{x_i})$) by filling it with every possible~$R_i$-fact over~$A$.
We leave the domains of all nulls unchanged.  The whole construction can be
performed in polynomial time (this uses the fact that $q$ is assumed to be
fixed, so that the arities of the relations mentioned in $q$ are fixed).
Hence, it only remains to be checked that~$\countvals(q')(D') =
\countvals(q)(D)$, that is, that the reduction works and is indeed
parsimonious.  It is clear that the valuations of~$D'$ are exactly the same as
the valuations of~$D$ (because they have the same sets of nulls).  Hence it is
enough to verify that for every valuation~$\nu$, we have~$\nu(D') \models q'$
if and only if~$\nu(D) \models q$.  Let~$h'$ be a homomorphism from~$q'$ to
$\nu(D')$ witnessing that~$\nu(D') \models q'$ (i.e., we have $h'(q) \subseteq
\nu(D')$).  Then~$h'$ can clearly be extended in the expected way into a
homomorphism~$h$ from~$q$ to~$\nu(D)$: this is in particular thanks to the fact
that we filled the missing columns with every possible constant.  Conversely,
let $h$ be a homomorphism from~$q$ to $\nu(D)$ witnessing that~$\nu(D) \models
q$.  Then the restriction~$h'$ of~$h$ to the variables occurring in~$q'$ is
such that~$h(q') \subseteq \nu(D')$, hence we have~$\nu(D') \models q'$.  This
concludes the proof for the case of naive tables in the non-uniform setting.
For the cases of Codd tables and/or for the uniform setting, the reduction is
exactly the same. Indeed, the domains of the nulls are unchanged, and it is
clear that the presented construction preserves the property of being a Codd
table.
\end{proof}

	The idea is then to show the \shp-hardness of our problems for some
simple patterns, which then we combine with
Lemma~\ref{lem:pattern-parsimonious} and with some tractability proofs to
obtain the desired dichotomies.  Our findings are summarized in the first two
columns of Table \ref{tab:dichos-count-intro} in the introduction.  We first
focus on the two dichotomies for the non-uniform setting in
Section~\ref{subsec:countvals-non-uniform}, and then we move to the case of
uniform incomplete databases in Section~\ref{subsec:countvals-uniform}. We
explicitly state when a \shp-hardness result holds even in the restricted
setting in which there is a fixed domain over which nulls are interpreted. In
other words, when there is a fixed domain $A$ such that the incomplete
databases used in the reductions are of the form $D = (T,\dom)$ and $\dom(\bot)
\subseteq A$, for each null $\bot$ of $T$.

\subsection{The complexity on the non-uniform case}
\label{subsec:countvals-non-uniform}

In this section, we study the complexity of the problems
$\countvals(q)$ and~$\ccountvals(q)$, providing dichotomy results in both
cases. We start by proving the \shp-hardness results needed for these
dichotomies. We first show that $\countvals(R(x,x))$ is \shp-hard by actually
proving that hardness holds already in the uniform case.

\begin{proposition}
\label{prp:Rxx-hard}
\begin{sloppypar}
	$\ucountvals(R(x,x))$ is \shp-hard. 
This holds even in the restricted setting in which all nulls are interpreted
over the same fixed domain~$\{1,2,3\}$.
\end{sloppypar}
\end{proposition}
\begin{proof}

We reduce from the problem of counting the number of
$3$-colorings of a graph~$G=(V,E)$, which is
\shp-hard~\cite{jaeger1990computational}.  For every node~$v\in V$ we have a
null~$\bot_v$, and for every edge~$\{u,v\} \in E$ we have the
facts~$R(\bot_v,\bot_u)$ and $R(\bot_u,\bot_v)$.  The domain  of the nulls
is~$\{1,2,3\}$.  It is then clear that the number of valuations of the
constructed database that do not satisfy~$R(x,x)$ is exactly the number of
$3$-colorings of~$G$. Since the total number of valuations can be computed in
PTIME, this concludes the reduction.
	\end{proof}

The next pattern that we consider is~$R(x) \land S(x)$.	This time, we
can show \shp-hardness of the problem even for Codd databases.

\begin{proposition}
\label{prp:RxSx-hard}
	$\ccountvals(R(x) \land S(x))$ is \shp-hard.
\end{proposition}

\begin{proof}
We start by recalling the setting of {\em consistent query answering under key
constraints}.  Intuitively, in this case we are given a set~$\Sigma$ of keys
and a database~$D$ that does not necessarily satisfy~$\Sigma$. Then the task is
to reason about the set of all {\em repairs} of~$D$ with respect to~$\Sigma$
\cite{ABC99}.  In our context, this means that one wants to count the number of
repairs of~$D$ with respect to~$\Sigma$ that satisfy a given CQ~$q$.  When~$q$
and~$\Sigma$ are fixed, we call this problem~$\#${\sf Repairs}$(q,\Sigma)$;
see, e.g., \cite{maslowski2013dichotomy}.  We formalize these notions below. 

Here we focus on the case when $\Sigma$ is a set of {\em primary} keys.  Recall
that this means that each relation name $R \in  \sigma$ of arity $n$ comes
equipped with its own key, i.e., ${\sf key}(R) = A$, where $A = \emptyset$ or
$A = [1,\ldots,p]$ for some $p \in \{1,\dots,n\}$.  Henceforth, $D$ is {\em
inconsistent} with respect to $\Sigma$ if there is a relation name $R \in
\sigma$ and facts $R(\bar a),R(\bar b) \in D$ with~$\bar{a} \neq \bar{b}$ such
that $$\text{${\sf key}(R) = A \quad$ and $\quad \pi_A(\bar a) = \pi_A(\bar
b)$.}$$   In this case we say that the pair $(R(\bar a), R(\bar b))$ is {\em
key-violating}.   Let us define a {\em block} in a database $D$ with respect to
a set $\Sigma$ of primary keys to be any maximal set $B$ of facts from $D$ such
that the facts in $B$ are pairwise key-violating.  A {\em repair} of $D$ with
respect to $\Sigma$ is a subset $D'$ of $D$ that is obtained by choosing
exactly one tuple from each block of $D$ with respect to $\Sigma$.  

Let us consider a schema $\sigma$ with two binary relations $R'$ and $S'$, such
that ${\sf key}(R') = {\sf key}(S') = \{1\}$.  That is, the first attribute of
both $R'$ and $S'$ defines a key over such relations.  We define this set of
keys over $\sigma$ to be $\Sigma$.  Also, let $q = \exists x,y,z (R'(y,x)
\wedge S'(z,x))$.  For simplicity, we write the pair~$(q,\Sigma)$
as~$R'(\underbar{y},x) \land S'(\underbar{z},x)$.  The problem $\#${\sf
Repairs}$(R'(\underbar{y},x) \land S'(\underbar{z},x))$, which given a database
$D'$ over schema $\sigma$ aims at computing the number of repairs of $D'$ under
$\Sigma$ that satisfy $q$, is known to be \shp-complete
\cite{maslowski2013dichotomy}.\footnote{To see
that~\cite{maslowski2013dichotomy} establishes the hardness
of~$q=R'(\underbar{y},x) \land S'(\underbar{z},x)$, first apply their rewrite
rule R7 (from Fig.~6) to obtain~$q'=R'(\underbar{y},x) \land
S'(\underbar{x},x)$, then apply rewrite rule R10 to
obtain~$q''=R'(\underbar{y},x) \land S'(\underbar{x},a)$. Then,~$q''$ is
\shp-hard by Lemma 19, and so is~$q$ by Lemma 7.}

Now, observe that the \shp-hardness of $\ccountvals(R(x) \land S(x))$ easily
follows from the hardness of the problem $\#${\sf Repairs}$(R'(\underbar{y},x)
\land S'(\underbar{z},x))$.  In fact, let~$D'$ be a database with binary
relation~$R',S'$. We construct an incomplete Codd database~$D$ with unary
relations~$R,S$ as follows. For every constant~$a$ that appears in the first
attribute of~$R'$, we have a tuple~$R(\bot)$ in~$D$, where~$\bot$ is a fresh
null, and we set~$\dom(\bot) = \{b \mid R'(a,b) \in D'\}$.  For every
constant~$a$ that appears in the first attribute of~$S'$, we have a
tuple~$S(\bot)$ in~$D$, where~$\bot$ is a fresh null, and we set~$\dom(\bot) =
\{b \mid S'(a,b) \in D'\}$. It is then clear that the number of repairs of~$D'$
that satisfy~$R'(\underbar{y},x) \land S'(\underbar{z},x)$ is equal to the
number of valuations of~$D$ that satisfy~$R(x) \land S(x)$, thus concluding the
proof.  We point out here that another proof of
Proposition~\ref{prp:RxSx-hard}, that uses different techniques, can be found
in the conference version of the article~\cite{arenas2020counting} (the proof
that we presented here is shorter).
	\end{proof}

Already with Propositions~\ref{prp:RxSx-hard} and~\ref{prp:Rxx-hard},
we have all the relevant hard patterns for the non-uniform setting.  We start
by proving our dichotomy result for naive tables, which is our default case.

\begin{theorem}[dichotomy]
\label{thm:countvals-sjfcqs}
Let~$q$ be an \sjfbcq.  If $R(x,x)$ or $R(x) \land S(x)$ is a
pattern of $q$, then $\countvals(q)$ is \shp-complete. Otherwise,
$\countvals(q)$ is in \fp.
\end{theorem}

\begin{proof}
The \shp-hardness part of the claim follows from the last two
propositions and from Lemma~\ref{lem:pattern-parsimonious}.  We explain why the
problems are in \shp\ right after this proof.  When~$q$ does not have any of
these two patterns then all variables have exactly one occurrence in~$q$. This
implies that every valuation~$\nu$ of~$D$ is such that~$\nu(D)$ satisfies~$q$
(except when one relation is empty, in which case the result is simply zero).
We can then compute the total number of valuations in FP by simply multiplying
the sizes of the domains of every null in~$D$.
\end{proof}

Notice that in this theorem, the membership of $\countvals(q)$ in \shp\ can be
established by considering a nondeterministic Turing Machine~$M$ that, with
input a non-uniform incomplete database~$D$, guesses a valuation~$\nu$ of~$D$
and verifies whether~$\nu(D)$ satisfies~$q$. This machine works in polynomial
time as we can verify whether~$\nu(D)$ satisfies~$q$ in polynomial time
(since~$q$ is a fixed FO query). Then given that $\countvals(q)(D)$ is equal to
the number of accepting runs of~$M$ with input~$D$, we conclude that
$\countvals(q)$ is in \shp. Obviously, the same idea works for codd tables,
that is, $\ccountvals(q)$ is also in \shp. But with this restriction we obtain
more tractable cases, as shown by the following dichotomy result.

\begin{theorem}[dichotomy]
\label{thm:countvals-sjfcqs-codd}
Let~$q$ be an \sjfbcq.  If $R(x) \land S(x)$ is a pattern of $q$, then
$\ccountvals(q)$ is \shp-complete. Otherwise, $\ccountvals(q)$ is in \fp.
\end{theorem}
\begin{proof}
We only need to prove the tractability claim, since hardness follows from
Proposition~\ref{prp:RxSx-hard} and Lemma~\ref{lem:pattern-parsimonious}.  We
will assume without loss of generality that~$D$ contains no constants, as we
can introduce a fresh null with domain~$\{c\}$ for every constant~$c$ appearing
in~$D$, and the result is again a Codd table, and this does not change the
output of the problem.  Let~$q$ be~$R_1({\bar x_1}) \land  \ldots \land
R_m({\bar x_m})$.  Observe that since~$q$ does not have~$R(x) \land  S(x)$ as a
pattern then any two atoms cannot have a variable in common.  But then,
since~$D$ is a Codd table we have $$\ccountvals(q)(D) \, = \, \prod_{i=1}^m
\ccountvals(R_i({\bar x_i}))(D(R_i)).$$ Hence it is enough to show how to
compute $\ccountvals(R_i({\bar x_i}))(D(R_i))$ for every~$1\leq i \leq m$.
Let~${\bar t_1},\ldots,{\bar t_n}$ be the tuples of~$D(R_i)$.  Let us
write~$\rho({\bar t_j})$ for the number of valuations of the nulls appearing
in~${\bar t_j}$ that do not match~${\bar x_i}$.  Clearly,
$\ccountvals(R_i({\bar x_i}))(D(R_i)) = \prod_{\bot \text{ appears in }D(R_i)}
|\dom(\bot)| - \prod_{j=1}^n \rho({\bar t_j})$, so we only have to show how to
compute~$\rho({\bar t_j})$ for~$1\leq j \leq n$.  Since we can easily compute
the total number of valuations of~${\bar t_j}$, it is enough to show how to
compute the number of valuations of~${\bar t_j}$ that match~${\bar x_i}$.  For
every variable~$x$ that appears in~${\bar x_i}$, compute the size of the
intersection of the domains of the corresponding nulls in~${\bar t_j}$, and denote it~$s_x$.
Then the number of valuations of~${\bar t_j}$ that match~${\bar x_i}$ is
simply~$\prod_{x\text{ appears in }{\bar x_i}} s_x$.  This concludes the proof.
\end{proof}

At this stage, we have completed the first column of
Table~\ref{tab:dichos-count-intro}, and we also know that~$R(x,x)$ is a hard
pattern in the uniform setting for naive tables (but not for Codd tables, by
Theorem~\ref{thm:countvals-sjfcqs-codd}).  In the next section, we treat the
uniform setting.

\subsection{The complexity on the uniform case}
\label{subsec:countvals-uniform}

In this section, we study the complexity of the problems
$\ucountvals(q)$ and~$\cucountvals(q)$, again providing dichotomy results in
both cases. 

\subsubsection{Na\"ive tables}

We start our investigation with the case of naive tables. In
Proposition~\ref{prp:Rxx-hard}, we already showed that $\ucountvals(R(x,x))$ is
\shp-hard. In the following proposition, we identify two other simple queries
for which this problem is still intractable.

\begin{proposition}
\label{prp:RxSxyTy-RxySxy-hard}
	$\ucountvals(R(x) \land S(x,y) \land T(y))$ and $\ucountvals(R(x,y)
\land S(x,y))$ are both \shp-hard. This holds even in the restricted setting in
which all nulls are interpreted over the same fixed domain~$\{0,1\}$.
\end{proposition}

\begin{proof}
We reduce both problems from the problem of counting the number of
independent sets in a graph (denoted by \sIS), which is
\shp-complete~\cite{provan1983complexity}.  We start
with~$\ucountvals(R(x)\land  S(x,y)\land  T(y))$.  Let $q = R(x)\land
S(x,y)\land  T(y)$ and~$G=(V,E)$ be a graph.  Then we define an incomplete
database $D$ as follows.  For every node~$v \in V$, we have a null~$\bot_v$,
and the uniform domain is~$\{0,1\}$.  For every edge~$\{u,v\} \in E$, we have
facts~$S(\bot_u,\bot_v)$ and~$S(\bot_v,\bot_u)$ in $D$.  Finally, we have
facts~$R(1)$ and~$T(1)$ in $D$.  For a valuation~$\nu$ of the nulls, consider
the corresponding subset~$S_\nu$ of nodes of~$G$, given by $S_\nu = \{t \in V
\mid \nu(\bot_t) = 1 \}$.  This is a bijection between the valuations of the
database and the node subsets of~$G$.  Moreover, we have that~$\nu(D)
\not\models q$ if and only if~$S_\nu$ is an independent set of~$G$.  Since the
total number of valuations of $D$ is $2^{|V|}$, we have that the number of
independent sets of $G$ is equal to~$2^{|V|} - \ucountvals(q)(D)$.  Hence, we
conclude that $\sIS \tr \ucountvals(q)$.  The idea is similar
for~$\ucountvals(R(x,y)\land  S(x,y))$: we encode the graph with the
relation~$S$ in the same way, and this time we add the fact~$R(1,1)$.
\end{proof}

As shown in the following result, it turns out that the three aforementioned
patterns are enough to fully characterize the complexity of counting valuations
for naive tables in the uniform setting.

\begin{toappendix}
\subsection{Proof of Theorem~\ref{thm:countvals-naive-uniform}}
	\label{apx:countvals-naive-uniform}
	In this section we prove the tractability claim of the following dichotomy theorem.
\end{toappendix}

\begin{theoremrep}[dichotomy]
\label{thm:countvals-naive-uniform}
	Let~$q$ be an \sjfbcq. If $R(x,x)$ or $R(x) \land S(x,y) \land T(y)$
or~$R(x,y) \land S(x,y)$ is a pattern of $q$, then $\ucountvals(q)$ is
\shp-complete.  Otherwise, $\ucountvals(q)$ is in \fp.
\end{theoremrep}

The \shp -completeness part of the claim follows directly from what we have
proved already.  Here, the most challenging part of the proof is actually the
tractability part. We only present a simple example to give an idea of the
proof technique, and defer the full proof to
Appendix~\ref{apx:countvals-naive-uniform}.  We will use the following
definition.  Given~$n,m \in \mathbb{N}$, let us write~$\surj_{n\rightarrow m}$
for the number of surjective functions from~$\{1,\ldots,n\}$
to~$\{1,\ldots,m\}$.  By an inclusion--exclusion argument, one can show
that~$\surj_{n\rightarrow m} = \sum_{i=0}^{m-1} (-1)^i \binom{m}{i} (m-i)^n$
(for instance, see~\cite{stackexchange_surjective}).  It is clear that this can
be computed in FP, when~$n$ and~$m$ are given in unary.

\begin{example}
\label{expl:RxSx}
	{\em Let~$q$ be the \sjfbcq\ $R(x)\land  S(x)$, and~$D$ be an
incomplete database over relations~$R,S$.  Notice that~$q$ does not have any of
the patterns mentioned in Theorem~\ref{thm:countvals-naive-uniform}. We will
show that~$\ucountvals(q)$ is in \fp.  Since~$q$ contains only two unary atoms
we can also assume without loss of generality that the input~$D$ is a Codd
table (otherwise all valuations are satisfying).
	
	Since we can compute in FP the total number of valuations, it is enough
to show how to compute the number of valuations of~$D$ that do not satisfy~$q$.
Let $\dom$ be the uniform domain,~$d$ be its size, $n_R$ (resp., $n_S$) be the
number of nulls in~$D(R)$ (resp., in $D(S)$) and~$C_R$ (resp., $C_S$) be the
set of constants occurring in~$D(R)$ (resp., in $D(S)$), with~$c_R$ (resp.,
$c_S$) its size.  We can assume without loss of generality that~$C_R \cap C_S =
\emptyset$, as otherwise all the valuations are satisfying, and this is
computable in PTIME.  Furthermore, we can also assume that~$C_R \cup C_S
\subseteq \dom$, since we can remove the constants that are not in~$\dom$, as
these can never match.
	
	Let~$M := \dom \setminus(C_R \cup C_S)$, and~$m$ its size (i.e., with
our assumptions we have $m = d-c_R-c_S$).  Fix some subsets~$M' \subseteq M$
and~$R' \subseteq C_R$. The quantity $\surj_{n_R \to |M'|+|R'|}$ then counts
the number of valuations of the nulls of~$D(R)$ that span exactly~$M'\cup R'$.
Moreover, letting~$\nu_R$ be a valuation of the nulls of~$D(R)$ that spans
exactly~$M'\cup R'$, the quantity $(d-c_R-|M'|)^{n_S}$ is the number of ways to
extend~$\nu_R$ into a valuation~$\nu$ of all the nulls of~$D$ so that~$\nu(D)
\not\models q$: indeed, every null of~$D(S)$ can take any value in~$\dom
\setminus(C_R \cup M')$.  The number of valuations of~$D$ that do not
satisfy~$q$ is then (keeping in mind that a null in~$D(R)$ cannot take a value
in~$C_S$): \[ \sum_{\substack{M' \subseteq M \\ R' \subseteq C_R}} \surj_{n_R
\to |M'|+|R'|} \times (d-c_R-|M'|)^{n_S}\] and since the summands only depends
on the sizes of~$M'$ and~$R'$, this is equal to \[ \sum_{\substack{0 \leq m'
\leq m \\ 0 \leq r' \leq c_R}} \binom{m}{m'} \binom{c_R}{r'} \surj_{n_R \to
m'+r'} \times (d-c_R-m')^{n_S}\] This last expression can clearly be computed
in PTIME.\footnote{Note that in the sum we do not need to specify that~$m'+r'
\leq n_R$, as when~$a <b$ we have~$\surj_{a \to b} = 0$.} } \qed
\end{example}

\begin{toappendix}
First, to characterize the queries that do not have these patterns, we will use
the notion of \emph{connectivity graph} of an \sjfbcq~$q$:

\begin{definition}
\label{def:connectivity-graph}
	Let~$q$ be an \sjfbcq.  The \emph{connectivity graph of~$q$} is the
graph~$G_q = (V,E)$ with labeled edges, where~$V$ is the set of atoms of~$q$,
and for every two atoms~$R(\bar{x_i}),S(\bar{y_i})$ of~$q$, if they share a
variable then we have an edge between the corresponding nodes of~$G_q$, that
edge being labeled with the variables in~$\bar{x_i} \cap \bar{y_i}$.
\end{definition}

\begin{example}
	\label{expl:connectivity-graph} {\em
Figure~\ref{fig:connectivity-graph} shows the connectivity graph of the query
\[R_1(x_1,x_1,y_1,t_1),R_2(x_1,y_1,t_2),S_1(x_2,t_3),S_2(x_2,t_4),\\S_3(x_2),T_1(x_3),T_2(x_3),T_3(x_3),T_4(x_3,t_5).\]}\qed
\end{example}

\begin{figure}[H]
	\centering
		\begin{tikzpicture}

\node (r1) at (0,2) {$R_1(x_1,x_1,y_1,t_1)$};
\node (r2) at (0,0) {$R_2(x_1,y_1,t_2)$};

\draw[black,thick] (r1) -- (r2) node[midway,right] {$x_1,y_1$};

\end{tikzpicture}
\qquad
\begin{tikzpicture}
\node (s1) at (0,2) {$S_1(x_2,t_3)$};
\node (s2) at (0,0) {$S_2(x_2,t_4)$};
\node (s3) at (2.3,1) {$S_3(x_2)$};

\draw[black,thick] (s1) -- (s2) node[midway,right] {$x_2$};
\draw[black,thick] (s1) -- (s3) node[midway,right,yshift=.3em] {$x_2$};;
\draw[black,thick] (s2) -- (s3) node[midway,right,yshift=-.5em] {$x_2$};;

\end{tikzpicture}
\qquad
\begin{tikzpicture}
\node (t1) at (0,2) {$T_1(x_3)$};
\node (t2) at (0,0) {$T_2(x_3)$};
\node (t3) at (2,0) {$T_3(x_3)$};
\node (t4) at (2,2) {$T_4(x_3,t_5)$};

\draw[black,thick] (t1) -- (t2) node[midway,left] {$x_3$};
\draw[black,thick] (t2) -- (t3) node[midway,below] {$x_3$};
\draw[black,thick] (t3) -- (t4) node[midway,right] {$x_3$};
\draw[black,thick] (t4) -- (t1) node[midway,above] {$x_3$};

\draw[black,thick] (t1) -- (t3) node[midway,xshift=-1.5em,yshift=.5em] {$x_3$};
\draw[black,thick] (t2) -- (t4) node[midway,xshift=1.5em,yshift=.5em] {$x_3$};

\end{tikzpicture}

	\caption{The connectivity graph~$G_q$ of the \sjfbcq~$q$ from Example~\ref{expl:connectivity-graph}.}
	\label{fig:connectivity-graph}
\end{figure}

The following is then readily observed:

\begin{lemma}
\label{lem:connectivity-graph-clique-naive}
	Let~$q$ be an \sjfbcq\ that does not contain any of the patterns
mentioned in Theorem~\ref{thm:countvals-naive-uniform}.  Then for every
connected component~$C$ of~$G_q$,~$C$ is a clique and there exists a variable
such that all edges of~$C$ are labeled by exactly that variable. 
\end{lemma}
\begin{proof}
	First, observe that every edge of~$G_q$ must be labeled by exactly one
variable, as otherwise the query~$q$ would contain the pattern~$R(x,y)\land
S(x,y)$.  Let~$C$ be a connected component of~$G_q$. Then we have:
\begin{itemize} 
\item $C$ is a clique. Indeed, assume by contradiction that~$C$
is not a clique. Then, since~$C$ is connected and is not a clique, we can find
3 nodes~$A_1(\overline{x}),A_2(\overline{x'}),A_3(\overline{x''})$ such
that~$A_1(\overline{x})$ is adjacent to~$A_2(\overline{x'})$,
$A_2(\overline{x'})$ is adjacent to~$A_3(\overline{x''})$,
and~$A_1(\overline{x})$ is not adjacent to~$A_3(\overline{x''})$. Let~$X$ be
$\overline{x} \cap \overline{x'}$ and~$Y$ be~$\overline{x'} \cap \overline{x''}$, i.e., the labels
on the two corresponding edges of~$C$. By definition of~$G_q$ and
since~$A_1(\overline{x})$ is not adjacent to~$A_3(\overline{x''})$, we must
have~$X \cap Y = \emptyset$.  But~$X$ and~$Y$ are not empty (again by
definition of~$G_q$), so by picking~$x$ in~$X$ and~$y$ in~$Y$ we see that~$q$
contains the pattern~$R(x)\land  S(x,y)\land  T(y)$, a contradiction.  
\item
There exists a variable that labels every edge of~$C$. Indeed, since every edge
of~$G_q$ is labeled by exactly one variable, and since~$C$ is a clique, if it
was not the case then again we could find the pattern~$R(x)\land  S(x,y)\land
T(y)$ in~$q$.  
\end{itemize}

This concludes the proof.
\end{proof}
	
	For instance, the query from Example~\ref{expl:connectivity-graph} does
not satisfy this criterion, since the edge in the first connected component
of~$G_q$ is labeled by two variables. However if we consider the query
$S_1(x_2,t_3),S_2(x_2,t_4),S_3(x_2),T_1(x_3),T_2(x_3),T_3(x_3),T_4(x_3,t_5)$
(i.e., we remove the first connected component), then it satisfies the
criterion.

We will also use the general fact that for an \sjfbcq~$q$, we can
assume wlog that~$q$ does not contain variables that occur only once:

\begin{lemma}
\label{lem:remove-ear-variables}
	Let~$q$ be an \sjfbcq, and let~$q'$ be the \sjfbcq\ obtained
from~$q$ by deleting all the variables that have only one occurrence in~$q$.
Then~$\ucountvals(q) \leq_{\mathrm{T}}^{\mathrm{p}} \ucountvals(q')$.
\end{lemma}
\begin{proof}
Let~$D$ be an incomplete database input of~$\ucountvals(q)$.  Let~$S$ be set of
nulls~$\bot$ such that:
\begin{itemize}
	\item $\bot$ occurs in a column corresponding to a variable that has been deleted; and
	\item $\bot$ does not occur in a column corresponding to a variable that has not been deleted.
\end{itemize}
Then, letting~$D'$ be the database obtained from~$D$ by projecting out the
columns corresponding to the deleted variables, it is clear that we
have~$\ucountvals(q)(D) = \ucountvals(q')(D') \times \prod_{\bot \in
S}|\dom(\bot)| $, where~$\dom$ is the uniform domain of the nulls.  We note
here that this lemma is also true in the non-uniform setting.
\end{proof}

By Lemma~\ref{lem:connectivity-graph-clique-naive} and
Lemma~\ref{lem:remove-ear-variables}, it is enough to show the tractability
of~$\ucountvals(q)$ when~$q$ is of the form~$C_1(x_1) \land  \ldots \land
C_m(x_m)$, where each~$C_i(x_i)$ is what we call a \emph{basic singleton
query}, i.e., is a conjunction of unary atoms over the same variable~$x_i$.  We
call such an \sjfbcq\ a \emph{conjunction of basic singletons}. For instance,
\[S_1(x_2),S_2(x_2),S_3(x_2),T_1(x_3),T_2(x_3),T_3(x_3),T_4(x_3)\] is such a
query, with~$m=2$.  We will use the following: 

\begin{lemma}
\label{lem:IE}
Let~$q=C_1(x_1) \land  \ldots \land  C_m(x_m)$ be a conjunction of basic
singletons \sjfbcq, and let~$D$ be an incomplete database.  For~$S
\subseteq [m]$, we define~$N_S(D) \defeq |\{\nu \text{ valuation of } D \mid
\nu(D) \not\models \bigvee_{i \in S} C_i(x_i)\}|$.  Then we
have~$\ucountvals(q)(D) = \sum_{S \subseteq [m]} (-1)^{|S|} N_S(D)$.
\end{lemma}
\begin{proof}
	Direct, by inclusion--exclusion.
\end{proof}

Hence, and remembering that we consider data complexity, it is enough to show
how to compute~$N_S(D)$ for every~$S \subseteq [m]$.  The main difficulties in
computing~$N_S(D)$ is that the relations can have nulls in common (since we
consider naive tables), and that they may also have constants; this makes it
technically painful to express a closed-form expression for~$N_S(D)$.  We
explain how to do it next, thus finishing the proof of
Theorem~\ref{thm:countvals-naive-uniform}.

\begin{proposition}
	\label{prp:compute-NSD}
	Let~$q=C_1(x_1) \land  \ldots \land  C_m(x_m)$ be a conjunction of
basic singletons \sjfbcq\ and~$S \subseteq [m]$.  Then, given an
incomplete database~$D$ as input, we can compute~$N_S(D)$ in polynomial time.
\end{proposition}
\begin{proof}
First, observe that to compute~$N_S(D)$ we can assume without loss of
generality that the input database~$D$ only contains facts over relation names
that occur in some~$C_i(x_i)$, for~$i\in S$.  Indeed, $N_S(D)$ counts the
valuations~$\nu$ of~$D$ that do not satisfy any of the~$C_i(x_i)$ for~$i\in S$,
so that for any~$j \notin S$ we do not care if~$\nu$ satisfies~$C_j(x_j)$ or
not; hence, we could simply multiply the result by the appropriate factor.
Therefore, we can assume that~$S$ is~$[m]$.  We now need to fix some notation.
Let us write the conjunction of basic singleton \sjfbcq~$q$ as \[R_1(x_1)
\land  \ldots \land  R_{m_1}(x_1) \land  R_{m_1 + 1}(x_2) \land  \ldots \land
R_{m_1 + m_2}(x_2) \land  \ldots \land  R_{\sum_{i=1}^{m-1} m_i}(x_m) \land
\ldots \land  R_{\sum_{i=1}^m m_i}(x_m)\] and let~$K$ be the number of atoms
in~$q$, that is, $K \defeq \sum_{i=1}^m m_i$.  Let~$\dom$ be the uniform domain
of the nulls occurring in~$D$ and~$d$ its size.  For~$\mathbf{s} \subseteq
[K]$, we write~$C_\mathbf{s}$ the set of constants that occur in each of the
relations~$D(R_i)$ for~$i\in \mathbf{s}$ but in none of the others, and
write~$c_\mathbf{s}$ the size of that set. We call such a set a \emph{block} of
constants.  Similarly for the nulls, we write~$N_\mathbf{s}$ the set of nulls
that occur in each of the relations~$D(R_i)$ for~$i\in \mathbf{s}$ but in none
of the others (and we call this a block of nulls), and~$n_\mathbf{s}$ for its
size. We can assume wlog that:

\begin{description}
	\item[(a)] For every~$1 \leq i \leq m$, there is no constant that occurs in every~$D(R)$ for~$R$ a relation name in~$C_i(x_i)$.
		Indeed otherwise any valuation would satisfy~$C_i(x_i)$, thus~$N_{[m]}(D)$ would simply be~$0$.
	\item[(b)] Every constant~$c$ appearing in~$D$ is in~$\dom$. Indeed otherwise, with the last item, this constant would have no chance to
		be part of a match, so we could simply remove it (i.e., remove all tuples of the form~$R(c)$ from~$D$).
\end{description}
For a subset~$A \subseteq \dom$, let us write~$A^\complement \defeq \dom
\setminus A$.  Finally, for a set~$Z=\{A_1,\ldots,A_l\}$ of subsets of~$\dom$,
we denote by~$\I(Z)$ the set
\[\I(Z) \defeq \{ \bigcap_{i=1}^l B_i \mid (B_1,\ldots,B_l) \in \{A_1,A_1^\complement\} \times \ldots \times \{A_l,A_l^\complement\} \} \]

We now explain informally how we can compute~$N_{[m]}(D)$.
Let~$L=\mathbf{s}_1,\ldots,\mathbf{s}_{2^K}$ be an arbitrary linear order of
the set of subsets of~$[K]$.  We will define by induction on~$i\in [2^K]$ an
expression computing~$N_{[m]}(D)$, which will be a nested sum of the form

\begin{equation}
\label{eq:huge-sum}
	\sum_{\text{sum}_{\mathbf{s}_1}} f_{\mathbf{s}_1} \times \bigg( \sum_{\text{sum}_{\mathbf{s}_2}} f_{\mathbf{s}_2} \times 
	\big( \ldots (\sum_{\text{sum}_{\mathbf{s}_{2^K}}} f_{\mathbf{s}_{2^K}}) \ldots \big) \bigg)
\end{equation}

where each~$\text{sum}_{\mathbf{s}_i}$ sums over the possible
images~$A_{\mathbf{s}_i}$ of the nulls in~$N_{\mathbf{s}_i}$ by a valuation,
and~$f_{\mathbf{s}_i}$ will simply be~$\surj_{n_{\mathbf{s}_i} \to
a_{\mathbf{s}_i}}$, where~$a_{\mathbf{s}_i} \defeq |A_{\mathbf{s}_i}|$, i.e.,
the number of valuations~$\nu$ of~$N_{\mathbf{s}_i}$ with image
exactly~$A_{\mathbf{s}_i}$. But there are two technicalities:
\begin{itemize}
	\item First, we need to ensure that each basic singleton query~$C_i(x_i)$ of~$q$ will not be satisfied. In order to 
		do that, $\text{sum}_{\mathbf{s}_i}$ will actually sum over all
		the possible partitions~$(B_{\mathbf{s_i}}^1,\ldots,B_{\mathbf{s_i}}^{|\I(Z_{i-1}|)})$ of~$A_{\mathbf{s}_i}$,
		where each of the~$B_{\mathbf{s}_i}^j$ 
		is included in one of the sets in~$\I(Z_{i-1})$, where~$Z_{i-1}$ contains all the blocks of constants and all the other~$B_{\mathbf{s_j}}^r$ for~$j < i$.
		We iteratively build that sum from the outside to the inside, starting with~$Z_0 \defeq \{\dom\} \cup \{C_{\mathbf{s}} \mid \mathbf{s} \subseteq [K]\}$.
		This will allow us to avoid summing over the~$B_{\mathbf{s_i}}^j$ that would render a basic singleton query true.
	\item Second, as is, such a sum is obviously not going to be computable in PTIME, as we are summing over subsets of~$\dom$.
		To fix this, observe that the value of the subsum for~$\mathbf{s}_i$ actually only depends on the \emph{sizes} of the sets in~$Z_{i-1}$.
		Hence, iterating from the outside to the inside, whenever~$\text{sum}_{\mathbf{s}_i}$ contains a
		sum of the form, say, $B_{\mathbf{s_i}}^k \subseteq B_{\mathbf{s_j}}^{k'}$ for~$j<i$,
		we can replace this with a sum over~$0 \leq b_{\mathbf{s_i}}^k \leq b_{\mathbf{s_j}}^{k'}$, and add to~$f_{\mathbf{s}_i}$ a factor
		of~$\binom{b_{\mathbf{s_j}}^{k'}}{b_{\mathbf{s_i}}^k}$.
		Now, because of how~$Z_0$ is defined, and because of how~$\I$ works, all the initial numbers in the first sum are either 
		$|\dom \setminus \bigcup_{i=1}^K C_{\{i\}}|$ or	one of the numbers~$c_{\mathbf{s}}$
		for~$\mathbf{s} \subseteq [K]$. These can all be computed in polynomial time.
\end{itemize}

The resulting expression then indeed evaluates to~$N_{[m]}(D)$, and is in a
form that allows us to directly compute it in polynomial time (but
non-elementary in the query).  This concludes the proof of
Proposition~\ref{prp:compute-NSD}.
\end{proof}

\end{toappendix}

\subsubsection{Codd tables}
We conclude this section by turning our attention to the 
case of Codd tables.
Notice that none of the results proved so far 
provides a hard pattern in this case.
We identify in the following proposition a simple query for which the problem is intractable.

\begin{proposition}
\label{prp:RxSxyTy-hard-codd}
	$\cucountvals(R(x) \land S(x,y) \land T(y))$ is \shp-hard.
\end{proposition}
\begin{proof}
We reduce from the problem of counting the number of independent sets of a
bipartite (simple) graph, written~\#BIS, which is
\shp-hard~\cite{provan1983complexity}.  Let~$G = (X \sqcup Y, E)$ be a
bipartite graph.  Without loss of generality, we can assume that~$|X|=|Y|=n$;
indeed, if~$|X| < |Y|$ then we could simply add~$|Y|-|X|$ isolated nodes to
complete the graph, which simply multiplies the number of independent sets
by~$2^{|Y|-|X|}$.  Also, observe that counting the number of independent sets
of~$G$ is the same as counting the number of pairs~$(S_1,S_2)$ with~$S_1
\subseteq X, S_2 \subseteq Y$, such that~$(S_1 \times S_2) \cap E = \emptyset$.
We will call such a pair an \emph{independent pair}.  For~$0 \leq i,j \leq n$,
let~$Z_{i,j}$ be the number of independent pairs~$(S_1,S_2)$ such that~$|S_1| =
i$ and~$|S_2| = j$.  It is clear that ($\star$) the number of independent sets
of~$G$ is then~$\#\mathrm{BIS}(G) = \sum_{0 \leq i,j \leq n} Z_{i,j}$.  The
idea of the reduction is to construct in polynomial time $(n+1)^2$ incomplete
databases~$D_{a,b}$ for~$0 \leq a,b \leq n$ such that, letting~$C_{a,b}$ be the
number of valuations~$\nu$ of~$D_{a,b}$ with~$\nu(D_{a,b}) \not \models
R(x)\land  S(x,y)\land  T(y)$, the values of the variables~$Z_{i,j}$
and~$C_{i,j}$ form a linear system of equations~$\mA \mZ = \mC$, with~$\mA$ an
invertible matrix.  This will allow us, using $(n+1)^2$ calls to an oracle for
$\cucountvals(R(x)\land  S(x,y)\land  T(y))$, to recover the~$Z_{i,j}$ values,
and then to compute~$\#\mathrm{BIS}(G)$ using ($\star$).  We now explain how we
construct~$D_{a,b}$ from~$G$ for~$0 \leq a,b \leq n$, and define~$\mA$.  First,
we fix an arbitrary linear order~$x_1,\ldots,x_n$ of~$X$, and similarly
$y_1,\ldots,y_n$ for~$Y$.  The database~$D_{a,b}$ has constants~$a_i$ for~$1
\leq i \leq n$, and has a fact~$S(a_i,a_j)$ whenever~$(x_i,y_j) \in E$. It has
nulls~$\bot_1,\ldots,\bot_a$ and facts~$R(\bot_i)$ for~$1 \leq i \leq a$
(if~$a=0$ there are no such nulls and facts), and
nulls~$\bot_1',\ldots,\bot'_b$ and facts~$T(\bot'_i)$ for~$1 \leq i \leq b$; in
particular, this is a Codd table.  The uniform domain of the nulls is~$\{a_i
\mid 1 \leq i \leq n\}$.  Given a valuation~$\nu$ of~$D_{a,b}$, let $P(\nu)$ be
the pair of subsets of~$V$ defined by \[P(\nu) \defeq (\{x_i \mid \exists 1
\leq k \leq a \text{ s.t. } \nu(\bot_k)=a_i\}, \\ \{y_i \mid \exists 1 \leq k
\leq b \text{ s.t. } \nu(\bot'_k)=a_i \})\] One can then easily check that the
following two claims hold:
\begin{itemize}
	\item For every valuation~$\nu$ of~$D_{a,b}$, we have that~$\nu(D_{a,b}) \not\models R(x)\land  S(x,y)\land  T(y)$ iff~$P(\nu)$ is an independent
		pair of~$G$;\footnote{This observation, and in fact the idea of reducing from \#BIS, is due to Antoine Amarilli.}
	\item For every independent pair~$(S_1,S_2)$ of~$G$, there are
		exactly $\surj_{a \to |S_1|} \times \surj_{b \to |S_2|}$ valuations~$\nu$ such that~$P(\nu) = (S_1,S_2)$.
\end{itemize}

But then, we have $C_{a,b} = \sum_{0 \leq i,j \leq n} (\surj_{a \to i} \times
\surj_{b \to j}) Z_{i,j}$.  In other words, we have the linear system of
equations~$\mA \mZ = \mC$, where~$\mA$ is the~$(n+1)^2 \times (n+1)^2$ matrix
defined by $\mA_{(a,b),(i,j)} \defeq \surj_{a \to i} \times \surj_{b \to j}$.
This matrix is the Kronecker product~$\mA' \otimes \mA' $ of the~$(n+1) \times
(n+1)$ matrix with entries~$\mA'_{a,i} \defeq \surj_{a \to i}$.  Since~$\mA'$
is a triangular matrix with non-zero coefficients on the diagonal, it is
invertible, hence so is~$\mA$, which concludes the proof.  
\end{proof}

Note that in Proposition~\ref{prp:RxSxyTy-RxySxy-hard}, we proved that
$\ucountvals(R(x) \land S(x,y) \land T(y))$ is \shp-hard in the general case
where naive tables are allowed. Hence, the hardness of that query for naive
tables was in fact a consequence of Proposition~\ref{prp:RxSxyTy-hard-codd}.
However, we decided to provide a separate proof for
Proposition~\ref{prp:RxSxyTy-RxySxy-hard}, because in this case intractability
holds already when nulls are interpreted over the fixed domain~$\{0,1\}$,
whereas we do not know if this is true for Codd tables.\\

The second pattern that we show is hard for~$\ucountvals$ for Codd tables is
$R(x,y) \land S(x,y)$. Again, notice that we already showed this query to be
hard in the case of naive tables (as Proposition~\ref{prp:RxSxyTy-RxySxy-hard}),
even for a fixed domain of~$\{0,1\}$. In the case of Codd tables, hardness
still holds, but the proof is more complicated and uses domains of unbounded
size (which is why we provide separate proofs). We show:

\begin{proposition}
\label{prp:Rxy-Sxy-hard-codd}
  $\cucountvals(R(x,y) \land S(x,y))$ is \shp-hard.
\end{proposition}
\begin{proof}
	Let~$q$ be the query~$R(x,y) \land S(x,y)$.  We reduce from the
problem of counting the number of matchings of a $2$-$3$--regular bipartite
graph, which is~\shp-complete~\cite{xia2007computational,cai2012holographic}.
Let~$G=(A\sqcup B,E)$ be a $2$-$3$--regular bipartite graph, with the nodes
in~$A$ having degree~$3$ and those in~$B$ having degree~$2$, and let~$n_A
\defeq |A|$, and~$n_B \defeq |B|$. Notice that since~$G$ is $2$-$3$--regular we
have that~$n_B = \frac{3n_A}{2}$.  We say that a set~$S \subseteq E$ of edges
of~$G$ is an \emph{$A$-semimatching} if every node~$a$ of~$A$ is adjacent to at
most~$1$ edge of~$S$; formally, for every~$a\in A$ we have $|\{e \in S \mid a
\in e\}| \leq 1$. The \emph{type} of an $A$-semimatching~$S$ is the number~$c$
of nodes in~$B$ that are adjacent to exactly~$2$ edges of~$S$; formally, $c
\defeq |\{b \in B \mid |\{ e \in S \mid b \in e\}|=2\}|$. For~$0 \leq c \leq
n_B$, we write~$T_c$ for the number of~$A$-semimatchings of~$G$ of type~$c$.
Observe then that~$T_0$ is simply the number of matchings of~$G$. The idea of
the reduction is then as follows. We will construct databases~$D_k$ for~$0 \leq
k \leq n_B$ such that, letting~$C_k$ be the number of valuations~$\nu$ of~$D_k$
such that~$\nu(D_k) \not\models q$, the values of the variables~$T_k$ and~$C_k$
form a system of linear equations~$\mA \mT = \mC$, with~$\mA$ and invertible
matrix. This will allow us, using~$n_B+1$ calls to and oracle
for~$\cucountvals(q)$, to recover the values~$T_k$, and thus to obtain~$T_0$ in
polynomial time, that is, the number of matchings of~$G$. We now explain how to
construct the database~$D_k$ for~$0 \leq k \leq n_B$. In what follows we use
the convention that~$\{1,\ldots,k\} = \emptyset$ when~$k=0$. The database~$D_k$
contains the following facts:
\begin{enumerate}
	\item One fact~$R(a,\bot_a)$ for every~$a\in A$;
	\item One fact~$S(\bot_b,b)$ for every~$b\in B$;
	\item One fact~$S(\bot_{a'},a')$ for every~$a\in A$ where~$a'$ is a fresh constant. In particular, observe that there are~$n_A$ such facts. In what follows we will write~$A' \defeq \{a' \mid a\in A\}$.
	\item One fact~$S(a_1,a_2')$ for every~$(a_1,a_2) \in A \times A$ with~$a_1 \neq a_2$;
	\item One fact~$S(a,i)$ for every~$(a,i)\in A \times \{1,\ldots,k\}$;
	\item One fact~$S(u,v)$ for every~$(u,v) \in (A \cup B)^2$ such that~$\{u,v\}$ is not in~$E$.
	\item One fact~$R(u,v)$ for every~$(u,v) \in (A' \cup B)^2$.
\end{enumerate}

And finally, the (uniform) domain for all the nulls is~$\dom \defeq A \cup B
\cup A' \cup \{1,\ldots,k\}$. Note that this is indeed a Codd database. Now,
let us compute~$C_k$, the number of valuations~$\nu$ of~$D_k$ such
that~$\nu(D_k) \not\models q$.  For such a valuation~$\nu$ of~$D_k$ such
that~$\nu(D_k) \not\models q$, observe that ($\star$) for every~$a \in A$ it
holds that~$\nu(\bot_a)$ is either~$a'$ or is one of the nodes in~$B$ that is a
neighbor of~$a$; this is because otherwise, the facts from ($4$-$6$) would make
the query be satisfied.  Then, for a valuation~$\nu$ of~$D_k$ such
that~$\nu(D_k) \not\models q$, let us define~$\sm(\nu) \defeq \{\{a,b\} \mid
(a,b)\in (A \times B) \cap R(\nu(D_k))\}$. Because of ($\star$), observe
that~$\sm(\nu)$ is a subset of~$E$, and that it is in fact an~$A$-semimatching.
We can then partition the valuations~$\nu$ of~$D_k$ with~$\nu(D_k) \not\models
q$ according to the type of~$\sm(\nu)$ as follows:
 \begin{equation}
  \label{eq:semimatchings}
 \{\nu \mid \nu \text{ is a valuation of $D_k$ with } \nu(D_k) \not\models q\} = \bigsqcup_{0\leq c \leq n_B} \,\,\,\bigsqcup_{\substack{S: \text{ $S$ is an $A$-semimatching}\\ \text{of~$G$ of type } c}} \,\, \bigsqcup_{\substack{\nu: \text{ $\nu$ valuation of $D_k$}\\ \nu(D_k)\not\models q\\ \sm(\nu)=S}} \{\nu\}.
 \end{equation}

Fix an~$A$-semimatching~$S$ of~$G$ of type~$c \in \{0,\ldots,n_B\}$, and let us
count how many valuations~$\nu$ of~$D_k$ there are such that~$\nu(D_k)
\not\models q$ and~$\sm(\nu)=S$.  First of all, observe that for such a
valuation~$\nu$, the value of~$\nu(\bot_a)$ for every~$a\in A$ is forced: it
is~$a'$ if~$a$ is adjacent to no edge of~$S$, and otherwise it is~$b$ for the
unique~$\{a,b\} \in S$. Therefore we have to count how many possibilities there
are for the remaining nulls, those of the form~$\bot_{a'}$ and~$\bot_b$ from
facts ($2$-$3$).  We have:
\begin{itemize}
	\item For a null~$\bot_b$ such that~$b$ is adjacent to two edges of~$S$, there
are~$n_A+k-2$ possible values in order to not satisfy the query. Note that there
are~$c$ such nulls~$\bot_b$. 
	\item For a
null~$\bot_{a'}$ such that~$a$ is not adjacent to an edge in~$S$ (so that we
know that~$\nu(\bot_a)=a'$), there are~$n_A+k-1$ possible values in order to not
satisfy the query. For a null~$\bot_b$ such that~$b$ is adjacent to exactly one
edge~$\{a,b\}$ of~$S$ (i.e., we have~$\nu(\bot_a)=b$) there are again~$n_A+k-1$
possible values to not satisfy the query. Observe that in total there are~$n_A
-2c$ such nulls~$\bot_{a'}$ or~$\bot_b$, because~$S$ is an~$A$-semimatching. 
	\item For a null~$\bot_{a'}$ such that~$a$ is adjacent to an edge in
~$S$ (so that we know that~$\nu(\bot_a)\neq a'$) there are~$n_A+k$
possibilities, and similarly for a null~$\bot_b$ such that~$b$ is not adjacent
to an edge in~$S$ there are~$n_A+k$ possible values. By the previous two items,
in total there are~$n_A+n_B - c - (n_A-2c) = n_B+c$ such nulls~$\bot_{a'}$
or~$\bot_b$.
\end{itemize}

Therefore, there are exactly~$(n_A+k-2)^c (n_A+k-1)^{n_A-2c} (n_A+k)^{n_B+c}$
valuations~$\nu$ of~$D_k$ such that~$\nu(D_k) \not\models q$ and~$\sm(\nu)=S$.
Since this depends only on the type of the~$A$-semimatching~$S$ (and not on~$S$
itself), we obtain from Equation~\ref{eq:semimatchings} that 

\begin{align*}C_k &= |\{\nu \mid \nu \text{ is a valuation of $D_k$ with } \nu(D_k) \not\models q\}|\\
&= \sum_{0 \leq c \leq n_B} T_c \times (n_A+k-2)^c (n_A+k-1)^{n_A-2c} (n_A+k)^{n_B+c}.
\end{align*}

 That is, we have the linear system of equations~$\mA \mT = \mC$,
with~$\mA_{k,c} = (n_A+k-2)^c (n_A+k-1)^{n_A-2c} (n_A+k)^{n_B+c}$ for~$0\leq
c,k \leq n_B$.  But observe that we have~$\mA = \mD \mV$, with~$\mD$ being the
diagonal~$(n_B+1)\times (n_B+1)$ matrix with
entries~$(n_A+k-2)^{n_A}(n_A+k)^{n_B}$, and~$\mV$ being the~$(n_B+1)\times
(n_B+1)$ Vandermonde matrix with coefficients~$\frac{(n_A+k-2)
(n_A+k)}{(n_A+k-1)^2}$ for~$0\leq k \leq n_B$.  Hence to show that~$\mA$ is
invertible, we only need to argue that the coefficients of this Vandermonde
matrix~$\mV$ are all distinct. But for the function~$f_{n_A}(x) \defeq
\frac{(n_A+x-2) (n_A+x)}{(n_A+x-1)^2}$, one can check that~$f'_{n_A}(x) =
\frac{2}{(n_A+x-1)^3}$, so that~$f_{n_A}$ is strictly increasing on~$[0,n_B]$
(we assume~$n_A \geq 2$ without loss of generality), so that all the
coefficients are indeed distinct.  This concludes the proof.
\end{proof}

As we show next, the patterns from Propositions~\ref{prp:Rxy-Sxy-hard-codd}
and~\ref{prp:RxSxyTy-hard-codd} are the only hard patterns for~$\cucountvals$.
The following is then the last dichotomy of this section.

\begin{theorem}[dichotomy]
\label{thm:countvals-codd-uniform}
	Let~$q$ be an \sjfbcq. If $R(x) \land S(x,y) \land T(y)$ or~$R(x,y)
\land S(x,y)$ is a pattern of $q$, then $\cucountvals(q)$ is \shp-complete.
Otherwise, $\cucountvals(q)$ is in \fp.
\end{theorem}

We only need to show the tractability part of that claim, as hardness follows
from Propositions~\ref{prp:Rxy-Sxy-hard-codd} and~\ref{prp:RxSxyTy-hard-codd}
and Lemma~\ref{lem:pattern-parsimonious}.  First of all, observe that we can
assume without loss of generality that the \sjfbcq~$q$ is connected. This is
because~$q$ has no self-join and the database~$D$ is Codd, so,
letting~$q_1,\ldots,q_t$ be the connected components of~$q$, and letting~$D_i$
be the database~$D$ restricted to the relations appearing in~$q_i$, we have
that

\[\cucountvals(q)(D) = \prod_{i=1}^t \cucountvals(q_i)(D_i).\]

Second, notice that, for a connected \sjfbcq~$q$, not containing any of these
two patterns is equivalent to the following: there exists a variable~$x$
such that all atoms of~$q$ contain variable~$x$, and for any two atoms
of~$q$, the only variable that they have in common is~$x$. In other words,~$x$
is in every atom and every other variable occurs in only one atom. For
instance~$R_1(x,y,y)\land R_2(x,x,z,z,z,u,u) \land R_3(x,x,v,t,t)$ is such a
query. 
We now provide an example of a connected query with only two atoms. 

\begin{example}
{\em We consider the query~$q=R(x,x,y,y)\land S(x,x)$. Let~$D$ be an incomplete Codd
database over relations~$R,S$, with uniform domain~$\dom$ of size~$d$. In this
proof we will use symbols~$a,a_1,a_2,\ldots$ to denote a constant or a null,
and symbols~$c,c_1,c_2,\ldots$ to denote constants. Moreover, unless stated
otherwise, these symbols can refer to the same constant (but not to the same
null because~$D$ is a Codd database). A fact that contains only constants is
called a \emph{ground fact}. Furthermore, since that database is Codd we will
always represent a null with~$\bot$, being understood that they are all distinct.

We start with a few simplifications.  First, we assume wlog that~$D$ does not
contain ground facts that already satisfy the query. Second, we assume wlog
that~$D$ does not contain facts of the form~$R(a,a',c,c')$ or~$R(c,c',a,a')$
or~$S(c,c')$ where~$c$ and~$c'$ are distinct constants. Indeed, because~$D$ is
Codd and such a fact~$f$ can never be part of a match, we could simply
remove~$f$ from~$D$ and multiply the end result by the appropriate value
(namely,~$d^t$ where~$t$ is the number of nulls of~$f$). Last, we assume wlog
that for any fact of the form~$R(a,a',c,\bot)$ or~$R(a,a',\bot,c)$
or~$R(\bot,c,a,a')$ or~$R(c,\bot,a,a')$ or~$S(c,\bot)$ or~$S(\bot,c)$, we
have~$c\in \dom$; indeed otherwise, those facts can never be part of a match
and we could again remove them and multiply the result by the appropriate
value.

Next, we need to introduce some notation.  For a fact~$f$ of~$D$, the
\emph{type} of~$f$ is the word in~$\{0,1\}^{\arity(f)}$ that has a~$1$ in
position~$i$ iff the~$i$-th element of~$f$ is a null. For instance the type
of~$R(\bot,\bot,c,\bot)$ is~$1101$ and that of~$S(c,c)$ is~$00$.  Observe that
there are a fixed number of possible types, because the query is fixed.   For a
constant~$c$ and fact~$f$ of~$D$, we say that \emph{$f$ is $c$-determined}
if~$f$ contains the constant~$c$ at a position for variable~$x$. For
instance~$R(\bot,c,c',\bot)$ is~$c$-determined and so is~$S(c,c)$.  A fact~$f$
that contains only nulls on the positions for~$x$ is called \emph{free}. With
the simplifications of the last paragraph, a fact of~$D$ is either free or it
is~$c$-determined for a unique constant~$c$.  Let~$\mathbf{t}$ be a type of~$D$
(for~$R$ or for~$S$).  We say that~$\mathbf{t}$ is \emph{free} if it has
only~$1$s in the positions corresponding to variable~$x$; in other words, if it
is the type of a free fact. Otherwise we say that~$\mathbf{t}$ is determined.
For a constant~$c$ and determined type~$\mathbf{t}$, we
write~$n_{R,c,\mathbf{t}}$ (resp.,~$n_{R,c,\mathbf{t}}$) the number
of~$R$-facts (resp., of~$S$-facts) of~$D$ that are~$c$-determined of
type~$\mathbf{t}$.  For a free type~$\mathbf{t}$ of~$R$ (resp., of~$S$) we
write~$F_{R,\mathbf{t}}$ (resp., $F_{S,\mathbf{t}}$) for the set of
free~$R$-facts (resp., of free~$S$-facts) and~$\ff_{R,\mathbf{t}}$
(resp.,$\ff_{\mathbf{t}}$) for its size.  Let~$f$ be a fact that
is~$c$-determined.  We write~$\alpha_{f}$ for the number of valuations~$\nu$ of
the nulls in~$f$ such that~$f$ matches the corresponding atom in~$q$.  For
instance if~$f$ is~$R(\bot,c,c',\bot)$ or again~$R(c,c,c',c')$
then~$\alpha_{f}=1$, while if~$f$ is~$R(c,c,\bot,\bot)$
then~$\alpha_{f}=d$.\footnote{When~$f$ is a ground fact -- such
as~$R(c,c,c',c')$ -- we recall the mathematical convention that there exists a
unique function with emtpy domain, hence a unique valuation of the nulls
of~$f$.}  Similarly we let~$\beta_{f}$ denote the number of valuations~$\nu$ of
the nulls in~$f$ such that~$f$ does not match the corresponding atom in~$q$;
which is then equal to~$d^t - \alpha_{f}$, for~$t$ the number of nulls in~$f$.
Observe that $\alpha_{f}$ and~$\beta_f$ depend only on the type~$\mathbf{t}$
of~$f$.  Hence we will write~$\alpha_\mathbf{t}$ and~$\beta_\mathbf{t}$
instead.  Let~$f$ be a free fact. We let~$\alpha_f$ be the number of valuations
of the nulls in~$f$ that are not on a position for variable~$x$ that match the
corresponding part of the atom in~$q$. For instance if~$f$
is~$R(\bot,\bot,c',\bot)$ then~$\alpha_f=1$ and if~$f$
is~$R(\bot,\bot,\bot,\bot)$ then~$\alpha_f=d$.  We also define~$\beta_f \defeq
d^t - d \alpha_f$ where~$t$ is the number of nulls, which correspond to the
number of valuations of the nulls in~$f$ such that the ground fact obtained
does not match its corresponding atom.  Again, since $\alpha_{f}$ and~$\beta_f$
depend only on the free type~$\mathbf{t}$ of~$f$, we will
write~$\alpha_\mathbf{t}$ and~$\beta_\mathbf{t}$ instead.  It is clear that we
can compute all values~$\alpha_{f,c}$ and~$\beta_{f,c}$, for every determined
fact~$f$ and constant~$c$ of~$D$, and values~$\alpha_{f}$ for every free fact
in polynomial time.    Last, we fix a linear order~$c_1,\ldots,c_n$ on the
constants~$c_i$ such that there exists a fact of~$D$ that is~$c_i$-determined.

Next, we explain how we can compute in~\fp\ the number of valuations of~$D$
that \emph{do not} satisfy~$q$.  To do so, we will first define some
quantities, and then show that we can compute these quantities in polynomial
time using a dynamic programming approach. One of these quantities will be the
number of valuations of~$D$ that do not satisfy~$q$, which is what we want to
compute. The quantities that we will define are of the form~$V(\params)$,
where~$\params$ consist of the following parameters ($\dagger$):
\begin{enumerate}[(a)]
	\item one parameter~$v$ whose range is~$0\ldots n$;
	\item one parameter~$q_{R,\mathbf{t}}$ for every free type~$\mathbf{t} \in \{0,1\}^4$ of~$R$, with range~$0\ldots \ff_{R,\mathbf{t}}$;
	\item one parameter~$r_R$ with range~$0,\ldots,n$; 
	\item one parameter~$r_S$ with range~$0,\ldots,r_R$; 
\end{enumerate}

We now explain what the quantity~$V(\params)$ with these parameters represent. To that end, we define the incomplete database~$D(\params)$ to be the database that contains only
\begin{itemize}
	\item all of the facts that are~$c_i$-determined for~$1 \leq i \leq v$
(if~$v=0$ then~$D_v$ contains no determined facts); \item for every free
type~$\mathbf{t}$ of~$R$, $D(\params)$ contains $q_{R,\mathbf{t}}$ free
$R$-facts of~$D$ of type~$\mathbf{t}$ (it doesn't matter which ones, say the
first~$q_{R,\mathbf{t}}$ ones);
	\item $D(\params)$ contains all the free~$S$-facts of~$D$.
\end{itemize}

(Note that parameters (c) and (d) are not used to define the database but we still write~$D(\params)$.)

The quantity~$V(\params)$ is then defined to be the number of valuations~$\nu$
of database~$D(\params)$ such that~$\nu(D(\params))\not\models q$ and such that
the following holds: (1) for every~$ 1 \leq i \leq r_S$ and free fact~$f$
of~$R$ or of~$S$, the ground fact~$\nu(f)$ does not match the corresponding
atom in~$q$ with variable~$x$ mapped to~$c_i$ (this condition is empty if~$r_S
= 0$); and (2) for every~$r_S + 1 \leq i \leq r_R$ and every free fact~$f$
of~$R$, the ground fact~$\nu(f)$ does not match~$R(x,x,y,y)$ with variable~$x$
mapped to~$c_i$ (this condition is empty if~$r_R = r_S$).  By definition we then have that,
when~$v=n$, $q_{R,\mathbf{t}}=\ff_{R,\mathbf{t}}$ for every free
type~$\mathbf{t}$ of~$R$ -- so that~$D(\params)$ is equal to~$D$ --, $r_R=0$
and~$r_S=0$ then~$V(\params)$ is then equal to~$\cucountvals(q)(D)$.  Observe
that there are a fixed number of parameters in~$\params$ (because the query is
fixed), and that the possible values of these parameters are polynomial in the
size of the input. We explain how compute the values~$V(\params)$ by dynamic
programming. 

\paragraph*{Base case.}
 Our base case will be when~$v$ is equal to zero (and the other parameters are
arbitrary as in ($\dagger$)); in other words, the database~$D(\params)$
does not contain any determined facts, it contains only free facts. We have to
compute the number of valuations of~$D(\params)$ that do not satisfy the query
and such that~(1) and (2) hold. 
We do so as follows.  We guess a subset~$S_R$ of~$\dom \setminus
\{c_1,\ldots,c_{r_R}\}$; this will be the set of constants~$c$ such
that~$R(c,c,c',c') \in \nu(D(\params))$ for some~$c'$.  We also guess a subset~$S_S$
of~$\dom \setminus \big( \{c_1,\ldots,c_{r_S}\} \cup S_R\big)$; this will
be the set of constants~$c$ such that~$S(c,c) \in \nu(D(\params))$ (observe
that~$S_R$ and~$S_S$ are disjoint; this is in order to avoid satisfying the
query). To “achieve” these subsets, we for each free type~$\mathbf{t}$ of~$R$
guess a subset~$W_{R,\mathbf{t}}$ of the free facts of~$R$ of
type~$\mathbf{t}$; these will be the~$R$-facts~$f$ of type~$\mathbf{t}$ such
that~$\nu(f)$ matches~$R(x,x,y,y)$ with a constant in~$S_R$ for variable~$x$.
Similarly we for each free type~$\mathbf{t}$ of~$S$ guess a
subset~$W_{S,\mathbf{t}}$ of the free facts of~$S$ of type~$\mathbf{t}$; these
will be the~$S$-facts~$f$ of type~$\mathbf{t}$ such that~$\nu(f)$
matches~$S(x,x)$ with a constant in~$S_S$ for variable~$x$. To ensure that
these subsets are indeed enough to cover~$S_R$ and~$S_S$, we surjectively
assign each of the corresponding facts a constant in~$S_R$ (for~$R$-facts)
or~$S_S$ (for~$S$-facts). For such facts~$f$, we then choose one of
the~$\alpha_\mathbf{t}$ valuations of~$f$ that ensure that the ground fact
obtained satisfies the correspoinding atom of the query; crucially, this
quantity depends only on the type of~$f$. For the remaining (free) facts~$f$, we
choose one of the~$\beta_\mathbf{t}$ valuations of~$f$ that ensure that the
ground facts obtained do not satisfy the corresponding atom of the query. In
the end we obtain that~$V(\params)$ is equal to the following
expression:

\begin{equation}
\label{eq:base-intermediate}
\sum_{S_R \subseteq \dom \setminus \{c_1,\ldots,c_{r_R}\}} \sum_{S_S \subseteq\dom \setminus \big( \{c_1,\ldots,c_{r_S}\} \cup S_R\big)}
\sum_{W_{R,\mathbf{t_1}} \subseteq Q_{R,\mathbf{t_1}}} \cdots
\sum_{W_{R,\mathbf{t_J}} \subseteq Q_{R,\mathbf{t_J}}}
\sum_{W_{R,\mathbf{t'_1}} \subseteq F_{S,\mathbf{t'_1}}} \cdots
\sum_{W_{R,\mathbf{t'_K}} \subseteq F_{S,\mathbf{t'_K}}}
\gamma
\end{equation}

where~$\mathbf{t_1}\ldots \mathbf{t_J}$ are all the free types of~$R$,
$\mathbf{t'_1}\ldots \mathbf{t'_K}$ are all the free types of~$S$,
$Q_{R,\mathbf{t_i}}$ is the set of the first~$q_{R,\mathbf{t_i}}$ free
~$R$-facts of type~$\mathbf{t_i}$, and where,
letting~$z \defeq \sum_{i=1}^J |W_{R,\mathbf{t_i}}|$ and $z' \defeq
\sum_{i=1}^K |W_{S,\mathbf{t'_i}}|$, we have

\begin{equation*}
\gamma = \surj_{z \to |S_R|} \times \surj_{z' \to |S_S|} \times \big(\prod_{i=1}^J \alpha_{\mathbf{t_i}}^{|W_{R,\mathbf{t_i}}|} \beta_{\mathbf{t_i}}^{|Q_{R,\mathbf{t_i}}| - |W_{R,\mathbf{t_i}}|}  \big)
\times \big(\prod_{i=1}^K \alpha_{\mathbf{t'_i}}^{|W_{S,\mathbf{t'_i}}|} \beta_{\mathbf{t'_i}}^{|F_{S,\mathbf{t'_i}}| - |W_{S,\mathbf{t'_i}}|}  \big).
\end{equation*}

Obviously, we cannot compute the expression in~\ref{eq:base-intermediate} in polynomial time, since we are summing over subsets of the input facts.
However, because~$\gamma$ depends only on the sizes of all these sets, we can express~$V(\params)$ as

\begin{equation}
\label{eq:base-final}	
\sum_{0 \leq s_R \leq d-r_R} \sum_{0 \leq s_S \leq d-r_S -s_R }
\sum_{0 \leq w_{R,\mathbf{t_1}} \leq q_{R,\mathbf{t_1}}} \cdots
\sum_{0 \leq w_{R,\mathbf{t_J}} \leq q_{R,\mathbf{t_J}}}
\sum_{0 \leq w_{R,\mathbf{t'_1}} \leq \ff_{R,\mathbf{t'_1}}} \cdots
\sum_{0\leq w_{R,\mathbf{t'_K}} \leq \ff_{R,\mathbf{t'_K}}}
\delta,
\end{equation}

where, letting again $z \defeq \sum_{i=1}^J w_{R,\mathbf{t_i}}$ and $z' \defeq \sum_{i=1}^K w_{S,\mathbf{t'_i}}$, we have

\begin{align*}
\delta = & \surj_{z \to s_R} \times \surj_{z' \to s_S} \times \big(\prod_{i=1}^J \alpha_{\mathbf{t_i}}^{w_{R,\mathbf{t_i}}} \beta_{\mathbf{t_i}}^{q_{R,\mathbf{t_i}} - w_{R,\mathbf{t_i}}}  \big)
\times \big(\prod_{i=1}^K \alpha_{\mathbf{t'_i}}^{w_{S,\mathbf{t'_i}}} \beta_{\mathbf{t'_i}}^{\ff_{S,\mathbf{t'_i}} - w_{S,\mathbf{t'_i}}}  \big)\\
&  \times \binom{d-d_R}{s_R}\binom{d-r_S-s_R}{s_S} \times \big(\prod_{i=1}^J \binom{q_{R,\mathbf{t_i}}}{w_{R,\mathbf{t_i}}}\big) 
\times \big(\prod_{i=1}^K \binom{\ff_{S,\mathbf{t'_i}}}{w_{S,\mathbf{t'_i}}}\big).
\end{align*}

But then, because the expression~\ref{eq:base-final} contains a fixed number of
nested sums (this number depends only on the query~$q$), and because indices
and summands are polynomial, this quantity can be computed in polynomial time.
This concludes the base case, that was when~$v=0$.

\paragraph*{Inductive case.}

Next we explain how we can compute the quantities~$V(\params)$ (with the
parameters~$\params$ as in ($\dagger$)) in polynomial time from
quantities~$V(\params')$ with a strictly smaller value for parameter~(a). Hence
we assume that parameter~$v$ is~$\geq 1$. The idea is to get rid of
the~$c_v$-determined facts of~$D(\params)$, by partitioning the valuations
of~$D(\params)$ that do not satisfy~$q$ and satisfy (1) and (2) into

\begin{description}
 \item[(A)] those valuations~$\nu$ that do not satisfy the query and satisfy (1) and (2) and such that
no~$c_v$-determined~$R$-fact or free $R$-fact~$f$ of~$D(\params)$ is such
that~$\nu(f)$ matches~$R(x,x,y,y)$ with variable~$x$ mapped to~$c_v$; and
\item[(B)] those valuations~$\nu$ that do not satisfy the query and satisfy (1) and (2) and such that at
least one $c_v$-determined~$R$-fact or free $R$-fact~$f$ of~$D(\params)$ is
such that~$\nu(f)$ matches~$R(x,x,y,y)$ with variable~$x$ mapped to~$c_v$ (and thus,
no~$c_v$-determined~$S$-fact or free $S$-fact~$f$ of~$D(\params)$ is such
that~$\nu(f)$ matches~$S(x,x)$ with variable~$x$ mapped to~$c_v$).
\end{description}

To compute valuations in $\mathbf{(A)}$, we choose for each~$R$-fact~$f$ that
is~$c_v$-determined one of the~$\beta_{f,c_v}$ valuations of its nulls that is
such that~$\nu(f)$ does not match~$R(x,x,y,y)$, we disallow the free facts
of~$R$ to have valuations that would match on~$c_v$ by increasing~$r^R$ by one,
and we choose any valuation for the~$c_v$-determined facts of~$S$ (since we
know that they will not be part of a match). Formally, letting~$t_f$ be the
number of nulls of a fact~$f$, we have the expression

\[A = \big(\prod_{\substack{c_v\text{-determined}\\R\text{-fact }f}} \beta_{f}\big)
 \times \big(\prod_{\substack{c_v\text{-determined}\\S\text{-fact }f}} d^{t_f}\big)
 \times V(\params')
\]

where~$\params'$ is equal to~$\params$ except that~$v' = v-1$ and~$r'_R = r'_R
+1$. This can be computed in polynomial time if we know the
value~$V(\params')$.  For valuations in~$\mathbf{(B)}$, we do as follows. We
choose exactly which subset of the~$c_v$-determined and free facts of~$R$ will
match~$R(x,x,y,y)$ with~$x=c_v$ by partioning according to the types, for the
remaining~$c_v$-determined~$R$-facts we choose one of the~$\beta$ valuations
that do not match on~$c_v$, for the remaining free facts of~$R$ we disallow to
match on~$c_v$ by increasing~$r^R$ by one, for each~$c_v$-determined fact~$f$
of~$S$ we choose one of the~$\beta_{f,c_v}$ valuations of~$f$ that are such
that~$\nu(f)$ is not~$S(c_v,c_v)$, and for the free facts of~$S$ we disallow
matching on~$c_v$ by increasing~$r_S$ by one. Formally,
letting~$\mathbf{t_1},\ldots,\mathbf{t_J}$ be the types of the~$c_v$-determined
$R$-facts and $\mathbf{t'_1},\ldots,\mathbf{t'_J}$ be the free types of~$R$ we
obtain an expression of the form

\begin{equation}
B = \sum_{0 \leq h_\mathbf{t_1} \leq n_{R,c_v,\mathbf{t_1}}} \cdots \sum_{0 \leq h_\mathbf{t_J} \leq n_{R,c_v,\mathbf{t_J}}}
    \sum_{0 \leq h'_\mathbf{t'_1} \leq q_{R,\mathbf{t'_1}}} \cdots \sum_{0 \leq h'_\mathbf{t'_K} \leq q_{R,\mathbf{t'_K}}}
    \gamma \times \delta \times V(\params')
\end{equation}

where~$\gamma$ is equal to~$1$ if~$\sum_{i=1}^J h_\mathbf{t_i} + \sum_{i=1}^K h_\mathbf{t'_i} \geq 1$ and to~$0$ otherwise (in order to match on~$c_v$ in~$R$),~$\delta$ is

\[ \big(\prod_{i=1}^J \binom{n_{R,c_v,\mathbf{t_i}}}{h_\mathbf{t_i}} \alpha_{\mathbf{t_i}}^{h_\mathbf{t_i}} \beta_{\mathbf{t_i}}^{n_{R,c_v,\mathbf{t_i}} - h_\mathbf{t_i}} \big)
\times \big(\prod_{i=1}^J \binom{q_{R,\mathbf{t'_i}}}{h'_\mathbf{t'_i}} \alpha_{\mathbf{t'_i}}^{h'_\mathbf{t'_i}}\big)
\times \big(\prod_{\substack{c_v\text{-determined}\\S\text{-fact }f}} \beta_f\big),
\]

and where~$\params'$ is equal to~$\params$ except that:
\begin{itemize}
  \item we have~$v'=v-1$;
  \item for every free type~$\mathbf{t}$ of~$R$, we have~$q'_{R,\mathbf{t}} = q'_{R,\mathbf{t}} - h'_\mathbf{t}$;
  \item we have~$r'_R = r_R +1$;
  \item we have~$r'_S = r_S +1$.
\end{itemize}

The crucial point to see that expressions A and B indeed compute what we
want is that we do not need to remember exactly which susbsets of constants are
disallowed to match for the free facts, but we only need to remember their
numbers (this is what allows us to use a dynamic programming approach); and
similarly for the free facts of~$R$, we only need to remember how many we have
left at each stage of each type, but not their precise subsets.  Again, if we know the values
of~$V(\params')$ where in~$\params'$ parameter~$v'$ is equal to~$v -1$ then we
can compute in polynomial time the value~$V(\params) = A + B$.  Thus, we can
compute all values~$V(\params)$ in polynomial time, and this concludes this
example.}\qed
\end{example}

The proof in this example can be extended as-is to any connected \sjfbcq\
containing only two atoms and having only one variable ($x$) that joins.  We
claim that the same idea works for an arbitrary number of atoms, which
concludes Theorem~\ref{thm:countvals-codd-uniform}.
Since a full proof is technically tedious and does not provide new insights, we omit it here.

\section{Dichotomies for counting completions}
\label{sec:countcompls-sjfcqs}
In this section, we study the complexity of the problems of counting
completions satisfying an~\sjfbcq~$q$, in the four cases that can be obtained
by considering naive or Codd tables and non-uniform or uniform domains.  We
will again use the notion of pattern as introduced in
Definition~\ref{def:pattern}. Our first step is to observe that
Lemma~\ref{lem:pattern-parsimonious}, which we used in the last section for the
problems or counting valuations, extends to the problems of counting
completions. 

\begin{lemma}
\label{lem:pattern-parsimonious-compls}
	Let~$q,q'$ be \sjfbcqs\ such that~$q'$ is a pattern of~$q$.  Then we
have that~$\countcompls(q') \pr \countcompls(q)$. Moreover, the same results
hold if we restrict to the case of Codd tables, and/or to the uniform setting.
\end{lemma}
\begin{proof}
	The reduction is exactly the same as the one of
Lemma~\ref{lem:pattern-parsimonious}.  To show that this reduction works
properly for counting completions, it is enough to observe that for every
pair of valuations~$\nu_1,\nu_2$ of~$D'$ (or of~$D$, since~$D$ and~$D'$ have exactly
the same set of nulls), we have that~$\nu_1(D')=\nu_2(D')$ iff~$\nu_1(D) =
\nu_2(D)$.
\end{proof}

We will then follow the same general strategy as in the last section, i.e.,
prove hardness for some simple patterns and combine these with
Lemma~\ref{lem:pattern-parsimonious-compls} and tractability proofs to obtain
dichotomies.  Our findings are summarized in the last two columns of Table
\ref{tab:dichos-count-intro} in the introduction.  We start in
Section~\ref{subsec:countcompls-non-uniform} with the non-uniform cases and
continue in Section~\ref{subsec:countcompls-uniform} with the uniform cases.
Again, we explicitly state when a \shp-hardness result holds even in the
restricted setting in which there is a fixed domain over which nulls are
interpreted.

\subsection{The complexity on the non-uniform case}
\label{subsec:countcompls-non-uniform}

	Here, we study the complexity of the problems $\countcompls(q)$ and
$\ccountcompls(q)$, providing dichotomy results in both cases.  In fact, it
turns out that these problems are~\shp-hard for all \sjfbcqs.  To prove this,
it is enough to show that the problem~$\ccountcompls(R(x))$ is hard, that is,
even counting the completions of a single unary table is \shp-hard in the
non-uniform setting. 

\begin{proposition}
\label{prp:countcompls}
	$\ccountcompls(R(x))$ is \shp-hard.
\end{proposition}
\begin{proof}
We provide a polynomial-time parsimonious reduction from the problem of
counting the vertex covers of a graph, which we denote by \sVC.  Let~$G=(V,E)$
be a graph. We construct a Codd table~$D$ using a single unary relation~$R$
such that the number of completions of~$D$ equals the number of vertex covers
of~$G$.  For every edge~$e=\{u,v\}$ of~$G$, we have one null~$\bot_e$ with
$\dom(\bot_e) = \{u,v\}$ and the fact~$R(\bot_e)$.  Let~$a$ be a fresh
constant.  For every node~$u\in V$ we have a null~$\bot_u$ with~$\dom(\bot_u) =
\{u,a\}$ and the fact~$R(\bot_u)$.  Last, we add the fact~$R(a)$.  We now show
that the number of completions of~$D$ equals the number of vertex covers
of~$G$. 
	
	Let~$\VC(G)$ be the set of vertex covers of~$G$.  For a valuation~$\nu$
of~$D$, define the set~$S_\nu := \{u \in V \mid R(u) \in D\}$.  Since the
fact~$R(a)$ is in every completion of~$D$, it is clear that the number of
completions of~$D$ is equal to~$|\{S_\nu \mid \nu \text{ is a valuation of }
D\}|$.  We claim that~$\VC(G)=\{S_\nu \mid \nu \text{ is a valuation of } D\}$,
which shows that the reduction works. ($\subseteq$) Let~$C \in \VC(G)$, and let
us show that there exists a valuation~$\nu$ of~$D$ such that~$S_\nu = C$.  For
a null of the form~$\bot_e$ with~$e=\{u,v\}\in E$, assuming wlog that~$u \in
C$, we define~$\nu(\bot_e)$ to be~$u$.  For a null of the form~$\bot_u$
with~$u\in V$, we define~$\nu(\bot_u)$ to be~$u$ if~$u\in C$ and~$a$ otherwise.
It is then clear that~$S_\nu = C$. ($\supseteq$) Let~$\nu$ be a valuation
of~$D$, and let us show that~$S_\nu$ is a vertex cover.  Assume by
contradiction that there is an edge~$e=\{u,v\}$ such that~$e \cap S_\nu =
\emptyset$.  By definition of~$D$, we must have~$\nu(\bot_e) \in \{u,v\}$, so
that one of~$u$ or~$v$ must be in~$S_\nu$, hence a contradiction. Therefore, we
conclude that $\sVC \pr \ccountcompls(R(x))$.
\end{proof}

Recall from Section~\ref{sec:preliminaries} that, to avoid trivialities, we
assume all \sjfbcqs\ to contain at least one atom and that all atoms have at
least one variable.  Using Lemma~\ref{lem:pattern-parsimonious-compls}, this allows us
to obtain the following dichotomy result.

\begin{theorem}[Dichotomy]
\label{thm:countcompls-sjfcq-hard} 	
For every \sjfbcq~$q$, it holds that $\countcompls(q)$ and $\ccountcompls(q)$
are~\shp-hard. 	
\end{theorem}

Notice here that we do not claim membership in \shp; in fact, we will come back
to this issue in Section~\ref{sec:misc} to show that this is unlikely to be
true for naive tables.  However, we can still show that membership in \shp\
holds for Codd tables.  We then obtain:

\begin{theorem}[Dichotomy]
	\label{thm:countcompls-sjfcqs-complete-codd}
For every \sjfbcq~$q$, the problem $\ccountcompls(q)$ is \shp-complete.
\end{theorem} 
\begin{proof}
	Hardness is from Theorem~\ref{thm:countcompls-sjfcq-hard}. To show
membership in \shp\, we will actually prove a more general result in
Section~\ref{subsec:countcompls-codd-sharp-p}. There, we show that for every
Boolean query~$q$ such that~$q$ has model checking in \ptime\, the
problem~$\ccountcompls(q)$ is in \shp. This in particular applies to all
\sjfbcqs.
\end{proof}

\subsection{The complexity on the uniform case}
\label{subsec:countcompls-uniform}
We now investigate the complexity of~$\ucountcompls(q)$ and
$\cucountcompls(q)$.  Recall that in the non-uniform case, even counting the
completions of a single unary table is a \shp-hard problem. This no longer
holds in the uniform case, as we will show that $\ucountcompls(q)$ is in \fp\
for every \sjfbcq\ that is defined over a schema consisting exclusively of
unary relation symbols.

Such a positive result, however, cannot be extended much further.  In fact, we
show next that~$R(x,x)$ and~$R(x,y)$ are hard patterns, both for naive and Codd
tables (and, thus, we also conclude that the problem of counting the
completions of a single binary Codd table is a \shp-hard problem).  We start
with the case of naive tables, for which hardness even holds when nulls are
interpreted over the fixed domain~$\{0,1\}$.

\begin{proposition}
\label{prp:Rxx-Rxy-hard-compls-naive}
The problems $\ucountcompls(R(x,x))$ and $\ucountcompls(R(x,y))$ are both \shp-hard, even
	when nulls are interpreted over the same fixed domain~$\{0,1\}$.
\end{proposition}
\begin{proof}
We reduce from $\sIS$, the problem of counting the number of independent sets
of a graph.  Let~$G=(V,E)$ be a graph. We will construct an incomplete
database~$D$ containing a single binary predicate~$R$ such that each completion
of~$D$ satisfies~$R(x,x)$ and the number of completions of~$D$ is~$2^{|V|} +
\#\IS(G)$, thus establishing hardness for the two queries.  For every node~$u
\in V$, we have a null~$\bot_u$ with~$\dom(\bot_u) = \{0,1\}$. We then
construct the naive table~$D$ as follows: \begin{itemize} \item for every
node~$u \in V$ we add to $D$ the fact~$R(u,\bot_u)$; \item then for every
edge~$\{u,v\} \in E$, we add the facts~$R(\bot_u,\bot_v)$
and~$R(\bot_v,\bot_u)$ to $D$; and \item last, we add the facts~$R(0,0)$,
$R(0,1)$, $R(1,0)$, and~$R(\bot,\bot)$, where~$\bot$ is a fresh null.
\end{itemize} It is clear that every completion of~$D$ satisfies~$R(x,x)$
(thanks to the fact~$R(\bot,\bot)$).  Let us now count the number of
completions of~$D$.  First, we observe that, thanks to the facts of the form
$R(u,\bot_u)$, for $u \in V$, for every two valuations~$\nu,\nu'$ that do not
assign the same value to the nulls of the form~$\bot_u$, it is the case that
$\nu(D) \neq \nu(D')$.  We then partition the completions of~$D$ into those
that contain the fact~$R(1,1)$, and those that do not contain~$R(1,1)$.
Because of the facts of the form $R(u,\bot_u)$, for $u \in V$, and thanks to
the fact~$R(\bot,\bot)$ which becomes~$R(1,1)$ when we assign~$1$ to~$\bot$,
there are exactly~$2^{|V|}$ completions of~$D$ that contain~$R(1,1)$.
Moreover, it is easy to see that there are~$\#\IS(G)$ valuations~$\nu$ of~$D$
that assign~$0$ to~$\bot$ and that yield a completion not containing~$R(1,1)$.
Indeed, one can check that a valuation of~$D$ that assigns~$0$ to~$\bot$ yields
a completion not containing~$R(1,1)$ if and only if the set~$\{u \in V \mid
\nu(\bot_u)=1 \}$ is an independent set of~$G$. Therefore, we conclude that the
number of completions of~$D$ is indeed~$2^{|V|} + \#\IS(G)$, and therefore
that~$\sIS \tr \ucountcompls(q)$, where~$q$ can be~$R(x,x)$ or~$R(x,y)$.
\end{proof} 

Next, we prove hardness of the same queries for Codd tables (but in this case
we do not know if harness holds when nulls are interpreted over a fixed domain,
as our proof will use domains of unbounded size).
We will reduce from the problem of counting the number of \emph{induced pseudoforests} of a graph, as defined next.

\begin{definition}
\label{def:pseudoforest}
	 A graph~$G$ is a \emph{pseudoforest} if every connected component
of~$G$ contains at most one cycle.  Let~$G=(V,E)$ be a graph. For~$S \subseteq
E$, let us denote by~$G[S]$ the graph~$(V',S)$, where~$V'$ is the set of nodes
of~$G$ that appear in some edge of~$S$. The problem~$\shpf$ is the problem that
takes as input a graph~$G=(V,E)$ and outputs the number of edge sets~$S
\subseteq E$ such that~$G[S]$ is a pseudoforest.
\end{definition}

	Using techniques from matroid theory, the authors
of~\cite{gimenez2006complexity} have shown that~$\shpf$ is \shp-hard on graphs.
We explain in Appendix~\ref{apx:shpf-hard-bipartite} how their proof actually
shows hardness of this problem for \emph{bipartite graphs} (which we need);
formally, we have:

\begin{toappendix}
	\subsection{Proof for Proposition~\ref{prp:shpf-hard-bipartite}}
	\label{apx:shpf-hard-bipartite}
	In this section we explain how to obtain the following hardness result. 
\end{toappendix}

\begin{propositionrep}[Implied by {\cite{gimenez2006complexity}}]
\label{prp:shpf-hard-bipartite}
	The problem~$\shpf$ restricted to bipartite graphs is \shp-hard.
\end{propositionrep}

\begin{toappendix}
	This result is proven for (non-necessarily bipartite) graphs
in~\cite{gimenez2006complexity} using techniques from matroid theory, in
particular using the notions of \emph{bicircular matroid} of a graph and of
\emph{Tutte polynomial} of a matroid.  We did not find a way to show that the
result holds on bipartite graphs without explaining their proof for general
graphs, and we did not find a way to explain the proof for general graphs
without introducing these concepts. Therefore, we need to define these concepts
here.  We have tried to keep this exposition as brief as possible, but more
detailed introductions to matroid theory and to the Tutte polynomial can be
found in~\cite{oxley2003matroid,welsh1999tutte}.  First, we define what is a
matroid.
	\begin{definition}
		\label{def:matroid}
		A matroid~$M=(E,\I)$ is a pair where~$E$ is a finite set (called the \emph{ground set}) and~$\I$ is a set of subsets of~$E$ whose elements are called \emph{independent sets} and that satisfies the following properties:
		\begin{description}
			\item[Non emptiness.] $\I \neq \emptyset$;
			\item[Heritage.] For every~$A' \subseteq A \subseteq E$, if~$A \in \I$ then~$A' \in \I$;
			\item[Independent set exchange.] For every~$A,B \in \I$, if~$|A| > |B|$ then there exists~$x\in A \setminus B$ such that~$B\cup \{x\} \in \I$.
		\end{description}
	\end{definition}
	
	In a matroid~$M=(E,\I)$, an independent set~$A \in \I$ is called a
	\emph{basis} if every strict superset~$A \subsetneq A' \subseteq E$ is
	not in~$\I$. Notice that, thanks to the independent set exchange
	property, all bases of~$M$ have the same number of elements. The \emph{rank} of~$M$ is defined as the number of elements in any basis of $M$. Given a matroid~$M=(E,\I)$
	and~$A \subseteq E$, we can define the \emph{submatroid} of~$M$ generated
	by~$A$ to be~$M_A = (A,\I')$, where for~$A' \subseteq A$ we
	have~$A' \in \I'$ iff~$A' \in \I$ (one should check that this is indeed
	a matroid).  The \emph{rank function}~$\rk_M:\{A \mid A \subseteq E\} \to \N$
	of~$M$ is then defined with~$\rk_M(A)$ being the rank of the
	matroid~$M_A$. We will now omit the subscript in~$\rk_M$ as this will not cause confusion.
	We are ready to define the Tutte polynomial of a matroid.

	\begin{definition}
		\label{def:tutte}
		Let~$M=(E,\I)$ be a matroid. The \emph{Tutte polynomial} of~$M$, denoted~$\T(M;x,y)$, is the two-variables polynomial defined by
		\begin{eqnarray*}
		\label{eq:tutte}
			\T(M;x,y) & = & \sum_{A \subseteq E} (x-1)^{\rk(M)-\rk(A)} (y-1)^{|A|-\rk(A)}
		\end{eqnarray*}
	\end{definition}

	We will use the following observation:
	\begin{observation}
		\label{obs:2-1}
		Let~$M=(E,\I)$ be a matroid. Then~$\T(M;2,1) = |\I|$, i.e., evaluating the Tutte polynomial of a matroid at point~$(2,1)$ simply counts its number of independent sets.
	\end{observation}
	\begin{proof}
		We have~$\T(M;2,1) =\sum_{A \subseteq E} 0^{|A|-\rk(A)}$. We recall the convention that~$0^0 = 1$, and the fact that~$0^k=0$ for~$k>0$.
		Observe then that we always have~$\rk(A) \leq |A|$, and that we have~$\rk(A) = |A|$ if and only if~$A\in \I$, which proves the claim.
	\end{proof}
	
	Next, we define what is called the \emph{bicircular matroid} of a graph~$G=(V,E)$.
	Recall from Definition~\ref{def:pseudoforest} the definition of the induced subgraph~$G[S]$ for~$S \subseteq E$.

	\begin{definition}
		\label{def:bicircular}
		Let~$G=(V,E)$ be a graph and~$\I = \{S\subseteq E \mid G[S] \text{ is a pseudoforest}\}$.
		Then one can check that~$(E,\I)$ is a matroid~\cite{zaslavsky1982bicircular}.
		This matroid is called the \emph{bicircular matroid of~$G$}, and is denoted by~$B(G)$.
	\end{definition}

	Notice then that the problem~$\shpf$ is exactly the same as the problem of computing, given as input a graph~$G$, the quantity~$\T(B(G);2,1)$.
	We now explain the steps used in~\cite{gimenez2006complexity} to prove that computing~$\T(B(G);2,1)$ is \shp-hard for graphs.
	The starting point of our explanation is that computing~$\T(B(G);1,1)$ is \shp-hard.

	\begin{proposition}[{\cite[Corollary 4.3]{gimenez2006complexity}}]
		\label{prp:1-1}
		The problem of computing, given a graph~$G$, the quantity~$\T(B(G);1,1)$ is \shp-hard.
	\end{proposition}
	Second, let us define the following univariate polynomial: for a graph~$G$, let~$P_G(x)$ be
	\[P_G(x) \ = \ \T(B(G);x,1).\]
	Notice that this is indeed a polynomial and that its degree is at most~$|E|$ (the degree is exactly~$|E|$ iff~$G$ is itself a pseudoforest).
	If we could compute efficiently the coefficients of~$P_G$, then we could in particular compute the value~$P_G(1) = \T(B(G);1,1)$, which is \shp-hard by the previous proposition.
	We recall that to compute the coefficients of a polynomial of degree~$n$, it is enough to know its value on~$n+1$ distinct points; in fact, given these values in $n+1$ distinct points, it is possible to efficiently compute the coefficients of the polynomial by using standard  interpolation techniques (for example, by using Lagrange polynomials).

	We need one last definition.
	\begin{definition}
		\label{def:stretch}
		Let~$G$ be a graph. For~$k\in \N$, let~$\s_k(G)$ be the graph obtained from~$G$ by replacing each edge of~$G$ by a path of lenght~$k$; this graph is called the \emph{$k$-stretch} of~$G$.
	\end{definition}

	Then, using a result attributed to Brylawski (see~\cite{jaeger1990computational}), the authors of~\cite{gimenez2006complexity} obtain that, “up to a trivial factor”, we have
	\[\T(B(\s_k(G));2,1) \ \simeq \ \T(B(G);2^k,1). \]
	A careful inspection of~\cite{jaeger1990computational} reveals\footnote{To be precise, we use Equations (7.1) and (7.2) of~\cite{jaeger1990computational} with~$x=1,y=0$, and Equation (2.2) with~$x=2,y=1$.} that, in fact, we have
	\[\T(B(\s_k(G));2,1) \ = \ (2^k -1)^{|E|-\rk_{B(G)}(E)} \times \T(B(G);2^k,1). \]
	Notice that~$\rk_{B(G)}(E)$ is the size (number of edges) of a pseudoforest of~$G$ that is maximal by inclusion of edges, which we can compute in polynomial time.\footnote{This is because, since~$B(G)$ is a matroid, any two such pseudoforests have the same number of edges. We can then simply start from the empty subgraph and iteratively add edges until it is not possible to add an edge such that the resulting graph is a pseudoforest. This also relies on the fact that we can check in polynomial time whether a graph is a pseudoforest.}
	
	With this, the authors of~\cite{gimenez2006complexity} can conclude the proof that computing~$\T(B(G);2,1)$ is hard for (non-necessarily bipartite) graphs, i.e., that~$\shpf$ is \shp-hard.
	Indeed, given as input~$G=(V,E)$, we can construct in polynomial time the graphs~$\s_k(G)$ for~$|E|+1$ distinct values of~$k$,
	then use oracle calls to obtain the numbers~$\T(B(\s_k(G));2,1)$, which gives us the value of~$P_G$ on~$|E|+1$ distinct points.
	With that we can recover the coefficients of~$P_G$ and compute~$P_G(1) = \T(B(G);1,1)$ as argued above, thus proving hardness for general graphs.
	To obtain hardness for bipartite graphs, it is enough to observe that when~$k$ is even then the~$k$-stretch of~$G$ is bipartite (even if $G$ is not bipartite).
	Hence, to obtain a proof of Proposition~\ref{prp:shpf-hard-bipartite} for bipartite graphs, we can simply change that proof and specify that we make~$|E|+1$ calls to the oracle~$\T(B(\s_k(G));2,1)$ for~$|E|+1$ disctinct \emph{even} values of~$k$.
\end{toappendix}

To prove that the reduction that we will present is correct, we will also need the following folklore lemma about pseudoforests.
We recall that an \emph{orientation} of an undirected graph~$G=(V,E)$ is a directed graph that can be obtained from~$G$ by orienting every edge of~$G$. Equivalently, one can see such an orientation as a function~$f:E \to V$ that assigns to every edge in~$G$ a node to which it is incident.
We then have: 

	\begin{lemma}
		\label{lem:orientation}
		A graph~$G$ is a pseudoforest if and only if there exists an orientation of~$G$ such that every node has outdegree at most~$1$.
	\end{lemma}
	\begin{proof}
	Folklore, see, e.g., \cite{kowalik2006approximation,fan2015decomposing,grout2019decomposing}.
	\end{proof}

Using the hardness of~$\shpf$ on bipartite graphs, we are able show hardness of~$\cucountcompls(R(x,x))$ and $\cucountcompls(R(x,y))$ as follows.
	\begin{proposition}
	\label{prp:Rxx-Rxy-hard-compls-codd}
	The problems $\cucountcompls(R(x,x))$ and $\cucountcompls(R(x,y))$ are both \shp-hard.
	\end{proposition}
	\begin{proof}
	We reduce both problems from $\shpf$ on bipartite graphs. Let~$G=(U \sqcup V,E)$ be a bipartite  graph. We will construct a uniform Codd table~$D$ over binary relation~$R$ such that (1) all the completions of~$D$ satisfy both queries; and (2) the number of completions of~$D$ is equal to~$\shpf(G)$, thus establishing hardness.
	For every~$(t,t') \in (U \cup V)^2 \setminus E$, we add to~$D$ the fact~$R(t,t')$; we call these the \emph{complementary facts}.
	For every~$u \in U$ we add to~$D$ the fact~$R(u,\bot_u)$ and for every~$v \in V$ the fact~$R(\bot_v,v)$. 
	Finally, we add to~$D$ a fact~$R(f,f)$ where~$f$ is a fresh constant.
	The uniform domain of the nulls if~$\dom = U \cup V$. It is clear that~$D$ is a Codd table and that every completion of~$D$ satisfies both queries (thanks to the fact~$R(f,f))$, so (1) holds.
	We now prove that (2) holds.
	First of all, observe that a completion~$\nu(D)$ of~$D$ is uniquely determined by the set of edges~$\{(u,v)\in E \mid R(u,v) \in \nu(D)\}$: this is because~$\nu(D)$ already contains all the complementary facts. For a set~$S\subseteq E$ of edges, let us define~$D_S$ to be the complete database that contains all the complementary facts and all the facts~$R(u,v)$ for~$(u,v) \in S$ (note that $D_S$ is not necessarily a completion of~$D$). We now argue that for every set~$S \subseteq E$, we have that~$D_S$ is a completion of~$D$ if and only if~$G[S]$ is a pseudoforest, which would conclude the proof.
		By Lemma~\ref{lem:orientation} we only need to show that~$D_S$ is a completion of~$D$ if and only if~$G[S]$ admits an orientation with maximum outdegee~$1$.
		We show each direction in turn. $(\Rightarrow)$ Assume~$D_S$ is a completion of~$D$, and let~$\nu$ be a valuation witnessing this fact, i.e., such that~$\nu(D)=D_S$. 
		First, observe that we can assume without loss of generality that~($\star$) for every~$e =(u,v) \in S$, we have either~$\nu(\bot_u)=v$ or~$\nu(\bot_v)=u$ but not both. Indeed, if we had both then we could modify~$\nu$ into~$\nu'$ by redefining, say,~$\nu'(\bot_u)$ to be~$u$, and we would still have that~$\nu'(D)=D_S$ (because~$R(u,u)$ is already present in~$D$: it is a complementary fact).
		We now define an orientation~$f_{\nu}:S \to U\cup V$ of~$G[S]$ from~$\nu$ as follows. Let~$e=(u,v) \in S$. Then: if we have~$\nu(\bot_u)=v$ we define~$f_\nu((u,v))$ to be~$v$, i.e., we orient the (undirected) edge~$(u,v)$ from~$u$ to~$v$. Else, if we have~$\nu(\bot_v)=u$ we define~$f_\nu((u,v))$ to be~$u$, i.e., we orient the (undirected) edge~$(u,v)$ from~$v$ to~$u$. Observe that by~($\star$) $f_\nu$ is well defined. It is then easy to check that the maximal outdegree of the directed graph defined by~$f_\nu$ is~$1$: this is because for every~$u \in U$ (resp.,~$v \in V$), there is only one fact in~$D$ of the form~$R(u,\text{null})$ (resp.,~$R(\text{null},v)$), namely, the fact~$R(u,\bot_u)$ (resp., $R(\bot_v,v)$).
		$(\Leftarrow)$ Let~$f:S \to U\cup V$ be an orientation of~$G[S]$ with maximum outdegree~$1$. Let~$\nu_f$ be the valuation of~$D$ defined from~$f$ as follows:
		for every~$u\in U$ (resp., $v\in V$), if there is an edge~$(u,v)\in S$ such that~$f((u,v)) = v$ (resp., such that $f((u,v)) = u$), then define~$\nu_f(\bot_u)$ to be~$v$ (resp., define~$\nu_f(\bot_v)$ to be~$u$). Observe that there can be at most one such edge because~$f$ has maximum outdegree~$1$, so this is well defined. If there is no such edge, define~$\nu_f(\bot_u)$ to be~$u$ (resp., define $\nu_f(\bot_v)$ to be~$v$). Since all edges in~$S$ are given an orientation by~$f$, it is clear that for every~$(u,v)\in S$ we have~$R(u,v) \in \nu_f(D)$. Moreover, since~$\nu_f(D)$ contains all the complementary facts, we have that~$\nu_f(D)=D_S$, which shows that~$D_S$ is a completion of~$D$ and concludes this proof.
\end{proof}

	As we show next, these two patterns 
	suffice to characterize the complexity of $\ucountcompls(q)$ and~$\cucountcompls(q)$.

	\begin{toappendix}
		\subsection{Proof of Theorem~\ref{thm:ucc-naive-fp}}
		\label{apx:ucc-naive-fp}
	In this section we prove the tractability part of the following:
	\end{toappendix}

	\begin{theoremrep}[Dichotomy]
	\label{thm:ucc-naive-fp}
	Let $q$ be an \sjfbcq. 
	 If $R(x,x)$ or $R(x,y)$ is a pattern of $q$, 
	 then $\ucountcompls(q)$ and~$\cucountcompls(q)$ are \shp-hard. Otherwise, these problems are in~\fp.
	\end{theoremrep}

From what precedes, we only have to prove the tractability part of that claim, and this for naive tables.
	To this end, we define a \emph{conjunction of basic singletons \sjfbcq} to be an \sjfbcq\ of the form~$C_1(x_1) \land \ldots \land C_m(x_m)$,
	where each~$C_i(x_i)$ is a conjunction of unary atoms over the same variable~$x_i$ (here the~$x_i$ are pairwise distinct).
	Since~$q$ does not contain the pattern~$R(x,x)$ nor the pattern~$R(x,y)$, observe that~$q$ must in fact be a conjunction of basic singletons \sjfbcq.
	The main difficulty is to decompose the computation in such a way that we do not count the same completion twice. Moreover, the fact that the database is naive and not Codd, and the fact that constants can appear everywhere, complicate a lot the description of the algorithm.
For these reasons, we provide in the next section an example that gives the intuition of the general proof, and defer the general proof to Appendix~\ref{apx:ucc-naive-fp}.

Note that, as in the last section, we do not claim membership in~\shp\ in the statement of Theorem~\ref{thm:ucc-naive-fp}. However, 
and also as in the last section, 
we can show that these problems are in \shp\ for Codd tables, which allows us to obtain our last dichotomy for exact counting.

\begin{theorem}[Dichotomy]
	\label{thm:ucc-codd-fp}
	Let~$q$ be an \sjfbcq.
	If $R(x,x)$ or~$R(x,y)$ is a pattern of $q$, then 
$\cucountcompls(q)$ is \shp-complete. Otherwise, this problem is in~\fp.
\end{theorem}
\begin{proof}
	Hardness follows from Theorem~\ref{thm:ucc-naive-fp}, while membership in \shp\ follows from the result proven in Section~\ref{subsec:countcompls-codd-sharp-p}.
\end{proof}

\subsection{Example of tractability for Theorem~\ref{thm:ucc-naive-fp}}
\label{subsec:ucc-naive}
	In this section, we prove that for the query~$q\defeq R(x)\land S(y)$, the problem~$\cucountcompls(q)$ is in~\fp. Observe that~$q$ is a conjunction of basic singletons query, and that it is always satisfied (as long as the database is not empty). We will show that~$\ucountcompls(q)$ is in~\fp.
	Let~$C_{RS},C_R,C_S$ (resp, $N_{RS},N_R,N_S$) be the sets of constants (resp., nulls) that occur respectively: in~$R$ and in~$S$, only in~$R$, only in~$S$, and denote~$c_{RS},c_R,c_S$ (resp., $n_{RS},n_R,n_S$) their sizes. Let~$\dom$ be the uniform domain of the nulls, and let~$d$ be its size.
	First of all, observe that we can assume without loss of generality that~$C \defeq C_{RS} \cup C_R \cup C_S \subseteq \dom$, as otherwise we could simply remove the facts of~$D$ that are over constants that are not in~$\dom$ and this would not change the result.
	Let~$c = c_{RS} + c_R + c_S$.
	The next two claims explain how we can decompose the computation in a way that we do not count a same completion twice (Claim~\ref{claim:expl-disjoint}), and that we count all the possible completions (Claim~\ref{claim:expl-have-all}).
	\begin{claim}
		\label{claim:expl-disjoint}
		For a triplet~$(I_R,I_S,I_{RS})$ of subsets of~$\dom$ satisfying the conditions~($\star$) $I_R \subseteq \dom \setminus C$, $I_S \subseteq \dom \setminus (C \cup I_R)$, and $I_{RS} \subseteq \dom \setminus (C_{RS} \cup I_R \cup I_S)$,
		let us define~$P(I_R,I_S,I_{RS})$ to be the complete database consisting of the following facts:
		\begin{enumerate}
			\item $R(a)$ and~$S(a)$ for~$a \in C_{RS} \cup I_{RS}$;
			\item $R(a)$ for~$a \in I_R \cup (C_R \setminus I_{RS})$;
			\item $S(a)$ for~$a \in I_S \cup (C_S \setminus I_{RS})$;
		\end{enumerate}
		Then, for any two such triplets of sets~$(I_R,I_S,I_{RS})$ and~$(I'_R,I'_S,I'_{RS})$ that are different, the complete databases~$P(I_R,I_S,I_{RS})$ and~$P(I'_R,I'_S,I'_{RS})$ are distinct.
	\end{claim}
	\begin{proof}
		To help the reader, we have drawn in Figure~\ref{fig:drawing} how the sets can intersect.
		If we have~$I_{RS} \neq I'_{RS}$ with~$a \in I_{RS}$ and~$a
		\notin I'_{RS}$, then~$P(I_R,I_S,I_{RS})$ contains both
		facts~$R(a)$ and~$S(a)$,
		while~$P(I'_R,I'_S,I'_{RS})$ does not. So let us
		assume now that~$I_{RS} = I'_{RS}$. If we have~$I_R \neq I'_R$
		with~$a \in I_R$ and~$a\notin I'_R$ then one can check
		that~$P(I_R,I_S,I_{RS})$ contains the
		fact~$R(a)$ while~$P(I'_R,I'_S,I'_{RS})$ does
		not. Hence let us assume that~$I_R = I'_R$. Using the same
		reasoning we obtain that~$I_S = I'_S$, thus completing the
		proof.
	\end{proof}

	\begin{figure}
		\centering
	\begin{tikzpicture}
		\draw (0,0) -- (5,0) -- (5,6) -- (0,6) -- (0,0);
		\draw \boundellipse{2.5,5}{1.5}{.6};
		\draw \boundellipse{2.5,3.5}{1.5}{.6};
		\draw \boundellipse{1,2.5}{.5}{1.5};
		\draw \boundellipse{4,2.5}{.5}{1.5};
		\draw \boundellipse{1.8,.8}{.5}{.5};
		\draw \boundellipse{3.2,.8}{.5}{.5};
		\node at (.6,5.5) {$\dom$};
		\node at (2.5,5) {$C_{RS}$};
		\node at (2.5,3.5) {$I_{RS}$};
		\node at (1,2.5) {$C_R$};
		\node at (4,2.5) {$C_S$};
		\node at (1.8,.8) {$I_R$};
		\node at (3.2,.8) {$I_S$};
	\end{tikzpicture}
		\caption{How the sets $\dom, I_R,I_S,I_{RS},C_{RS},C_R$ and $C_S$ from Claim~\ref{claim:expl-disjoint} are allowed to intersect when they satisfy~($\star$). The sets themselves and the intersections can be empty.}
		\label{fig:drawing}

	\end{figure}

	Our next step is to show that every completion of~$D$ is of the form~$P(I_R,I_S,I_{RS})$ for some triplet~$(I_R,I_S,I_{RS})$ satisfying~($\star$):
	\begin{claim}
		\label{claim:expl-have-all}
		For every completion~$D'$ of~$D$, there exist a triplet~$(I_R,I_S,I_{RS})$ satisfying~($\star$) such that~$D' = P(I_R,I_S,I_{RS})$.
	\end{claim}
	\begin{proof}
		We define:
		\begin{itemize}
			\item $I_R \defeq D'(R) \setminus (C_R \cup D'(S))$; where we see~$D'(R)$ as the set of constants occurring in relation~$R$ of~$D'$.
			\item $I_S \defeq D'(S) \setminus (C_S \cup D'(R))$;
			\item $I_{RS} \defeq (D'(R) \cap D'(S)) \setminus C_{RS}$.
		\end{itemize}
		Then one can easily check that $(I_R,I_S,I_{RS})$ satisfies~($\star$) and that~$D' = P(I_R,I_S,I_{RS})$.
	\end{proof}

	By combining these two claims, we have that the result that we wish to compute (which, we recall, is simply the total number of completions of~$D$, because any valuation satisfies the query) is equal to
	\[\sum_{I_R \subseteq \dom \setminus C} \sum_{I_S \subseteq \dom \setminus (C \cup I_R)} \sum_{I_{RS} \subseteq \dom \setminus (C_{RS} \cup I_R \cup I_S)} \checkp(I_R,I_S,I_{RS}) \]
	where~$\checkp(I_R,I_S,I_{RS}) \defeq \begin{cases}
		1 \text{ if } P(I_R,I_S,I_{RS}) \text{ is a possible completion of } D \\
		0 \text{ otherwise}
	\end{cases}$.

	Next, we show that the value of~$\checkp(I_R,I_S,I_{RS})$ can be computed in polynomial time and actually only depends on the sizes of these sets. 
	In order to show this, we will use the following:
	\begin{claim}
		\label{claim:expl-conditions}
		We have~$\checkp(I_R,I_S,I_{RS}) = 1$ if and only if the following conditions hold:
		\begin{enumerate}
			\item if~$n_R \geq 1$ and $|C_R \cup C_{RS} \cup I_{RS}| = 0$, then we have~$|I_R| \neq 0$. Intuitively, this means that the value of a null in~$N_R$ cannot be “absorbed” by~$C_R \cup C_{RS} \cup I_{RS}$.
			\item if~$n_S \geq 1$ and $|C_S \cup C_{RS} \cup I_{RS}| = 0$, then we have~$|I_S| \neq 0$. (Same reasonning for nulls in~$N_S$.)
			\item if~$n_{RS} \geq 1$ and $|C_{RS}| = 0$, then we have~$|I_{RS}| \neq 0$. (Same reasonning for nulls in~$N_{RS}$.)
			\item the following system of equations, whose variables are natural numbers between~$0$ and~$d$, has a solution:
				\begin{align*}
					z_{N_R}^{\{N_R\}} + z_{N_R}^{\{N_R,C_S\}} +  z_{N_R}^{\{N_R,N_S\}} &\leq n_R  \\
					z_{N_S}^{\{N_S\}} + z_{N_S}^{\{N_S,C_R\}} +  z_{N_S}^{\{N_S,N_R\}} &\leq n_R  \\
					z_{C_R}^{\{C_R,N_S\}} &\leq c_R  \\
					z_{C_S}^{\{C_S,N_R\}} &\leq c_S  \\
					z_{N_R}^{\{N_R\}} &\geq |I_R| \\
					z_{N_S}^{\{N_S\}} &\geq |I_S| \\
					n_{RS} + \min(z_{N_R}^{\{N_R,C_S\}}, z_{C_S}^{\{C_S,N_R\}}) + 
					\min(z_{N_R}^{\{N_R,N_S\}}, z_{N_S}^{\{N_S,N_R\}}) + 
					\min(z_{N_S}^{\{N_S,C_R\}}, z_{C_R}^{\{C_R,N_S\}})
					&\geq |I_{RS}|
				\end{align*}
		\end{enumerate}
	\end{claim}
		\begin{proof}
			We prove the claim informally by explaining the main ideas, because a formal proof would be too long and not that interesting.
			Conditions~(1-3) are easily checked to be necessary. We now explain why condition~(4) is also necessary.
			Suppose that~$P(I_R,I_S,I_{RS})$ is a completion of~$D$.
			Observe that~($\dagger$) to obtain the constants in~$I_{RS}$, we had to use some or all of the following:
		\begin{itemize}
			\item the nulls in~$N_{RS}$; or
			\item the nulls in~$N_R$ together with those in~$N_S$; or
			\item the nulls in~$N_R$ together with the constants in~$C_S$; or
			\item the nulls in~$N_S$ together with the constants in~$C_R$.
		\end{itemize}

			But then, to obtain~$P(I_R,I_S,I_{RS})$ as a completion, we must have used three disjoint (possibly empty) sets $Z_{N_R}^{\{N_R\}}, Z_{N_R}^{\{N_R,C_S\}},  Z_{N_R}^{\{N_R,N_S\}}$ of the nulls in~$N_R$ of sizes~$0\leq z_{N_R}^{\{N_R\}}, z_{N_R}^{\{N_R,C_S\}},  z_{N_R}^{\{N_R,N_S\}} \leq d$, we have done the same for the nulls in~$N_S$ and we also used a subset of the constants of~$C_R$ (and~$C_S$) in such a way that, according to~($\dagger$):
		\begin{itemize}
			\item the nulls in~$Z_{N_R}^{\{N_R\}}$ have been used to obtain the set~$I_R$ (which, we recall, is the set of constants that occur only in~$R$ and that are not in~$C_R$). Note that only the nulls in~$N_R$ could have been used to obtain constants in~$I_R$. This is what the fifth equation expresses.
			\item the nulls in~$Z_{N_R}^{\{N_R,C_S\}}$ have values in~$Z_{C_S}^{\{C_S,N_R\}}$, which gives us constants in~$I_{RS}$. Observe that at maximum we could obtain $\min(z_{N_R}^{\{N_R,C_S\}}, z_{C_S}^{\{C_S,N_R\}})$ constants in this manner.
			\item the nulls in~$Z_{N_R}^{\{N_R,N_S\}}$ and those in~$Z_{N_R}^{\{N_R,N_S\}}$ have common values, which gives us constants in~$I_{RS}$.
				Again, observe that we can get at most~$\min(z_{N_R}^{\{N_R,N_S\}}, z_{N_S}^{\{N_S,N_R\}})$ constants using these.
			\item the nulls in~$Z_{N_S}^{\{N_S,C_R\}}$ have values in~$Z_{C_R}^{\{C_R,N_S\}}$, which gives us constants in~$I_{RS}$. Observe that at maximum we could obtain $\min(z_{N_S}^{\{N_S,C_R\}}, z_{C_R}^{\{C_R,N_S\}})$ constants in this manner.

		\end{itemize}
			The first~$4$ equations express the partitioning process, and
			the last equation then expresses that by combining all these constants we indeed obtained the whole set~$I_{RS}$.

			We now explain why conditions~(2-4) are sufficient. If~$|I_R|,|I_S|$ and~$I_{RS}$ are all~$\geq 1$ then condition~(4) is sufficient, because we can use the nulls and constants as explained above, and we have enough of them to obtain the sets~$I_R,I_S,I_{RS}$. We explain what happens when~$I_R = \emptyset$ for instance. In that case, condition~(1) ensures us that we have either~$n_R=0$ or $C_R \cup C_{RS} \cup I_{RS} \neq \emptyset$. If we have~$n_R=0$ then it is fine, since the only nulls that could be used to fill~$I_R$ are those in~$N_R$. If we have~$n_R \geq 1$ and~$C_R \cup C_{RS} \cup I_{RS} \neq \emptyset$ then we can use these to absorb the values of the nulls in~$N_R$, and we are fine (i.e., we will be able to obtain~$I_R = \emptyset$). We leave it to the reader to complete the small gaps in this proof.
		\end{proof}

	Using this, we have that the value of~$\checkp(I_R,I_S,I_{RS})$ only depends on the sizes of~$I_R,I_S,I_{RS}$, and moreover can be computed in polynomial time
	\begin{claim}
		The value of~$\checkp(I_R,I_S,I_{RS})$ only depends on~$|I_R|,|I_S|,|I_{RS}|,n_R,n_S,n_{RS},c_R,c_S,c_{RS}$, and can be computed in \fp.
	\end{claim}
	\begin{proof}
		The fact that this value only depends on the sizes of these sets is simply by inspection of the conditions in Claim~\ref{claim:expl-conditions}. Conditions~(1-3) can obviously be checked in PTIME. The fact that condition~(4) can be checked in PTIME is because we can test all possible assignments between~$0$ and~$d$ for all these variables and see if there is one assignment that satisfies the equations; indeed, observe that the number of variables is fixed because the query is fixed.
	\end{proof}

	But then, we can express the result as follows
	\[\sum_{0 \leq i_R,i_S,i_{RS} \leq d} \binom{d-c}{i_R} \binom{d-c-i_R}{i_S} \binom{d-c_{RS} - i_R - i_S}{i_{RS}}  \times \checkp(i_R,i_S,i_{RS}) \]
	and we can evaluate this expression as-is in FP because computing~$\checkp(I_R,I_S,I_{RS}) \in \{0,1\}$ is in PTIME by the last claim. This concludes this (long) example.

\begin{toappendix}
	Let~$q$ be an sjfBCQ not containing any of these two patterns. Then, as observed in Section~\ref{subsec:countcompls-uniform},~$q$ is a conjunction of basic singletons query.
	Let~$\sigma = \{R_1,\ldots,R_l\}$ be the set of relation symbols of~$q$, and~$D$ be an incomplete database over these relations, with~$\dom$ the uniform domain of the nulls and~$d$ its size.
	For every~$\mathbf{s} \subseteq \sigma$, $\mathbf{s} \neq \emptyset$, let:
	 \begin{itemize}
		\item $C_\mathbf{s}$ be the set of constants that occur in all relations of~$\mathbf{s}$ and in none of the others;~$c_\mathbf{s}$ be its size;
		\item~$N_\mathbf{s}$ be the set of nulls that occur in all relations of~$\mathbf{s}$ and in none of the others;~$n_\mathbf{s}$ be its size.
	 \end{itemize}
	We also define~$c$ as $\sum_{\emptyset \neq \mathbf{s} \subseteq \sigma} c_\mathbf{s}$.
	We can assume wlog that~$C_\mathbf{s} \subseteq \dom$ for all~$\emptyset \neq \mathbf{s} \subseteq \sigma$, otherwise we can simply remove from~$D$ the corresponding facts. 
	Let~$L \defeq 2^l-1$, and let~$\mathbf{s}_1,\ldots,\mathbf{s}_L$ be an arbitrary linear order of~$\{\mathbf{s} \subseteq \sigma \mid \mathbf{s} \neq \emptyset\}$ (for instance, by non-decreasing size).
	We will follow the same steps as in the example of Section~\ref{subsec:ucc-naive}.
	The following lemma is the generalization of Claim~\ref{claim:expl-disjoint}, and explains how we can guide the computation so that we do not count the same completion twice:

	\begin{lemma}
		\label{lem:disjoint}
		For a tuple~$(I_{\mathbf{s}_1},\ldots,I_{\mathbf{s}_L})$ of subsets of~$\dom$ satisfying~($\star$)
		\[I_{\mathbf{s}} \subseteq (\dom \setminus (C \cup \bigcup_{\substack{\emptyset \neq \mathbf{s'}\subseteq \sigma \\ \mathbf{s'} \neq \mathbf{s}}} I_{\mathbf{s'}})) \cup \bigcup_{\emptyset \neq \mathbf{s'}\subsetneq \mathbf{s}} C_\mathbf{s'}\]
		for every~$\mathbf{s} \in (\mathbf{s}_1,\ldots,\mathbf{s}_L)$ (in other words, all the sets~$I_\mathbf{s}$ are mutually disjoint subsets of~$\dom$, and a set~$I_\mathbf{s}$ can only contain a constant~$b\in C$ if~$b$ is in one of the sets~$C_\mathbf{s'}$ for which~$\mathbf{s'}$ is striclty included in~$\mathbf{s}$),
		let us define~$P(I_{\mathbf{s}_1},\ldots,I_{\mathbf{s}_L})$ to be the complete database consisting of the following facts, for every~$\emptyset \neq \mathbf{s} \subseteq \sigma$:
		\begin{itemize}
			\item $R(a)$ for every~$R \in \mathbf{s}$ and~$a \in I_\mathbf{s}$ or~$a \in C_\mathbf{s} \setminus \bigcup_{\mathbf{s} \supsetneq \mathbf{s'}} I_\mathbf{s'}$
		\end{itemize}
		Then, for every two such tuples~$(I_{\mathbf{s}_1},\ldots,I_{\mathbf{s}_L})$ and~$(I'_{\mathbf{s}_1},\ldots,I'_{\mathbf{s}_L})$ satisfying ($\star$) and that are distinct, we have that~$P(I_{\mathbf{s}_1},\ldots,I_{\mathbf{s}_L}) \neq P(I'_{\mathbf{s}_1},\ldots,I'_{\mathbf{s}_L})$.
	\end{lemma}
	\begin{proof}
		Let us write~$P = P(I_{\mathbf{s}_1},\ldots,I_{\mathbf{s}_L})$
		and~$P' = P(I'_{\mathbf{s}_1},\ldots,I'_{\mathbf{s}_L})$.
		Assume that~$P=P'$, and let us show
		that~$(I_{\mathbf{s}_1},\ldots,I_{\mathbf{s}_L}) =
		(I'_{\mathbf{s}_1},\ldots,I'_{\mathbf{s}_L})$.  Assume by way
		of contradiction that for some~$\emptyset \neq \mathbf{s}
		\subseteq \sigma$ we have~$I_\mathbf{s} \neq I'_\mathbf{s}$. Then
		(wlog) there exists~$a \in I_\mathbf{s} \setminus
		I'_\mathbf{s}$. By the definition of~$P$, we have that~$P$
		contains all the facts~$R(a)$ for~$R \in \mathbf{s}$. Let us
		show that~$P$ does not contain any fact~$R(a)$
		for~$R\notin\mathbf{s}$. Otherwise, assume that~$P$
		contains~$R(a)$ with~$R \notin \mathbf{s}$. Then there
		exists~$\mathbf{s'} \subseteq \sigma$ such that~$R \in
		\mathbf{s'}$ and such that~$a \in I_{\mathbf{s'}} \cup (
		C_\mathbf{s'} \setminus \bigcup_{\mathbf{s'} \supsetneq
		\mathbf{s''}} I_\mathbf{s''})$.
		Since~$\mathbf{s}$ does not contain~$R$ while~$\mathbf{s'}$
		does, we have~$\mathbf{s'} \not\subseteq \mathbf{s}$. But then
		by~($\star$) we have that~$I_\mathbf{s}$ and~$I_{\mathbf{s'}}
		\cup  C_\mathbf{s'}$ are disjoint, which is a contradiction
		because~$a$ is supposed to be in both~$I_\mathbf{s}$
		and~$I_{\mathbf{s'}} \cup ( C_\mathbf{s'} \setminus
		\bigcup_{\mathbf{s'} \supsetneq \mathbf{s''}}
		I_\mathbf{s''})$.  Therefore, it is indeed the case that~$P$
		does not contain any fact~$R(a)$ for~$R\notin\mathbf{s}$.  Now,
		if~$P'$ contains a fact~$R(a)$ for some~$R \notin \sigma$ then
		we are done since this would imply~$P \neq P'$, a
		contradiction. Hence we can assume that~$P'$ does not contain any
		fact~$R(a)$ for~$R \notin \sigma$.  We will now prove
		that~$P'$ does not contain all the facts~$R(a)$ for~$R \in
		\sigma$, thus establishing a contradiction (because~$P$ does, so we would have~$P
		\neq P'$) and concluding this proof.  Assume by contradiction
		that~$P'$ contains all the facts~$R(a)$ for~$R \in \mathbf{s}$.
		First of all, observe that we have~$a \notin C_\mathbf{s}$
		because by~($\star$) we have that~$I_\mathbf{s}$
		and~$C_\mathbf{s}$ are disjoint, and we know that~$a \in
		I_\mathbf{s}$. Hence, the only way in which~$P'$ could contain
		all the facts~$R(a)$ for~$R \in \mathbf{s}$ is if there
		exist~$\mathbf{s'}_1,\ldots,\mathbf{s'}_k$ with~$k \geq 1$
		and~$\mathbf{s'}_j \subsetneq \mathbf{s}$ for~$1 \leq j\leq k$
		such that~$\bigcup_{1 \leq j \leq k} \mathbf{s'}_j =
		\mathbf{s}$ and such that for every~$1 \leq j \leq k$ we have
		that~(i) $a \in I_{\mathbf{s'}_j} \cup ( C_{\mathbf{s'}_j}
		\setminus \bigcup_{\mathbf{s'}_j \supsetneq \mathbf{s''}} I_\mathbf{s''})$.  Observe that there must
		exist~$1 \leq j_1,j_2 \leq k$ such that~$\mathbf{s'}_{j_1}$
		and~$\mathbf{s'}_{j_2}$ are incomparable by inclusion
		(otherwise, since all~$\mathbf{s}_j$ are strictly included
		in~$\mathbf{s}$, their union could not be equal
		to~$\mathbf{s}$). Also observe that by~($\star$) we have that
		the sets~$I_{\mathbf{s'}_{j_1}} \cup C_{\mathbf{s'}_{j_1}}$ and
		$I_{\mathbf{s'}_{j_2}} \cup C_{\mathbf{s'}_{j_2}}$ must be
		disjoint. But then~(i) applied to~$j_1$ and~$j_2$ gives a
		contradiction (namely, these two sets are not disjoint since
		they both contain~$a$).  This finishes the proof.
	\end{proof}

	This next Lemma generalizes Claim~\ref{claim:expl-have-all} and tells us that by summing over all such tuples~$(I_{\mathbf{s}_1},\ldots,I_{\mathbf{s}_L})$ we cannot miss a completion of~$D$:
	\begin{lemma}
		\label{lem:have-all}
		Let~$D'$ be a completion of~$D$. Then there exists a tuple~$(I_{\mathbf{s}_1},\ldots,I_{\mathbf{s}_L})$ of subsets of~$\dom$ satisfying~($\star$) such that~$D' = P(I_{\mathbf{s}_1},\ldots,I_{\mathbf{s}_L})$.
	\end{lemma}
	\begin{proof}
		For~$\emptyset \neq \mathbf{s} \subseteq \sigma$, let us define~$D_\mathbf{s}$ to be the set of constants that occur in all relation of~$\mathbf{s}$ and in none of the others. Define the set~$I_\mathbf{s}$ for~$\emptyset \neq \mathbf{s} \subseteq \sigma$ as follows: $I_\mathbf{s} \defeq D_\mathbf{s} \setminus C_\mathbf{s}$.
		It is then routine to check that~$(I_{\mathbf{s}_1},\ldots,I_{\mathbf{s}_L})$ satisfies~($\star$) and is such that~$D' = P(I_{\mathbf{s}_1},\ldots,I_{\mathbf{s}_L})$.
	\end{proof}

	Lemma~\ref{lem:disjoint} and~\ref{lem:have-all} allows us to express the result as

	\begin{equation}
		\sum_{I_{\mathbf{s}_1} \subseteq \dom \setminus C} \ldots 
		\sum_{I_{\mathbf{s}_j} \subseteq (\dom \setminus (C \cup \bigcup_{1 \leq k < j} I_{\mathbf{s}_k})) \cup \bigcup_{\emptyset \neq \mathbf{s'}\subsetneq \mathbf{s}} C_\mathbf{s'}} \ldots
		\sum_{I_{\mathbf{s}_L} \subseteq \dom \setminus (C_{\mathbf{s}_L} \cup \bigcup_{1 \leq k < L} I_{\mathbf{s}_k}) }
		\checkp(I_{\mathbf{s}_1},\ldots,I_{\mathbf{s}_L})
	\end{equation}

	where~$\checkp(I_{\mathbf{s}_1},\ldots,I_{\mathbf{s}_L}) \in \{0,1\}$ is defined by 
	$$\checkp(I_{\mathbf{s}_1},\ldots,I_{\mathbf{s}_L})\defeq \begin{cases}
		1 \text{ if } P(I_{\mathbf{s}_1},\ldots,I_{\mathbf{s}_L}) \text{ is a completion of } D \text{ that satisfies } q \\
		0 \text{ otherwise}
	\end{cases}.$$
	
	As such we cannot evaluate this expression in \ptime.
	The next step is to show that the value of~$\checkp(I_{\mathbf{s}_1},\ldots,I_{\mathbf{s}_L})$ only depends on~$(|I_{\mathbf{s}_1}|,\ldots,|I_{\mathbf{s}_L}|)$, which would allow us to rewrite the result as
	\begin{equation}
		\label{eq:final}
		\sum_{0 \leq i_{\mathbf{s}_1},\ldots,i_{\mathbf{s}_L} \leq d} \prod_{1 \leq j \leq L} 
		\binom{d - c - \sum_{1\leq k < j} i_{\mathbf{s}_k} + \sum_{\emptyset \neq \mathbf{s'}\subsetneq \mathbf{s}} C_\mathbf{s'}}{i_{\mathbf{s}_j}}
		\times \checkp(i_{\mathbf{s}_1},\ldots,i_{\mathbf{s}_L})
	\end{equation}

	We give here the necessary and sufficient conditions for~$P(I_{\mathbf{s}_1},\ldots,I_{\mathbf{s}_L})$ to be a completion of~$D$ that satisfies~$q$.

	\begin{lemma}
	We have~$\checkp(I_{\mathbf{s}_1},\ldots,I_{\mathbf{s}_L}) = 1$ if and only if the following conditions hold:
		\begin{enumerate}
			\item for every basic singleton query~$C_i(x)$ of~$q$, letting~$\mathbf{s}$ be its sets of relation symbols,
				there exists~$\mathbf{s} \subseteq \mathbf{s'} \subseteq \sigma$ such that we have~$|I_\mathbf{s'}|\geq 1$ or~$c_\mathbf{s'}\geq 1$.
			\item for every~$\emptyset \neq \mathbf{s} \subseteq \sigma$, if~$n_\mathbf{s} \geq 1$ and
				$|\bigcup_{\mathbf{s'} \supseteq \mathbf{s}} C_\mathbf{s'} \cup \bigcup_{\mathbf{s'} \supsetneq \mathbf{s}} I_\mathbf{s'}| = 0$ then~$|I_\mathbf{s}| \neq 0$.
			\item consider the following system of equations, with integer variables between~$0$ and~$d$:
		\begin{itemize}
			\item for every two sets~$A,A'$ of subsets of~$\{\emptyset \neq \mathbf{s} \subseteq \sigma \}$, we have a variable~$z_{N_\mathbf{s}}^{A,A'}$ for every~$\mathbf{s} \in A$ and a variable~$z_{C_\mathbf{s}}^{A,A'}$ for every~$\mathbf{s} \in A'$.
				For instance if~$\sigma = \{R,S,T,U\}$ and if~$A = \{\{R,S\},\{S,T\}\}$ and~$A' = \{\{U\}\}$ we have the variables 
				$$z^{\{\{R,S\},\{S,T\}\},\{\{U\}\}}_{N_{\{R,S\}}} \text{ and } z^{\{\{R,S\},\{S,T\}\},\{\{U\}\}}_{N_{\{S,T\}}} \text{  and } 
				z^{\{\{R,S\},\{S,T\}\},\{\{U\}\}}_{C_{\{U\}}}.$$ 
				The intuition is that we will use $z^{\{\{R,S\},\{S,T\}\},\{\{U\}\}}_{N_{\{R,S\}}}$ of the nulls in~$N_{\{R,S\}}$ and combine them with $z^{\{\{R,S\},\{S,T\}\},\{\{U\}\}}_{N_{\{S,T\}}}$ of the nulls in~$N_{\{S,T\}}$ and with $z^{\{\{R,S\},\{S,T\}\},\{\{U\}\}}_{C_{\{U\}}}$ of the constants in~$C_{\{U\}}$ in order to obtain constants in~$I_{\{R,S,T,U\}}$. Let us write~$V$ this set of variables.
				(we note here that we are using sligthly different notation than for the example in Section~\ref{subsec:ucc-naive}; this is for readability reasons only.)
			\item Now, for every~$\emptyset \neq \mathbf{s} \subseteq \sigma$ we have the constraint
				\[\sum_{z_{N_\mathbf{s}}^{A,A'} \in V} z_{N_\mathbf{s}}^{A,A'} \leq n_\mathbf{s}\]
				as well as the constraint
				\[\sum_{z_{C_\mathbf{s}}^{A,A'} \in V} z_{N_\mathbf{s}}^{A,A'} \leq c_\mathbf{s}\]
				intuitively expressing that we do not use more nulls and constants than there are available.
			\item for every~$\emptyset \neq \mathbf{s} \subseteq \sigma$ we have a constraint

				\[\sum_{\substack{A,A' \subseteq \{\emptyset \neq \mathbf{s} \subseteq \sigma\}\\A\cup A' = \mathbf{s}}} \min_{z_{*}^{A,A'} \in V} z_{*}^{A,A'}
				\geq I_\mathbf{s}
				\]
								intuitively meaning that we have allocated the groups of nulls and constants in a way that allows us to fill the set~$I_\mathbf{s}$.
		\end{itemize}
		Then this system of equations must have a solution.
		\end{enumerate}
	\end{lemma}
	\begin{proof}
		The idea is the same as in Claim~\ref{claim:expl-conditions}. The only difference is that we added condition (1), which ensures that the guessed completion indeed satisfies the query.
	\end{proof}

	As in the example of Section~\ref{subsec:ucc-naive}, this implies that the value of~$\checkp(I_{\mathbf{s}_1},\ldots,I_{\mathbf{s}_L})$ only depends on~$(|I_{\mathbf{s}_1}|,\ldots,|I_{\mathbf{s}_L}|)$ and can be computed in FP (by testing all assignments of the~$z_*^*$ variables; because the schema is fixed so there are only a fixed number of such variables).
	But then we can compute the result in FP by evaluating the expression~\ref{eq:final}, which finishes the proof. 
\end{toappendix}

\section{Approximating the numbers of valuations and completions}
\label{sec:approx} 
As we saw in the previous sections, counting valuations and completions of an incomplete database are usually intractable problems. 
However, this does not necessarily rule out the existence of efficient approximation algorithms for such counting problems, in particular if some source of randomization is allowed. 
In this section, we investigate this question by focusing on the well-known notion of Fully Polynomial-time Randomized Approximation Scheme (FPRAS) for counting problems~\cite{jerrum1986random}.
Formally, 
let~$\Sigma$ be a finite alphabet and $f : \Sigma^* \to \mathbb{N}$ be a counting problem. 
Then~$f$
is said to have an FPRAS if there is a randomized algorithm~$\mathcal{A} : \Sigma^* \times (0,1) \to \mathbb{N}$ and a polynomial~$p(u,v)$ such that, given~$x \in \Sigma^*$ and~$\epsilon \in (0,1)$,
algorithm~$\mathcal{A}$ runs in time~$p(|x|,\nicefrac{1}{\epsilon})$ and satisfies the following condition:
\begin{eqnarray*}
\Pr\big(|f(x) - \mathcal{A}(x,\epsilon)| \, \leq \, \epsilon f(x)\big) & \geq &\frac{3}{4}.
\end{eqnarray*}
Observe that the property of having an FPRAS is closed under polynomial-time parsimonious reductions, that is, if we have an FPRAS for a counting problem~$A$ and for counting problem~$B$ we have that~$B \pr A$, then we also have an FPRAS for~$B$.

In the following sections, we investigate the existence of FPRAS for the problems of counting valuations and completions of an incomplete database. 
The overall picture that we obtain is shown in Table~\ref{tab:fpras-count}. 
We first deal with counting valuations in Section~\ref{sec:approx-countvals}, where we show a general condition under which
this problem has an FPRAS (which will apply, in particular, to all Boolean conjunctive queries). Then, in Section~\ref{sec:approx-countcompls}, we show that the situation is quite different for counting completions, as in most cases this problem does not admit an FPRAS.

\subsection{Approximating the number of valuations}
\label{sec:approx-countvals}

To prove the main result of this section, we need to consider the 
counting complexity class \spanl~\cite{alvarez1993very}. Given a finite alphabet~$\Sigma$, an \nl-transducer~$M$ over~$\Sigma$ is a nondeterministic Turing Machine with input and output alphabet~$\Sigma$, a read-only input tape, a write-only output tape (where the head is always moved to the right once a symbol in~$\Sigma$ is written on it, so that the output cannot be read by~$M$), and a work-tape of which, on input~$x$, only the first~$c \cdot \log(|x|)$ cells can be used for a fixed constant~$c > 0$ (so that the space used by~$M$ is logarithmic). Moreover,~$y \in \Sigma^*$ is said to be an output of~$M$ with input~$x$, if there exists an accepting run of~$M$ with input~$x$ such that~$y$ is the string in the output tape when~$M$ halts. Then a function $f: \Sigma^* \to \mathbb{N}$ is said to be in \spanl\ if there exists an \nl-transducer~$M$ over~$\Sigma$ such that for every~$x \in \Sigma^*$, the value~$f(x)$ is equal to the number of distinct outputs of~$M$ with input~$x$. 
In \cite{alvarez1993very}, it was proved that $\spanl \subseteq \shp$, and also that this inclusion is strict unless~$\nl = \np$.

Very recently, the authors of~\cite{arenas2019efficient} have shown 
that every problem in \spanl\ has an~FPRAS.
\begin{theorem}[{\cite[Corollary 3]{arenas2019efficient}}]
\label{thm:spanl-fpras}
	Every problem in \spanl\ has an FPRAS.
\end{theorem}
By using this result, we can give a general condition on a Boolean query~$q$ under which~$\countvals(q)$ has an FPRAS, as this condition ensures that~$\countvals(q)$ is in \spanl. More precisely,
a Boolean query~$q$ is said to be \emph{monotone} if for every pair of (complete) databases~$D$, $D'$ such that~$D \subseteq D'$, if $D \models q$, then $D' \models q$. Moreover,~$q$ is said to have \emph{bounded minimal models} if there exists a constant~$C_q$ (that depends only on~$q$) satisfying that for every (complete) database~$D$, if~$D \models q$, then there exists
$D'\subseteq D$ such that~$D' \models q$ and the number of facts in~$D'$ is at most~$C_q$.
Finally, the \emph{model checking problem} for~$q$, denoted by~$\mc(q)$, is the problem of deciding, given a (complete) database~$D$, whether~$D \models q$. Then~$q$ is said to have a model checking in a complexity class~$\CC$ if $\mc(q) \in \CC$. With this terminology, we can state the main result of this section. 
\begin{proposition}
\label{prp:countvals-in-spanl}
\begin{sloppypar}
	Assume that a Boolean query~$q$ is monotone, has model checking in nondeterministic linear space, and has bounded minimal models.
	Then~$\countvals(q)$
	is in~\spanl.
\end{sloppypar}
\end{proposition}
\begin{proof}
	Let~$D$ be the input incomplete database, with the domains for each null.
	First, the machine guesses a subset~$D'\subseteq D$ of size~$\leq C_q$, such that each fact of~$D'$ is over a relation symbol that appears in~$q$.
	Observe that~$D'$ contains at most~$|D'| \times \arity(q) \leq C_q \times \arity(q)$ distinct nulls, and that this is a constant.
	The machine then guesses and remembers a valuation~$\nu$ of~$D'$ and computes~$\nu(D')$.
	The encoding size~$||\nu(D')||$ of~$\nu(D')$ is~$O(\log |D|)$, so the machine can check in nondeterministic linear space
	whether~$\nu(D') \models q$, and stops and rejects in the branches that fail the test.
	Then, the machine reads the input tape left to right and for every occurrence of a null~$\bot$ (appearing in~$D$) that it finds, it does the following:
	\begin{itemize}
		\item It checks whether~$\bot$ appears before on the input tape and if so it simply continues;
		\item Else if~$\bot$ does not appear before on the input tape but appears in~$D'$ then the machine writes~$\nu(\bot)$ on its output tape;
		\item Else if~$\bot$ does not appear before on the input tape and does not appear in~$D'$ then it guesses a
			value for it and writes that value on the output tape (but it does not remember that value).
	\end{itemize}
	It is easy to see that this procedure can be carried out by a logspace nondeterministic transducer, so we only need to show that the distinct outputs of the machine
	correspond exactly to the distinct valuations~$\nu$ of~$D$ such that~$\nu(D) \models q$.
	Since the machine writes values for nulls in order of first appearance on the input tape, it is clear that every valuation is outputted exactly once.
	Let~$\nu$ be a valuation that is outputted, and let~$D'$ be the subdatabase such that~$\nu(D') \models q$. Since~$\nu(D') \subseteq \nu(D)$ and~$q$
	is monotone, we have~$\nu(D)\models q$. Inversely, let~$\nu$ be a valuation of~$D$ such that~$\nu(D) \models q$, and let us show that it must be outputted.
	Since~$\nu(D) \models q$ and~$q$ has bounded minimal models, there exists~$D_\nu \subseteq \nu(D)$ of size~$\leq C_q$ such that~$D_\nu \models q$.
	But~$D_\nu$ is~$\nu(D')$ for some~$D' \subseteq D$ of size~$\leq C_q$.
	Then it is clear that one of the branches of the machine has guessed~$D'$
	and then~$\nu_{|D'}$ and then has written~$\nu$ on the output tape.
\end{proof}

In particular, given that a union of Boolean of conjunctive queries satisfies the three properties of the previous proposition, we conclude from 
Theorem~\ref{thm:spanl-fpras}
that $\countvals(q)$ can be efficiently approximated by using a randomized algorithm if~$q$ is a union of BCQs.\footnote{As a matter of fact, this holds even for the larger class of unions of BCQs with {\em inequalities} (that is, atoms of the form $x \neq y$), as 
	such queries also satisfy the aforementioned three properties.}

\begin{corollary}
\label{cor:countvals-has-fpras}
	If~$q$ is a union of BCQs, then~$\countvals(q)$ has an FPRAS (and the same holds if restricted to the uniform setting and/or to Codd tables).
\end{corollary}

We prove in the next section that the good properties stated in 
Proposition~\ref{prp:countvals-in-spanl} do not hold for counting completions.

\subsection{Approximating the number of completions}
\label{sec:approx-countcompls}

In this section, we prove that the problem of counting completions of an incomplete database is much harder in terms of approximation than the problem of counting valuations. In this investigation, two randomized complexity classes play a fundamental role. Recall that~$\rp$ is the class of decision problems~$L$ for which there exists a polynomial-time probabilistic Turing Machine~$M$ such that: (a) if~$x \in L$, then~$M$ accepts with probability at least~$\nicefrac{3}{4}$; and (b) if~$x \not\in L$, then~$M$ does not accept~$x$. Moreover,~$\bpp$ is defined exactly as~$\rp$ but with condition (b) replaced by: (b') if~$x \not\in L$, then~$M$ accepts with probability at most~$\nicefrac{1}{4}$. Thus,~$\bpp$ is defined as~$\rp$ but allowing errors for both the elements that are and are not in~$L$. It is easy to see that~$\rp \subseteq \bpp$. Besides, it is known that~$\rp \subseteq \np$, and this inclusion is widely believed to be strict. Finally, it is not known whether~$\bpp \subseteq \np$ or~$\np \subseteq \bpp$, but it is widely believed that~$\np$ is not included in~$\bpp$.

\paragraph{{\bf The non-uniform case}}
Recall that \sIS\ is the problem of counting the number of independent sets of a graph. This problem
will play a fundamental role when showing non-approximability of counting completions in the non-uniform case. More precisely, the following is known about the approximability of \sIS.
\begin{theorem}[{\cite[Theorem 3.1]{dyer2002counting}}]
\label{thm:IS-no-fpras}
	The problem \sIS\ does not admit an FPRAS unless $\np=\rp$. 
\end{theorem}
In the proof of Proposition~\ref{prp:countcompls}, we considered the problem $\sVC$ of counting the number of vertex covers of a graph $G=(V,E)$, and showed that $\sVC \pr \ccountcompls(R(x))$. By observing that $S \subseteq V$ is an independent set of~$G$ if and only if~$V \setminus S$ is a vertex cover of~$G$, we can conclude that $\sIS(G) = \sVC(G)$ and, thus, the same reduction from the proof of Proposition~\ref{prp:countcompls} establishes that $\sIS \pr \ccountcompls(R(x))$.
Therefore, from the fact that 
the reduction in Lemma~\ref{lem:pattern-parsimonious-compls} is also parsimonious and preserves the property of being a Codd table, and the fact that the existence of an FPRAS is closed under polynomial-time parsimonious reductions, we obtain the following result from Theorem~\ref{thm:IS-no-fpras}.

\begin{theorem}[Dichotomy]
\label{thm:countcompls-sjfcq}
	For every \sjfbcq~$q$, it holds that $\ccountcompls(q)$ does not admit an FPRAS unless $\np = \rp$ $($and, hence, the same holds for~$\countcompls(q)$$)$.\footnote{Again, we remind the reader that, to avoid trivialities, we assume all \sjfbcqs\ to contain at least one atom and all atoms to have at least one variable.}
\end{theorem}

\paragraph{{\bf The uniform case}}
Recall that from Theorem~\ref{thm:ucc-naive-fp}, we know that if an \sjfbcq~$q$ contains neither $R(x,x)$ nor $R(x,y)$ as a pattern, then $\ucountcompls(q)$ is in \fp. Thus, the question to answer in this section is whether~$\ucountcompls(q)$ and~$\cucountcompls(q)$ can be efficiently approximated if~$q$ contains any of these two patterns. 
For the case of naive tables, we will give a negative answer 
to this question.
Notice that, this time, our reduction from~\sIS\ in Proposition~\ref{prp:Rxx-Rxy-hard-compls-naive} is not parsimonious, so we cannot use Theorem~\ref{thm:IS-no-fpras} as we did for the non-uniform case.
Instead, we will rely on the following well-known fact: if there exists a \bpp\ algorithm for a problem that is \np-complete,
then 
$\np \subseteq \bpp$, which in turn implies that $\np = \rp$~\cite{K82}.

\begin{proposition}
	\label{prp:ucountcompls-Rxy-no-frpas}
	Neither~$\ucountcompls(R(x,x))$ nor $\ucountcompls(R(x,y))$ admits an FPRAS unless $\np=\rp$. This holds even in the restricted setting 
		in which all nulls are interpreted over the same fixed domain~$\{1,2,3\}$.
\end{proposition}

\begin{proof}
	Let~$G=(V,E)$ be a graph. First, we explain how to construct an incomplete database~$D$ containing a single binary relation~$R$, with uniform domain~$\{1,2,3\}$, and such that (a) all completions of~$D$ satisfy both queries; (b) if~$G$ is~$3$-colorable then~$D$ has~$8$ completions; and (c) if~$G$ is not~$3$-colorable then~$D$ has~$7$ completions. For every node~$u\in V$ we have a null~$\bot_u$.
	The database~$D$ consists of the following three disjoint sets of facts:
		\begin{itemize}
			\item For every edge~$\{u,v\} \in E$, we have the two facts~$R(\bot_u,\bot_v)$ and~$R(\bot_v,\bot_u)$; we call these the \emph{coding facts}.
			\item We have the facts~$R(1,2),R(2,1),R(2,3),R(3,2),R(1,3)$, and $R(3,1)$; we call these the \emph{triangle facts};
			\item We have six fresh nulls~$\bot_1,\bot'_1,\bot_2,\bot'_2,\bot_3,\bot'_3$ and the facts $R(\bot_i,\bot'_i)$ and $R(\bot'_i,\bot_i)$ for~$1 \leq i \leq 3$; we call these the \emph{auxiliary facts};
			\item Last, we have a fact~$R(c,c)$, where~$c$ is a fresh constant.
		\end{itemize}
		It is clear that all the completions of~$D$ satisfy both queries (thanks to the fact~$R(c,c)$), so we only need to prove~(b) and~(c).
		Observe that a candidate completion of~$D$ can be equivalently seen as an undirected graph, \emph{possibly with self-loops}, over the nodes~$\{1,2,3\}$ (we omit the fact~$R(c,c)$ since it is in every completion) and that contains the triangle. Thanks to the auxiliary facts, it is easy to show that all such graphs with at least one self-loop can be obtained as a completion of~$D$. For instance, the completion that is triangle with a self-loop only on~$1$ can be obtained by assigning~$1$ to all the nulls in the coding facts, assigning~$1$ to~$\bot_1$, $\bot'_1$,~$\bot_2$ and~$\bot_3$ and assigning~$2$ to~$\bot'_2$ and~$\bot'_3$.
		There are~$7$ such completions in total. Then, the completion whose graph is the triangle with no self-loops is obtainable if and only if~$G$ is~$3$-colorable (we assign a~$3$-coloring to the nulls in the coding facts, and assign~$1$ to~$\bot_i$ and~$2$ to~$\bot'_i$ for every~$i \in \{1,2, 3\}$). This indeed proves~(b) and~(c).
		Next, we show that any FRPAS with~$\epsilon=\nicefrac{1}{16}$ for counting the number of completions of~$D$ would yield a \bpp\ algorithm to solve $3$-colorability, thus implying~$\np=\rp$ since 
		$3$-colorability
		is an \np-complete problem.

	Let~$\mathcal{A}$ be an FPRAS for~$\ucountcompls(q)$, with~$q$ being $R(x,x)$ or~$R(x,y)$, and let us define a \bpp\ algorithm~$\mathcal{B}$ for~$3$-colorability using~$\mathcal{A}$. On input graph~$G$, algorithm $\mathcal{B}$ does the following. First, it computes in polynomial time the naive table~$D$ as described above.
	Then~$\mathcal{B}$ calls~$\mathcal{A}$ with input~$(D,\nicefrac{1}{16})$, and if~$\mathcal{A}(D,\nicefrac{1}{16}) \geq 7.5$, then~$\mathcal{B}$ accepts, otherwise $\mathcal{B}$ rejects.
	We now prove that~$\mathcal{B}$ is indeed a \bpp\ algorithm for~$3$-colorability. Assume first that~$G$ is~$3$-colorable. Then by~(b) and by definition of what is an FPRAS, we have that~$\Pr\big(|8 - \mathcal{A}(D,\nicefrac{1}{16})| \leq \nicefrac{8}{16} \big) \geq \frac{3}{4}$. This implies in particular that
	$\Pr\big(\mathcal{A}(D,\nicefrac{1}{16}) \geq 8 -\nicefrac{8}{16} \big) \geq \frac{3}{4}$. Since~$8-\nicefrac{8}{16} = 7.5$ we conclude that if~$G$ is~$3$-colorable, then~$\mathcal{B}$ accepts with probability at least~$\nicefrac{3}{4}$.
	Next, assume that~$G$ is not~$3$-colorable. Then by~(c) we have that $\Pr\big(|7 - \mathcal{A}(D,\nicefrac{1}{16})| \leq \nicefrac{7}{16} \big) \geq \frac{3}{4}$. This implies in particular that~$\Pr\big(\mathcal{A}(D,\nicefrac{1}{16}) \leq 7 +\nicefrac{7}{16} \big) \geq \frac{3}{4}$. Since~$7+\nicefrac{7}{16} < 7.5$, this implies in particular that~$\Pr\big(\mathcal{A}(D,\nicefrac{1}{16}) < 7.5 \big) \geq \frac{3}{4}$. From this, we conclude that if~$G$ is not~$3$-colorable, then~$\mathcal{B}$ rejects with probability at least~$\nicefrac{3}{4}$. This concludes the proof of the~proposition.
\end{proof}

By observing again that the reduction in Lemma~\ref{lem:pattern-parsimonious-compls} is parsimonious, and that the existence of an FPRAS is closed under parsimonious reductions, we obtain that $\ucountcompls(q)$ cannot be efficiently approximated if $q$ contains $R(x,x)$ or $R(x,y)$ as a pattern.

\begin{theorem}[Dichotomy]
	Let $q$ be an \sjfbcq. If $q$ has $R(x,x)$ or $R(x,y)$ as a pattern, then $\ucountcompls(q)$ does not admit an FPRAS unless $\np=\rp$. Otherwise, 
	this problem 
	is in~\fp\ $($by Theorem~\ref{thm:ucc-naive-fp}$)$.
\end{theorem}

We do not know if this result still holds for Codd tables, or if it is possible to design an FPRAS in this setting.
We leave this question open for future research.

\section{On the general landscape: beyond~\shp}
\label{sec:misc}
Recall that, when studying the complexity of counting completions for \sjfbcqs\ in Section~\ref{sec:countcompls-sjfcqs}, we did
claim that these problems are in~\shp\ for Codd tables, but that we did not claim so for naive tables.
The goal of this section is then threefold. 
First, we want to prove that the problem of counting completions in indeed in~\shp\ for Codd tables.
Second, we want to give formal evidence that we indeed could not show membership in \shp\ for naive tables.
Third, we want to identify a counting complexity class that is more appropriate to describe the complexity of $\countcompls(q)$.
We deal with these three objectives in the next three sections.

\subsection{Membership in \shp\ of~$\ccountcompls(q)$}
	\label{subsec:countcompls-codd-sharp-p}

In this section, and as promised in the proofs of Propositions~\ref{thm:countcompls-sjfcqs-complete-codd} and~\ref{thm:ucc-codd-fp}, we show that for any Boolean query~$q$, if the model checking problem for~$q$ (denoted~$\mc(q)$, recall the definition from Section~\ref{sec:approx-countvals}) is in \ptime\ then the problem of counting completions for~$q$ under Codd tables is in~\shp. 

\begin{proposition}
	\label{prp:countcompls-codd-sharp-p}
	If a Boolean query~$q$ has the property that model checking is in \ptime, then we have that $\ccountcompls(q)$ is in~\shp.
\end{proposition}

	We recall that a fact that contains only constants is a \emph{ground fact}.
	To show Proposition~\ref{prp:countcompls-codd-sharp-p}, we first prove that we can check in polynomial time if a given set of ground facts is a possible completion of an incomplete database:
	\begin{lemma}
		\label{lem:matchings}
		Given as input an incomplete Codd table~$D$ and a set~$S$ of ground facts, we can decide in polynomial time whether there exists a valuation~$\nu$ of~$D$ such that~$\nu(D) = S$.
	\end{lemma}
	\begin{proof}
		For every fact~$f$ of~$D$, let us denote by~$P(f)$ the
		set of ground facts that can be obtained from~$f$ via a valuation ($P(f)$ can
		be~$\{f\}$ if~$f$ is already a ground fact).
		The first step is to check that for every fact~$f$ of~$D$, it holds
		that ($\star$)~$P(f) \cap S \neq \emptyset$. If this is not the case, then we know for sure that for every
		valuation~$\nu$ of~$D$ we will have~$\nu(D) \not\subseteq S$, so that we can safely reject.
		Next, we build the bipartite graph~$G_{D,S}$ defined as follows:
		the nodes in the left partition of~$G_{D,S}$ are the facts of~$D$, the nodes in the right partition
		are the facts in~$S$, and we connect a fact~$f$ of~$D$ with all the ground facts in the right
		partition that are in~$S \cap P(f)$.
		It is clear that we can construct~$G_{D,S}$ in polynomial time.
		We then compute in polynomial time the size~$m$ of a maximum-cardinality matching of~$G_{D,S}$,
		for instance using~\cite{edmonds1965paths}. It is clear that we have~$m \leq |S|$.
		At this stage, we claim that there exists a valuation~$\nu$ of~$D$ such that~$\nu(D) = S$
		if and only if~$m= |S|$.
		We prove this by analysing the two possible cases:
		\begin{itemize}
			\item If~$m < |S|$, then let us show that there is no such valuation. Indeed,
				assume by way of contradiction that such a valuation~$\nu$ exists.
				Let~$B$ be a subset of~$D$ of minimal size such that~$\nu(B) = S$.
				It is clear that such a subset exists, and moreover that its size is
				exactly~$|S|$. But then, consider the set~$M$ of edges of~$G_{D,S}$
				defined by~$M \defeq \{(f,\nu(f)) \mid f \in B\}$.
				Then~$M$ is a matching of~$G_{D,S}$ of size~$|S|$,
				contradicting the fact that~$m$ is the size of a maximum-cardinality matching.
			\item If~$m = |S|$, let us show that such a valuation exists.
				Let~$M$ be a matching of~$G_{D,S}$
				of size~$|S|$.
				It is clear that every node corresponding to a ground fact~$f \in S$ is incident to (exactly)
				one edge of~$M$; let us denote that edge by~$e_f$.
				Moreover, since~$M$ is a matching, the mapping that associates to a ground fact~$f \in S$ the
				fact~$f'_f$ at the other end of~$e_f$ is injective.
				Hence, we can define~$\nu(\bot)$ of every null~$\bot$ occurring in such a fact $f'_f \in D$
				to be the unique constant such that~$\nu(f'_f) = f$ holds, and
				for every other fact~$f'$ in~$D$ not incident to an edge in~$M$, we chose a value for its nulls
				so that~$\nu(f') \in S$,
				which we can do thanks to~($\star$).
				It is then clear that we have~$\nu(D) = S$.
		\end{itemize}
		But then, we can simply accept if~$m = |S|$ and reject otherwise.
	\end{proof}

	We can now prove Proposition~\ref{prp:countcompls-codd-sharp-p}:

	\begin{proof}[Proof of Proposition~\ref{prp:countcompls-codd-sharp-p}.]
	We define a non-deterministic turing machine~$M_q$ such that, given as input an incomplete Codd table~$D$,
	its number of accepting computation paths is exactly the number of completions of~$D$
	that satisfy~$q$. 
	First, compute in polynomial time the set~$A \defeq \bigcup_{f \in D} P(f)$, where~$P(f)$ is defined just as in Lemma~\ref{lem:matchings}.
	Then, the machine~$M_q$ guesses a subset~$S$ of~$A$.
	It then checks in polynomial time if~$S$, when seen as a database, satisfies~$q$, and rejects if it is not the case.
		Then, using Lemma~\ref{lem:matchings}, it checks in polynomial time
		whether there exists a valuation~$\nu$ of~$D$ such that~$\nu(D) = S$, and accepts iff this is the case.
		It is then clear that~$M_q$ satisfies the conditions, which shows that~$\ccountcompls(q)$
		is in~\#P.
	\end{proof}

In the next two sections, all upper bounds will be proved for the most general scenario 
of non-uniform naive tables,
while all lower bounds will 
be proved for the most restricted scenario of uniform naive tables with 
a fixed domain.

\subsection{Non-membership in \shp\ of $\countcompls(q)$}

We now want to show that~\shp\ is not the right complexity class for problems of the form~$\countcompls(q)$, over naive tables.
One could try to show membership in \shp\ of~$\countcompls(q)$ as we did in the proof of Proposition~\ref{prp:countcompls-codd-sharp-p}; that is, guess a set of ground facts, then check in polynomial time that it satisfies the query and that it is a possible completion of the incomplete database. However, this proof strategy does not work, as we show next that checking if a set of ground fact is a completion of an incomplete database is an \np-complete problem. Moreover, this holds already for a fixed schema containing only a binary relation and for a fixed set of ground facts.

\begin{proposition}
\label{prp:check-completion-np-c}
The exist a set~$S$ of ground facts over binary relation~$R$ such that the following is \np-complete: given as input an incomplete database~$D$ over~$R$, decide if~$S$ is a completion of~$D$.
\end{proposition}
\begin{proof}
We reduce from $3$-colorability. Given a graph~$G$, we build the same incomplete database~$D$ as in the proof of Proposition~\ref{prp:ucountcompls-Rxy-no-frpas}, and the (fixed) set of ground facts 
is the triangle, that is, $S=\{R(1,2),R(2,1),R(2,3),R(3,2),R(1,3)$, $R(3,1)\}$. Then, as in that proof, we have that~$S$ is a completion of~$D$ if and only if~$G$ is~$3$-colorable.
\end{proof}

This does not, however, constitute a proof that~$\countcompls(q)$ is not in~\shp\ (but it is a good hint).
To prove that~$\countcompls(q)$ is unlikely to be in \shp, we need to define the complexity class \spp\ introduced in \cite{gupta1991power,OH93,FFK94}. Given a nondeterministic Turing Machine $M$ and a string $x$, let $\acc_M(x)$ (resp., $\rej_M(x)$) be the number of accepting (resp., rejecting) runs of $M$ with input $x$, and let $\gap_M(x) = \acc_M(x) - \rej_M(x)$.
Then a language $L$ is said to be in \spp~\cite{FFK94} if there exists a polynomial-time nondeterministic Turing Machine $M$ such that: (a) if $x \in L$, then $\gap_M(x) = 1$; and (b) if $x \not\in L$, then $\gap_M(x) = 0$. 
It is conjectured that $\np \not\subseteq \spp$ as, for example, for every known polynomial-time nondeterministic Turing Machine $M$ accepting an \np-complete problem, the function $\gap_M$ is not bounded. In the following proposition, we show how this conjecture helps us to reach our second goal.

\begin{proposition}
	\label{prp:not-in-shp}
	There exists an \sjfbcq~$q$ such that $\ucountcompls(q)$ is not in \shp\ unless $\np \subseteq \spp$.
\end{proposition}

The proof of this result relies on the proof of Theorem~\ref{thm:general-spanp-complete}, in the next section (we presented the results in this order for narrative purposes). We will then defer its presentation until the proof of Theorem~\ref{thm:general-spanp-complete} is given. 

\subsection{An appropriate counting complexity class for~$\countcompls(q)$: \spanp}

To meet our third goal, we need to introduce one last counting complexity class. The class \spanp~\cite{kobler1989counting} is defined exactly as the class \spanl\ introduced in Section~\ref{sec:approx-countvals}, but considering \emph{polynomial-time} nondeterministic Turing machines with output, instead of \emph{logarithmic-space} nondeterministic Turing machines with output. It is straightforward to prove that $\shp \subseteq \spanp$. Besides, it is known that $\shp = \spanp$ if and only if $\np = \up$~\cite{kobler1989counting}.\footnote{Recall that \up\ is the class Unambiguous Polynomial-Time introduced in \cite{V76}, and that $L \in \up$ if and only if there exists a polynomial-time nondeterministic Turing Machine $M$ such that if $x \in L$, then $\acc_M(x) = 1$, and if $x \not\in L$, then $\acc_M(x) = 0$.} Therefore, it is widely believed that $\shp$ is properly included in $\spanp$.
The following easy observation can be seen as a first hint that $\spanp$ is a good alternative to describe the complexity of counting completions.

\begin{observation}
\label{obs:general-upper-bound}
If $q$ is a Boolean query such that $\mc(q)$ is in \ptime, then $\countcompls(q)$ is in $\spanp$.
\end{observation}

Notice that this result applies to all~\sjfbcqs\ and, more generally, to all FO Boolean queries. In fact, this results applies to even more expressive query languages such as Datalog~\cite{abiteboul1995foundations}. More surprisingly, in the following theorem we show that $\ucountcompls(q)$ can be $\spanp$-complete for an FO query $q$ and, in fact, already for the negation of an~\sjfbcq.

\begin{theorem}
\label{thm:general-spanp-complete}
\begin{sloppypar}
	There exists an \sjfbcq\ $q$ such that $\ucountcompls(\neg q)$ is \spanp-complete under polynomial-time parsimonious reductions. 
\end{sloppypar}
\end{theorem}

To prove this result, we will use the problem of counting the number of satisfying assignments of a 3-CNF formula that are distinct in the first~$k$ variables, that we denote by~$\ksat$. 
Formally, 
	the problem~$\ksat$ takes as input a 3-CNF formula~$F$ on variables~$\{x_1,\ldots,x_n\}$ and an integer
	$1 \leq k \leq n$, and outputs the number of assignments of the first~$k$ variables that can be extended
	to a satisfying assignment of~$F$. This problem is shown to be~\spanp-complete in~\cite{kobler1989counting}:

	\begin{proposition}[{\cite[Section 6]{kobler1989counting}}]
		$\ksat$ is \spanp\ complete (under polynomial-time parsimonious reductions).
	\end{proposition}

We are ready to prove Theorem~\ref{thm:general-spanp-complete}.

\begin{proof}[Proof of Theorem~\ref{thm:general-spanp-complete}]
Notice that we only need to show hardness for a fixed \sjfbcq\ $q$.
We reduce from $\ksat$ to $\ucountcompls(\neg q)$, for a fixed \sjfbcq\ $q$ to be defined.
	Let~$F$ be a 3-CNF on variables~$\{x_1,\ldots,x_n\}$, and~$1 \leq k \leq n$.
	We first explain how we build the incomplete database~$D$, and we will define the sjfBCQ~$q$ after.
	For every variable~$x_i$, $1 \leq i \leq n$, we have a null~$\bot_{x_i}$, and the (uniform) domain
	is~$\{0,1\}$.
	For~$(a,b,c) \in \{0,1\}^3$, we have a relation~$C_{abc}$ of arity~$3$, and we fill
	it with every tuple of the form~$C_{abc}(a',b',c')$ with~$(a',b',c') \in \{0,1\}^3$
	such that~$a=a' \lor b=b' \lor c=c'$ holds; hence for every~$(a,b,c) \in \{0,1\}^3$ there are exactly~$7$ facts of this form.
	For every clause~$K=l_1 \lor l_2 \lor l_3$ of~$F$ with~$l_1,l_2,l_3$ being
	literals over variables~$y_1,y_2,y_3$,
	letting~$(a_1,a_2,a_3) \in \{0,1\}^3$ be the unique tuple such that~$a_i=1$ iff~$l_i$ is
	a positive literal, we add to~$C_{a_1 a_2 a_3}$ the
	fact~$C_{a_1 a_2 a_3}(\bot_{y_1},\bot_{y_2},\bot_{y_3})$.
	Last, we have a binary relation~$S$ that we fill with the tuples~$S(i,\bot_{x_i})$
	for~$1 \leq i \leq k$.
	The sjfBCQ~$q$ then simply says that there exists a tuple that appears in all the relations~$C_{abc}$:
	\begin{eqnarray}\label{eq:q}
	q & = & \exists x \exists y\, S(x,y) \land \exists x \exists y \exists z\, \bigg(\bigwedge_{(a,b,c)\in \{0,1\}^3} C_{abc}(x,y,z)\bigg)
	\end{eqnarray}
	Note that we added the seemingly useless query $\exists x \exists y\, S(x,y)$ to~$q$ because the set of relations in $D$ has to be a subset of the set of relations occurring in $q$ (indeed, this is how we defined our problems in Section~\ref{sec:preliminaries}).
	We now show that the number of completions of~$D$ that do not satisfy~$q$ is equal to the number of assignments of the first~$k$ variables that can be extended
	to a satisfying assignment of~$F$, thus establishing that~$\ucountcompls(\lnot q)$ is \spanp-hard (under polynomial-time parsimonious reductions).
	First, observe that the assignments of the variables are in bijection with the valuations of the nulls of~$D$.
	One can then readily observe the following:
	 \begin{itemize}
		 \item If~$q$ is falsified in a completion of~$D$, it can only be because there does not exist a tuple that occurs in all the relations; this is because the query~$\exists x \exists y\, S(x,y)$ is always satisfied by any completion of~$D$.
		 \item For every assignment of the variables, letting~$\nu$ be the corresponding valuation of the nulls, there exists a tuple that is in all relations~$C_{abc}$ of~$\nu(D)$ if and only if that assignment is not satisfying for~$F$. Indeed, this happens if and only if there exists a relation~$C_{abc}$ such that~$\nu(D)(C_{abc})$ contains exactly~$8$ facts.
		 \item For every two valuations~$\nu,\nu'$ such that the corresponding assignments are not satisfying the formula, we have that~$\nu(D) \neq \nu'(D)$ if and only if~$\nu$ and~$\nu'$ differ on the first~$k$ variables. This is because, by the previous item, each relation~$C_{abc}$ contains exactly the~$7$ ground tuples that we initially put in~$D$.
	 \end{itemize}
	By putting it all together, we obtain that the reduction works as expected.
\end{proof}

This theorem 
gives evidence that $\spanp$ is the right class to describe the complexity of counting completions for FO queries (and even for queries with model checking in polynomial time). 
It is important to notice that \spanp-hardness is proved in Theorem~\ref{thm:general-spanp-complete} by considering parsimonious reductions. This is a delicate issue, because from the main result in \cite{TW92}, it is possible to conclude that every counting problem that is \shp-hard (even under polynomial-time parsimonious reductions) is also \spanp-hard under polynomial-time Turing reductions, so a more restrictive notion of reduction has to be used when proving that a counting problem is \spanp-hard~\cite{kobler1989counting}.

Before continuing, we prove Proposition \ref{prp:not-in-shp}.  

\begin{proof}[Proof of Proposition \ref{prp:not-in-shp}]
	Let $q$ be the \sjfbcq\ defined in Equation~\eqref{eq:q} in the proof of Theorem~\ref{thm:general-spanp-complete}.
	Its schema $\sigma = \{S\} \cup \{C_{abc} \mid (a,b,c) \in \{0,1\}^3\}$ consists of~$10$ relation symbols, with~$S$ being binary and each~$C_{abc}$ being ternary.
	Let us denote by~$\ucountcompls(\sigma)$ the problem that takes as input an incomplete database over schema~$\sigma$ and outputs its number of completions.
	The first part of our proof is to reduce~$\ucountcompls(\sigma)$ to~$\ucountcompls(q)$; formally, we claim that~$\ucountcompls(\sigma) \pr \ucountcompls(q)$.
		Indeed, let~$D$ be an incomplete database over schema~$\sigma$, that is an input of~$\ucountcompls(\sigma)$. 
		We construct in polynomial time an incomplete database~$D'$ over the same schema such that~$\ucountcompls(\sigma)(D) = \ucountcompls(q)(D')$, thus establishing the parsimonious reduction.	
		Let~$f$ be a fresh constant that does occurs neither in~$D$ nor in the domain of some null.
		Then the relation~$D'(S)$ is the same as the relation~$D(S)$, plus a fact~$S(f,f)$.
		Moreover, for every~$(a,b,c) \in \{0,1\}^3$, the relation~$D'(C_{abc})$ consists of all the facts in~$D(C_{abc})$, plus a fact~$C_{abc}(f,f,f)$.
		It is easy to see that~$D$ and~$D'$ have the same number of completions. Moreover, thanks to the facts that use the constant~$f$,
		we have that every completion of~$D'$ satisfies~$q$. Therefore, we indeed have that~$\ucountcompls(\sigma)(D) = \ucountcompls(q)(D')$.
		This proves that~$\ucountcompls(\sigma) \pr \ucountcompls(q)$. 

For the second part of the proof, we need to introduce the complexity class \gapp. This class consists of function problems that can be expressed as the difference of two functions in $\shp$~\cite{FFK94,G95}. It is known that if the inclusion $\spanp \subseteq \gapp$ holds, then we have that $\np \subseteq \spp$~\cite{MTV94}.\footnote{In fact, the class \gapspanp\ is defined in \cite{MTV94}, where it is proved that a function $f$ is in \gapspanp\ if and only if $f = g - h$, where $h,g$ are functions in \spanp. Then it is shown in \cite[Corollary 3.5]{MTV94} that the inclusion $\gapspanp \subseteq \gapp$ implies that $\np \subseteq \spp$. But if we have that $\spanp \subseteq \gapp$, then we also have that $\gapspanp \subseteq \gapp$ as $\gapp$ is closed under subtraction and, therefore, we conclude that $\np \subseteq \spp$ as desired.} With this, we are able to prove the proposition.
	Assume that~$\ucountcompls(q)$ is in $\shp$. Then, by the first part of the proof we have that~$\ucountcompls(\sigma) \in \shp$ as well (because~$\shp$ is closed under polynomial-time parsimonious reductions).
	Now, observe that for every incomplete database~$D$ over~$\sigma$, the following holds:
	\begin{eqnarray*}
	\ucountcompls(\lnot q)(D) & = & \ucountcompls(\sigma)(D) - \ucountcompls(q)(D).
	\end{eqnarray*}
	But then this means that~$\ucountcompls(\lnot q)$ is in \gapp\ (since both problems in the right hand side are in~$\shp$).
	Since~$\ucountcompls(\lnot q)$ is \spanp-complete by Theorem~\ref{thm:general-spanp-complete} under polynomial-time parsimonious reductions, and since \gapp\ is closed under polynomial-time parsimonious reductions, this would indeed imply that $\spanp \subseteq \gapp$ and, hence, that $\np \subseteq \spp$.
\end{proof}

We conclude this section by considering an even more general scenario where queries have model checking in \np. Interestingly, in this case \spanp\ is again the right class to describe the complexity not only of counting completions, but also of counting valuations.

\begin{theorem}\label{theo-np-spanp}
	If~$q$ is a Boolean query with~$\mc(q) \in \np$, then both $\countvals(q)$ and $\countcompls(q)$ are in \spanp.
	Moreover, there exists such a Boolean query~$q$ for which $\ucountvals(q)$ is~\spanp-complete under polynomial-time parsimonious~reductions $($and for~$\ucountcompls(q)$, we can even take~$q$ to be the negation of an \sjfbcq, hence with model checking in~\ptime, as given by Theorem~\ref{thm:general-spanp-complete}$)$.
\end{theorem}

\begin{proof}
	It is straightforward to prove that these problems are in \spanp. The part in between parenthesis has been shown in theorem~\ref{thm:general-spanp-complete}. Thus, we need to prove that $\ucountvals(q)$ is \spanp-hard for a fixed Boolean query $q$ such that $\mc(q) \in \np$, under polynomial-time parsimonious reductions. To do this, we will reduce from the \spanp-complete problem $\hamsubgraphs$, defined as follows. 
	
		Let~$G=(V,E)$ be a undirected graph, and let~$S \subseteq V$.
		The \emph{subgraph of~$G$ induced by~$S$}, denoted by~$G[S]$, is the graph with set of
		nodes~$S$ and set of edges $\{\{u,v\} \in E \mid u,v\in S\}$,
		We recall that a graph~$G$ is \emph{Hamiltonian} when there exists a cycle in~$G$ that visits every node of~$G$ exactly once.
		The problem~$\hamsubgraphs$ takes as input a simple graph~$G=(V,E)$ and an integer~$k$, and outputs the number of induced subgraphs~$G[S]$ with~$|S| = k$
		such that~$G[S]$ is Hamiltonian.
	
	\begin{proposition}[{\cite[Section 6]{kobler1989counting}}]
	\label{prp:hamsubgraphs-hard}
	$\hamsubgraphs$ is \spanp-complete (under polynomial-time parsimonious reductions).
	\end{proposition}

Next we show that $\hamsubgraphs \pr \ucountvals(q)$, for a fixed Boolean query $q$ (to be defined). 
Let~$G=(V,E)$ be an undirected graph. We first explain how we construct the incomplete database~$D$, and we will then define the query~$q$.
		The schema contains two binary relation symbols~$R,T$ and one unary relation symbol~$K$. Fix a linear order~$a_1,\ldots,a_n$ of the nodes of~$G$.
		For every edge~$\{u,v\} \in E$ we have the facts~$R(u,v)$ and~$R(v,u)$.
		For~$1 \leq i \leq n$ we have a fact~$T(a_i,\bot_i)$, and the domain of the nulls is~$\{0,1\}$.
		For~$1 \leq j \leq k$ we have a fact~$K(j)$. Observe that~$D$ is a Codd table.
		We now define the Boolean query~$q$, which will be a sentence in existential second-order logic ($\exists$SO) over relational signature~$R,T,K$.
		Before doing so, we explain the main idea: intuitively,~$q$ will check that there are exactly~$k$ facts of the form~$T(a_i,1)$ in the relation~$T$ and that, letting~$S$ be the set of nodes~$v$ such that~$T(v,1)$ is in relation~$T$, the induced subgraph~$G[S]$ is Hamiltonian. This will indeed ensure
		that we have $\ucountvals(q)(D) = \hamsubgraphs(G,k)$, thus completing this reduction, which is parsimonious and can be performed in polynomial-time.
		The query is
		\[q \ = \ \exists S\, \psi_1(S) \land \psi_2(S)\]
		where~$S$ is a unary second order variable and the formula $\psi_1(S)$ states that (a) the elements~$s$ of~$S$ are exactly all the elements such that~$T(s,1)$ holds,
		and that (b) there are exactly the same number of elements in~$S$ as there are elements~$j$ for which~$K(j)$ holds.
		It is clear that (a) can be expressed in FO.
		Moreover, (b) can be expressed in $\exists$SO by asserting the existence of a binary second-order relation~$U$
		that represents a bijective function from~$S$ to the elements in~$K$.
		Then~$\psi_2(S)$ is a formula that asserts that~$G[S]$ is Hamiltonian. Since this is a property in NP,~$\psi_2(S)$ can be expressed in~$\exists$SO by Fagin's theorem (see, e.g.,~\cite{immerman2012descriptive}).
		This shows that the reduction is correct.
		Finally, the fact that~$\mc(q)$ is in NP again follows from Fagin's theorem.
		This concludes the proof.
\end{proof}

\section{Extensions to queries with constants and free variables}
\label{sec:extensions}
So far, we have only considered our counting problems for queries that are
Boolean and that do not contain constants.  In this section we explain how our
framework can be adapted to queries with constants and with free variables.  
Specifically, we will explain how one can obtain dichotomies for
self-join--free conjunctive queries with constants and
free variables.

Before that, we have to formally define our counting problems for a query with
free variables.  Let~$q(\bar{x})$ be a query with free variables~$\bar
x$. For a tuple of constants~$\bar t$ of appropriate arity, we write~$q(\bar
t)$ the Boolean query obtained by substituting the variables~$\bar x$ with the
constants~$\bar t$. The problem~$\countvals(q(\bar x))$ then takes as input an
incomplete database~$D$ over relations~$\sig(q(\bar x))$, a tuple of
constants~$\bar t$, and returns the number of
valuations~$\nu$ of~$D$ such that~$\nu(D) \models q(\bar t)$. We write this
output~$\countvals(q(\bar x))(D,\bar{t})$. The problem~$\countcompls(q(\bar x))$ 
is defined similarly.

We first explain in Section~\ref{subsec:freevars} how to obtain dichotomies
for self-join-free conjunctive
query with free variables and constants, assuming we have dichotomies
for self-join--free Boolean conjunctive queries
with constants.  In section~\ref{subsec:constants}, we then explain how to
obtain dichotomies for the later case.

\subsection{Dealing with free variables}
\label{subsec:freevars}

Suppose in this section that we have a dichotomy between \shp-hardness and \fp\
of our counting problems, for every \sjfbcq\ that is allowed to contain constants.
We then show how to obtain dichotomies, for every
self-join--free CQ that is allowed to have constants \emph{and} free variables.
To this end, we will need the following definition. Let~$q(\bar x)$ be a
self-join-free CQ with free variables~$\bar x$, and let~$\bar t$ and~$\bar{t'}$
be two tuples of constants with appropriate arity. We say that~$\bar t$
and~$\bar t'$ are \emph{equivalent with respect to~$q(\bar x)$} if one can go
from~$q(\bar t)$ to~$q(\bar{t'})$ by iteratively renaming constants into fresh
constants. For instance, if~$q(x_1,x_2)$ is~$\exists y\, R(y,c,x_1,x_2) \land
S(c',x)$, then~$(c,c'')$ is equivalent to~$(c,c''')$, and~$(c'',c''')$
and~$(c''',c'')$ are also equivalent. It is clear that, if~$q(\bar x)$ is
fixed, this defines an equivalence relation that has only finitely many classes.
Furthermore, it is also clear that if~$\bar t$ and~$\bar{t'}$ are equivalent
with respect to~$q(\bar x)$ then the problems~$\countvals(q(\bar t))$
and~$\countvals(q(\bar{t'}))$ (resp.,~$\countcompls(q(\bar t))$
and~$\countcompls(q(\bar{t'}))$) have the same complexity. We can now show what we wanted.
In what follows, hardness refers to~\shp-hardness (but it is not important for the proof).

\begin{lemma}
\label{lem:freevars}
Assume that the following is true: for every \sjfbcq~$q$ that is allowed to
have constants, the problem~$\countvals(q)$ is either hard or is tractable.
Then the following is also true: for every self-join--free conjunctive query~$q(\bar x)$, the
problem~$\countvals(q(\bar x))$ is either hard or is tractable.
This holds also for counting completions, and when restricted to Codd tables and/or to the uniform setting.
\end{lemma}
\begin{proof}
We only deal with counting valuations for the naive and non-uniform case, as the other cases are similar.
We prove the following for any self-join--free CQ~$q(\bar x)$, which implies
the claim: if there exists a tuple of constants~$\bar t$ such
that~$\countvals(q(\bar t))$ is hard, then~$\countvals(q(\bar x))$ is hard as
well, otherwise~$\countvals(q(\bar x))$ is tractable.  We start with the “if”
direction. Let~$\bar t$ be a tuple of constants such that~$\countvals(q(\bar
t))$ is hard.  By definition, it is clear that, for any incomplete
database~$D$, we have that~$\countvals(q(\bar{t}))(D) = \countvals(q(\bar
x))(D,\bar{t})$; this shows hardness of~$\countvals(q(\bar x))$.  Now for the
“otherwise” direction.  Let~$\bar{t_1},\ldots,\bar{t_k}$ be representatives of
the finitely many equivalence classes of the equivalence relation defined
above.  We have access to oracles for~$\countvals(q(\bar{t_1})), \ldots,
\countvals(q(\bar{t_k}))$.  Let~$D,\bar{t}$ be an input of~$\countvals(q(\bar
x))$. We then simply recognize (in constant time since the query if fixed) to
which~$\bar{t_i}$ the tuple~$\bar t$ is equivalent with respect to~$q(\bar x)$,
and call the appropriate oracle.
\end{proof}

(Notice that this idea actually works for conjunctive queries (with self-joins),
or even unions of conjunctive queries.)  Hence, the problem becomes that of
obtaining dichotomies for self-join--free Boolean conjunctive queries that can
contain constants.  We explain in the next section how this can be done.

\subsection{Dealing with constants}
\label{subsec:constants}

In this section, we simply write “an \sjfbcq” to mean an \sjfbcq\ that can
contain some constants.  To the best of our knowledge, there is no general
reduction that allows us to easily obtain dichotomies for the case with
constants from the case where queries do not have constants.  Hence, the
strategy to obtain dichotomies will be the same as in
Section~\ref{sec:countvals-sjfcqs}; namely, we will again use the notion of
pattern to find a set of hard patterns, and then show that when a query does
not have any of the hard patterns then the problem is tractable.  The notion of
pattern for an \sjfbcq\ that can contain constants is the same as the one we used
in Definition~\ref{def:pattern}, but we simply add the possibility of
deleting an occurrence of a constant. For instance the query~$R(x,c)$ is a
pattern of the query~$R(x,x,c,c,c')$ (in this section all variables are
existentially quantified, since we consider only Boolean queries).
The main property of patterns that we used then extends in this setting, as shown next.
\begin{lemma}
\label{lem:pattern-parsimonious-constants}
	Let~$q,q'$ be \sjfbcqs\ such that~$q'$ is a pattern of~$q$.  Then we
have~$\countvals(q') \pr \countvals(q)$.  Moreover, the same results hold if we
restrict to Codd tables, and/or to the uniform setting, and/or to counting
completions.
\end{lemma}

\begin{proof}
The reduction is exactly the same as that of Lemma~\ref{lem:pattern-parsimonious}, the
only difference being that, when we delete an occurrence of a constant~$c$, we
simply fill the corresponding columns of every tuple with this constant.
\end{proof}

Next, we prove dichotomies for counting valuations for the non-uniform setting,
for naive and Codd databases. For Codd databases, it turns out that there is no
new hard pattern that involves constants, so the dichotomy is the same as for
\sjfbcqs\ without constants. Formally:

\begin{theorem}[dichotomy]
\label{thm:countvals-sjfcqs-codd-constants}
Let~$q$ be an \sjfbcq.  If $R(x) \land S(x)$ is a pattern of $q$, then
$\ccountvals(q)$ is \shp-complete. Otherwise, $\ccountvals(q)$ is in \fp.
\end{theorem}
\begin{proof}
The proof is exactly as the proof of Theorem~\ref{thm:countvals-sjfcqs-codd}, with
the following modification: we let~$\rho({\bar t_j})$ be the number of
valuations of the nulls appearing in~${\bar t_j}$ that do not match the
corresponding atom of~$q$; and clearly, this can again be computed in polynomial time.
\end{proof}

For naive tables however, we find two new hard patterns that involve constants,
as shown next.

\begin{proposition}
\label{prp:Rcc-and-Rcc'-hard}
Let~$c,c'$ be two distinct constants.
The problems~$\countvals(R(c,c))$ and~$\countvals(R(c,c'))$ are both~\shp-hard.
\end{proposition}
\begin{proof}
We only explain for~$\countvals(R(c,c'))$, as the other case is analogous.  The
reduction is similar to that used in Proposition~\ref{prp:RxSxyTy-RxySxy-hard},
but we reduce from counting the number of independent sets in \emph{bipartite}
graphs (in Proposition~\ref{prp:RxSxyTy-RxySxy-hard} we did not need the graphs
to be bipartite). Let~$G=(U\sqcup V,E)$ be a bipartite graph. We have one
null~$\bot_u$ for every node~$u \in U$ with domain~$\dom(\bot_u)=\{0,c\}$, and
one null~$\bot_v$ for every node~$u \in V$ with domain~$\dom(\bot_v)=\{0,c'\}$.
For every edge~$(u,v)\in E$ we have a fact~$R(\bot_u,\bot_v)$ in~$D$.  Then it
is clear that the number of valuations of~$D$ that do not contain~$R(c,c')$ is
equal to the number of independent sets of~$G$, thus establishing hardness.  (Note
that for~$R(c,c)$, we do not need the graph to be bipartite for the reduction
to work.)
\end{proof}

We then claim that these are the only additional patterns that are necessary
to obtain a dichotomy for~$\countvals(q)$ for \sjfbcqs\ with constants.

\begin{theorem}[dichotomy]
\label{thm:countvals-sjfcqs-constants}
Let~$q$ be an \sjfbcq.  If $R(x,x)$ or $R(x) \land S(x)$ or~$R(c,c)$
or~$R(c,c')$ for~$c\neq c'$ is a pattern of $q$, then $\countvals(q)$ is
\shp-complete. Otherwise,
$\countvals(q)$ is in \fp.
\end{theorem}

The \shp-hardness part
of the claim follows from Lemma~\ref{lem:pattern-parsimonious-constants} and
Propositions~\ref{prp:Rxx-hard} and~\ref{prp:RxSx-hard}
and~\ref{prp:Rcc-and-Rcc'-hard}. We now show the tractability claim.  First,
observe that not having any of these patterns means the following: every
variable in~$q$ has exactly one occurrence and every atom of~$q$ contains at
most one constant (but notice that the same constant can appear in multiple
atoms). But then, because the database is Codd and because~$q$ has no
self-joins, by multiplying by the appropriate factor, the
problem~$\countvals(q)$ reduces to the problem~$\countvals(q')$ where~$q'$ is an
\sjfbcq\ of the form~$R_1(c_1)\land R_2(c_2)\land \ldots \land R_k(c_k)$, where
the constants~$c_1,\ldots c_k$ are not necessarily distinct.  We give next an
example proof that the problem is tractable for such a simple query.

\begin{example}
\em{
Let~$q$ be the query~$R(c)\land S(c')$ with~$c\neq c'$, and~$D$ be an
incomplete naive database.  We explain how to compute~$\countvals(\lnot q)(D)$
(the number of valuations that do not satisfy the query) in \fp. This is
enough, since the total number of valuations can clearly be computed in \fp.
First of all, we can assume without loss of generality that~$D$ does not
contain ground atoms that already satisfy the query.  
Then, let~$B^{RS} \defeq \{\bot^{RS}_1,\ldots,\bot^{RS}_{n_{RS}}\}$ be the set
of nulls that occur in both~$D(R)$ and~$D(S)$, $B^R \defeq
\{\bot^{R}_1,\ldots,\bot^{R}_{n_R}\}$ be the set of nulls that occur in~$D(R)$
but not in~$D(S)$, and $B_S \defeq \{\bot^{S}_1,\ldots,\bot^{S}_{n_S}\}$ be the
set of nulls that occur in~$D(S)$ but not in~$D(R)$. Let~$D'$ be the database
that contains only the facts~$R(\bot)$ and~$S(\bot)$ for~$\bot \in B^{RS}$.
Notice that the number of valuations of the nulls in~$B^R$ such that no null
has value~$c$ is~$\prod_{i=1}^{n_R} |\dom(\bot^R_i)\setminus \{c\}|$, and that
the number of valuations of the nulls in~$B^R$ such that some null has
value~$c$ is then~$\prod_{i=1}^{n_R} |\dom(\bot^R_i)| - \prod_{i=1}^{n_R}
|\dom(\bot^R_i)\setminus \{c\}|$, and a similar expression can be obtained for
the valuations of the nulls in~$B^S$ and constant~$c'$.  Then, by case analysis
of whether some null in~$B^R$ has value~$c$ or not, and whether some null
in~$B^S$ has value~$c'$ or not, we obtain that~$\countvals(\lnot q)(D) = A + B
+ C$, where

\[ A = \big( \prod_{i=1}^{n_R} |\dom(\bot^R_i)| - \prod_{i=1}^{n_R} |\dom(\bot^R_i)\setminus \{c\}|\big) \times \prod_{i=1}^{n_S} |\dom{\bot^S_i}|\setminus \{c'\}| \times \prod_{i=1}^{n_{RS}} |\dom{\bot^{RS}_i}\setminus \{c'\}|\]

is the number of valuations of~$D$ that do not satisfy~$q$ and such that some
null in~$B^R$ has value~$c$,

\[B = \prod_{i=1}^{n_R} |\dom(\bot^R_i)|\setminus \{c\}| \times \big( \prod_{i=1}^{n_S} |\dom(\bot^S_i)| - \prod_{i=1}^{n_S} |\dom(\bot^S_i)\setminus \{c'\}| \big)\times \prod_{i=1}^{n_{RS}} |\dom{\bot^{RS}_i}\setminus \{c\})\]

is the number of valuations of~$D$ that do not satisfy~$q$ and such that no null
in~$B^R$ has value~$c$ and some null in~$B^S$ has value~$c'$, and

\[C = \prod_{i=1}^{n_R} |\dom(\bot^R_i)\setminus \{c\}| \times \prod_{i=1}^{n_S} |\dom(\bot^S_i)\setminus \{c'\}| \times \countvals(\lnot q)(D')\]

is the number of valuations of~$D$ that do not satisfy~$q$ and such that no null
in~$B^R$ has value~$c$ and no null in~$B^S$ has value~$c'$.  Hence, we only
have to explain how to compute~$\countvals(\lnot q)(D')$ in polynomial time.
For~$k \in \{0,\ldots,n_{RS}\}$, let~$D'_k$ be the database containing only the
facts of~$D'$ over the nulls~$\bot^{RS}_1,\ldots,\bot^{RS}_k$.  We then define
the quantities~$V(S,k)$, for~$S\in \{\{c\},\{c'\},\{c,c'\},\emptyset\}$ and $k
\in \{0,\ldots,n_{RS}\}$ to be the number of valuations of~$D_k$ that do not
satisfy~$q$ and such that no null has a value that is in~$S$.  Notice then
that~$V(\emptyset,n_{RS}) = \countvals(\lnot q)(D')$.  For an element~$a$ and
set~$A$, we write~$[a \in A]$ to mean~$1$ if~$a \in A$ and~$0$ otherwise.  But
then we can easily compute the quantities~$V(S,k)$ by dynamic programming using
the following relations:

\begin{align*}
V(\emptyset,k) =& [c \in \dom(\bot^{RS}_k)] \times V(\{c\},k-1)\\
                & + [c' \in \dom(\bot^{RS}_k)] \times V(\{c'\},k-1)\\
		& + |\dom(\bot^{RS}_k) \setminus \{c,c'\}| \times V(\emptyset,k-1)
\end{align*}

and

\begin{align*}
V(\{c,c'\},k) = |\dom(\bot^{RS}_k) \setminus \{c,c'\}| \times V(\{c,c'\},k-1)
\end{align*}

and

\begin{align*}
V(\{c\},k) =& [c' \in \dom(\bot^{RS}_k)] \times V(\{c,c'\},k-1)\\
		& + |\dom(\bot^{RS}_k) \setminus \{c,c'\}| \times V(\{c,c'\},k-1)
\end{align*}

and a similar relation for~$V(\{c'\},k)$.
}\qed
\end{example}

The proof for the general case is simply a tedious generalization of the
proof for this example, and is not very interesting, so we omit it.

Therefore, by combining Theorems~\ref{thm:countvals-sjfcqs-codd-constants}
and~\ref{thm:countvals-sjfcqs-constants} with Lemma~\ref{lem:freevars}, we
obtain dichotomies for counting valuations of self-join--free conjunctive
queries that can contain constants and free variable, for the non-uniform
setting for both naive and Codd tables.  It is likely that dichotomies can be
obtained for counting valuations in the uniform setting by using the same methodology, but we do not
pursue this further as our goal in this section is not to be exhaustive but rather to give the main ideas
to be able to handle free variables and constants.

For counting completions, the reader can check that the pattern~$R(c)$ is hard
for~$\ccountcompls$ (so that in this case again all queries are hard), and that
for~$\cucountcompls$, any query that contains an atom that is not unary is
again hard (it is evident from the proof of
Proposition~\ref{prp:Rxx-Rxy-hard-compls-codd}). By combining these observations
with Lemma~\ref{lem:freevars} and with the tractability proof
for~$\ucountcompls$ (which can be shown to extend in this case), we obtain four
dichotomies for counting completions for self-join--free conjunctive queries
that can contain constants and free variables.

Last, we point out that by using the same methodology, one also can extend our
results on approximations to the case of queries containing constants and free
variables (since the relevant reductions are parsimonious).

\section{Related Work}
\label{sec:related}
There are two main lines of work that must be compared to what we do in this article. In both cases the goal is to go beyond the traditional notion of \emph{certain answers} that so far had been used almost exclusively to deal with query answering over uncertain data.
We discuss them here, explain how they relate to our problems and what are the fundamental differences.

\paragraph{{\bf Best answers and 0-1 laws for incomplete databases}}
Libkin has recently introduced a framework that can be used to measure the certainty with which a Boolean query holds on an incomplete database, and also to compare query answers (for a non-Boolean query)~\cite{libkin2018certain}.  
For a Boolean query~$q$, incomplete database~$D$, and integer~$k$, he defines the quantity~$\mu^k(q,D)$ as~$\frac{|\mathrm{Supp}^k(q,D)|}{|V^k(D)|}$, where~$V^k(D)$ denotes the set of valuations of~$D$ with domain~$\{1,\ldots,k\}$, and $\mathrm{Supp}^k(q,D)$ denotes the set of valuations~$\nu \in V^k(D)$ such that $\nu(D) \models q$; hence,~$\mu^k(q,D)$ represents the relative frequency of valuations~$\nu$ in~$\{1,\ldots,k\}$ for which the query is satisfied. He then shows that, for a very large class of queries (namely, \emph{generic queries}), the value~$\mu^k(q,d)$ always tends to~$0$ or~$1$ as~$k$ tends to infinity (and the same results holds when considering completions instead of valuations). This means that, intuitively,  over an infinite domain the query~$q$ is either almost certainly true or almost certainly false. 

He also studies the complexity of finding best answers for a non Boolean query~$q$. As mentioned in the introduction, 
a tuple~$\overline{a}$ is a better answer than another tuple~$\overline{b}$ when for every valuation~$\nu$ of~$D$, if we have~$\overline{b} \in q(\nu(d))$ then we also have~$\overline{a} \in q(\nu(d))$. A best answer is then an answer such that there is no other answer strictly better than it (under inclusion of the sets of satisfying valuations). He studies the complexity of comparing answers under this semantics, and that of computing the set of best answers (see also \cite{GS19}). 

There are several crucial differences between this previous work and ours. First, Libkin does not study the complexity of computing~$\mu^k(q,d)$. We do this under the name~$\ucountvals(q)$; moreover, we also study the setting in which the domains are not uniform.
Second, knowing that a tuple is the best answer might not tell us anything about the size of its ``support'', i.e., the number of valuations that support it.  In particular, a best answer is not necessarily an answer which has the biggest support. 
Finally, under the semantics of better answers it does not matter if we look at the completions or at the valuations (i.e., a tuple is a best answer with respect to inclusion of valuations iff it is the best answer with respect to completions); while we have shown that it does matter for counting problems.

\paragraph{{\bf Counting problems for probabilistic databases and consistent query answering.}}
Remarkably, 
counting problems have received considerable attention in other database 
 scenarios where uncertainty issues appear.
 As mentioned in the introduction, this includes the settings of probabilistic databases and inconsistent databases. 
 In the former case, uncertainty is represented as a probability distribution on the possible states of the data~\cite{suciu2011probabilistic,dalvi2013dichotomy}.
 There, query answering 
 amounts to computing 
 a weighted sum of the probabilities of the possible states of the data that satisfy a query~$q$. We call this problem {\sf Prob}$(q)$.  
 In the case of 
inconsistent databases, 
 we are given a set~$\Sigma$ of constraints and a 
database~$D$ that does not necessarily satisfy~$\Sigma$; cf.~\cite{ABC99,2011Bertossi,Bertossi19}. Then the task is to reason about the set of all {\em repairs} of~$D$ with respect to~$\Sigma$ \cite{ABC99}.
In our context, this means that one wants to count the number of repairs of~$D$ with respect to~$\Sigma$ that satisfy a given query~$q$. 
When~$q$ and~$\Sigma$ are fixed, we call this problem~$\#${\sf Repairs}$(q,\Sigma)$. 
 
 Both {\sf Prob}$(q)$ and~$\#${\sf Repairs}$(q,\Sigma)$
have been intensively studied already.
To start with, counting complexity dichotomies have been obtained for the problem $\#${\sf Repairs}$(q,\Sigma)$; e.g., \cite{maslowski2013dichotomy} gives a dichotomy for this problem when~$q$ is an \sjfbcq\ and~$\sigma$ consists of primary keys, and~\cite{maslowski2014counting} extends this result to CQs with self-joins but only for unary keys constraints. We also mention~\cite{calautti2019counting}, where the problem of counting repairs such
that a particular input tuple is in the result of the query on the repair is studied. 
A seemingly close counting problem for probabilistic databases is the problem~{\sf Prob}$(q)$ over \emph{block independent disjoint} (BIDs) databases. We do not define it formally here, but counting repairs under primary keys can be seen as a special case of this problem,
where the tuples in a “block” all have the same probability, and where the sum of the probabilities sum to $1$ (and in BIDs this sum is allowed to be $<1$, meaning that a block
can be completely erased). Dichotomies for this problem have been obtained in~\cite{dalvi2011queries} for \sjfbcqs.
Counting complexity dichotomies for other models of probabilistic databases also exist; e.g., for {\em tuple-independent} 
probabilistic databases in which each fact is assigned an independent probability of being part of the actual dataset. 
Interestingly, dichotomies in this case hold for arbitrary unions of BCQs, and thus not just for \sjfbcqs~\cite{dalvi2013dichotomy}. 

In some cases, one can use a problem of the form~$\#${\sf Repairs}$(q,\Sigma)$ (or {\sf Prob}$(q)$) to show the hardness of a problem of the form~$\countvals(q')$. 
For instance, in 
Section \ref{subsec:countvals-non-uniform} 
 we used the \shp-hardness of $\#${\sf Repairs}$(R'(\underbar{y},x) \land S'(\underbar{z},x))$ to prove that of $\ccountvals(R(x) \land S(x))$.
In general however, the problems $\#${\sf Repairs}$(q,\Sigma)$ and~{\sf Prob}$(q)$ seem to be unrelated to our problems, for the following reasons.
		First, in our setting the nulls can appear anywhere, so there is no notion of primary keys here; hence it seems unlikely that one can design a generic reduction from the problem of counting valuations/completions to the problem of counting repairs.
In fact, it would perfectly make sense to study our counting problems where we add constraints such as functional dependencies.
		Second, in the BID and counting repairs problems, each “valuation” (repair) gives a different complete
database, while in our case we have seen that this is not necessarily the case.
In particular, problems of the form $\countcompls(q)$ have no analogues in these settings, whereas we have seen that they behave very differently in our setting.

Concerning approximation results, it is known that the problems $\#${\sf Repairs}$(q,\Sigma)$ and {\sf Prob}$(q)$
admit an FPRAS in some important settings. In particular, when~$q$ is a union of BCQs, this 
holds for $\#${\sf Repairs}$(q,\Sigma)$ when $\Sigma$ is a set of primary keys \cite{calautti2019counting}, and 
for {\sf Prob}$(q)$ over BID and tuple-independent probabilistic databases~\cite{dalvi2011queries}.
We observe here that this is reminescent of our Corollary~\ref{cor:countvals-has-fpras}, which shows that problems of the form~$\countvals(q)$ have an FPRAS for every union of BCQs.

\section{Final Remarks}
\label{sec:conclusion}
Our work aims to be a first step in the study of counting problems over
incomplete databases.  The main conclusion behind our results is that the
counting problems studied in this article are particularly hard from a
computational point of view, especially when compared to more positive results
obtained in other uncertainty scenarios; e.g., over probabilistic and
inconsistent databases.  As we have shown, a particularly difficult problem in
our context is that of counting completions, even in the uniform setting where
all nulls have the same domain.  In fact,
Proposition~\ref{prp:Rxx-Rxy-hard-compls-codd} shows that this problem is
\shp-hard even in very restricted scenarios, and
Proposition~\ref{prp:ucountcompls-Rxy-no-frpas} that it cannot be approximated
by an FPRAS. It seems then that the only way in which one could try to tackle
this problem is by developing suitable tractable heuristics, without provable
quantitative guarantees, but that work sufficiently well in practical
scenarios. An example of this could be developing algorithms that compute
``under-approximations'' for the number of completions of a naive table
satisfying a certain \sjfbcq\ $q$. Notice that a related approach has been
proposed by Console et al. for constructing under-approximations of the set of
certain answers by applying methods based on many-valued logics \cite{CGL16}. 

We plan to continue working on several interesting problems that are left open
in this article. First of all, we would like to pinpoint the  complexity of
$\countcompls(q)$ when $q$ is an \sjfbcq; in particular, whether this problem
is~\spanp-complete for at least one such a query. We also want to study whether
the non-existence of FPRAS for $\ucountcompls(q)$ established in
Proposition~\ref{prp:ucountcompls-Rxy-no-frpas} continues to hold over Codd
tables.  We would also like to develop a more thorough understanding of the
role of fixed domains in our dichotomies.  In several cases, that we have
explicitly stated, our lower bounds hold even if nulls in tables are
interpreted over a fixed domain.  Still, in some cases we do not know whether
this holds. These include, e.g., Proposition~\ref{prp:RxSxyTy-hard-codd},
Proposition~\ref{prp:countcompls}, and
Proposition~\ref{prp:Rxx-Rxy-hard-compls-codd}.  Finally, it would also be
interesting to study these counting problems under bag semantics (instead of
the set semantics used in this article), or
consider arbitrary conjunctive queries as opposed to only self-join--free ones.

\begin{acks}
  We thank Antoine Amarilli for suggesting to use \#BIS in the proof of Proposition~\ref{prp:RxSxyTy-hard-codd}, as well as the anonymous reviewers for their careful proofreading. This work was partially funded by ANID - Millennium Science Initiative Program - Code ICN17\_002. Arenas is funded by Fondecyt grant 1191337 and 
Barcel\'o by Fondecyt grant 1200967. 
\end{acks}

\bibliographystyle{ACM-Reference-Format}
\bibliography{main}

\appendix

\end{document}